\documentclass[12pt]{iopart}
\usepackage[utf8]{inputenc}
\usepackage[T1]{fontenc}
\usepackage[english]{babel}
\usepackage[pdftex]{graphicx}
\usepackage{bm,bbm,amsthm,amsfonts,amssymb,amsopn,amsmath}
\usepackage{iopams,braket}
\usepackage{dsfont,mathtools,physics}
\usepackage{url,cite}
\usepackage{doi}
\usepackage{etoolbox}
\usepackage{enumitem}
\usepackage[makeroom]{cancel}
\usepackage{scalerel,color,pgf,tikz}
\usepackage{calc,xstring,ifthen}
\usepackage{pdftexcmds,adjustbox}

\usepackage[section]{placeins}

\renewcommand{\vec}[1]{{\bm{\mathrm{#1}}}}
\newcommand{\ul}[1]{{\underline{#1}}}
\newcommand{\W}{{W}^{\phantom{\prime}}}
\newcommand{\V}{{W}^{\prime}}
\newcommand{\WV}{{W}^{(\prime)}}
\newcommand{\vW}{\mathbf{W}^{\phantom{\prime}}}
\newcommand{\vV}{\mathbf{W}^{\prime}}
\newcommand{\A}{{A}^{\phantom{\prime}}}
\newcommand{\B}{{A}^{\prime}}
\newcommand{\AB}{{A}^{(\prime)}}
\newcommand{\vA}{\mathbf{A}^{\phantom{\prime}}}
\newcommand{\vB}{\mathbf{A}^{\prime}}
\newcommand{\Uo}{\mathbb{U}_{\mathrm{o}}}
\newcommand{\U}{\mathbb{U}}
\newcommand{\Ue}{{\mathbb U}_{\mathrm{e}}}
\newcommand{\Ueo}{{\mathbb U}_{\mathrm{e/o}}}
\newcommand{\tUo}{\hat{\mathbb U}_{\mathrm{o}}}
\newcommand{\tUe}{\hat{\mathbb U}_{\mathrm{e}}}
\newcommand{\tUeo}{\hat{\mathbb U}_{\mathrm{e/o}}}
\newcommand{\tdUo}{\tilde{\mathbb U}_{\mathrm{o}}}
\newcommand{\tdUe}{\tilde{\mathbb U}_{\mathrm{e}}}
\newcommand{\lrV}{X^{\mathrm{L/R}}}
\newcommand{\lrW}{Y^{\mathrm{L/R}}}
\newcommand{\lV}{X^{\mathrm{L}}}
\newcommand{\lW}{Y^{\mathrm{L}}}
\newcommand{\rV}{X^{\mathrm{R}}}
\newcommand{\rW}{Y^{\mathrm{R}}}
\newcommand{\lL}{L^{\mathrm{L}}}
\newcommand{\rL}{L^{\mathrm{R}}}
\newcommand{\lR}{R^{\mathrm{L}}}
\newcommand{\rR}{R^{\mathrm{R}}}
\newcommand{\lspan}{\mathrm{span}}
\newcommand{\End}{\mathrm{End}}
\newcommand{\ave}[1]{{\langle #1\rangle}}
\newcommand\scalemath[2]{\scalebox{#1}{\mbox{\ensuremath{\displaystyle #2}}}}
\newcommand{\be}{\begin{eqnarray}}
\newcommand{\ee}{\end{eqnarray}}
\newcommand{\un}[1]{\underline{#1}}
\newcommand{\ii}{ {\rm i} }
\renewcommand{\hm}[2]{{\hat{#1}^{(#2)}}}
\newcommand{\hmp}[2]{{\hat{#1}^{'(#2)}}}
\newcommand{\hmu}[2]{{\hat{\mathbf{#1}}_{#2}}}
\newcommand{\hmup}[2]{{\hat{\mathbf{#1}}'_{#2}}}
\newcommand{\mmu}[2]{{\mathbf{#1}_{#2}}}
\newcommand{\mmup}[2]{{\mathbf{#1}'_{#2}}}

\def\tr{{\,{\rm tr}\,}}

\def\one{\mathbbm{1}}

\def\cf{\sigma}
\def\cA{{\mathfrak a}}    
\def\cB{{\mathfrak b}}    
\def\cC{{\mathfrak c}} 
\def\cE{{\vartheta}}

\def\WW{{{\hat{\mathbf{W}}}}}

\def\cO{{\mathcal O}}

\newcommand{\CC}{\mathbb{C}}

\definecolor{half}{rgb}{0.95,0.95,0.95}
\definecolor{full}{rgb}{0,0,0}
\definecolor{halfborder}{rgb}{0.8,0.8,0.8}
\definecolor{border}{rgb}{0.3,0.3,0.3}
\definecolor{colU}{rgb}{0.71,0.8,0.76}
\definecolor{colP}{rgb}{0.69,0.50,0.86}
\definecolor{colLines}{rgb}{0.31,0.31,0.31}
\definecolor{colObs}{rgb}{1,0.39,0.28}

\usetikzlibrary{decorations.pathmorphing,decorations.markings,patterns,decorations.pathreplacing,shapes.misc}

\newcommand\textrectangle[3]{
  \draw[border,fill=half] ({(#1)},{(#2-1)})  -- ({(#1+1)},{(#2)})  -- ({(#1)},{(#2+1)})  
  -- ({(#1-1)},{(#2)})  -- cycle;
  \node at ({(#1)},{(#2)}) {\scalebox{0.7}{#3}}
}

\newcommand\halfdashedrectangle[2]{
  \draw[halfborder,thick] ({(#1)},{(#2-1)})  -- ({(#1+1)},{(#2)})  -- ({(#1)},{(#2+1)})  
  -- ({(#1-1)},{(#2)})  -- cycle;
  \draw[border,thin,dashed,fill=halfborder]
  ({(#1)},{(#2-1)})  -- ({(#1+1)},{(#2)})  -- ({(#1)},{(#2+1)}) -- ({(#1-1)},{(#2)})  -- cycle;
}

\newcommand\redemptyrectangle[2]{
  \draw[red,very thick] ({(#1)},{(#2-1)})  -- ({(#1+1)},{(#2)})  -- ({(#1)},{(#2+1)})  
  -- ({(#1-1)},{(#2)})  -- cycle;
}

\newcommand\redfullrectangle[2]{
  \draw[red,very thick,fill=full] ({(#1)},{(#2-1)})  -- ({(#1+1)},{(#2)})  -- ({(#1)},{(#2+1)})  
  -- ({(#1-1)},{(#2)})  -- cycle;
}

\newcommand\redundeterminedrectangle[2]{
  \draw[red,very thick,fill=halfborder] ({(#1)},{(#2-1)})  -- ({(#1+1)},{(#2)}) 
  -- ({(#1)},{(#2+1)}) -- ({(#1-1)},{(#2)})  -- cycle;
}

\newcommand\emptyrectangle[2]{
  \draw[border] ({(#1)},{(#2-1)})  -- ({(#1+1)},{(#2)})  -- ({(#1)},{(#2+1)})  
  -- ({(#1-1)},{(#2)})  -- cycle;
}

\newcommand\fullrectangle[2]{
  \draw[border,fill=full] ({(#1)},{(#2-1)})  -- ({(#1+1)},{(#2)})  -- ({(#1)},{(#2+1)})  
  -- ({(#1-1)},{(#2)})  -- cycle;
}

\newcommand\undeterminedrectangle[2]{
  \draw[border,fill=halfborder] ({(#1)},{(#2-1)})  -- ({(#1+1)},{(#2)}) 
  -- ({(#1)},{(#2+1)}) -- ({(#1-1)},{(#2)})  -- cycle;
}

\newcommand\rectangle[3]{
  \ifthenelse{\equal{#3}{1}}{\fullrectangle{#1}{#2}}{\ifthenelse{\equal{#3}{2}}{\undeterminedrectangle{#1}{#2}}
  {\emptyrectangle{#1}{#2}}};
}

\newcommand\redrectangle[3]{
  \ifthenelse{\equal{#3}{1}}{\redfullrectangle{#1}{#2}}{\ifthenelse{\equal{#3}{2}}{\redundeterminedrectangle{#1}{#2}}
  {\redemptyrectangle{#1}{#2}}};
}

\newcommand\bCircle[3]{
    \draw [thick,colLines,fill=#3] ({(#1)},{(#2)}) circle (0.4);
}

\newcommand\sCircle[3]{
    \draw [thick,colLines,fill=#3] ({(#1)},{(#2)}) circle (0.15);
}

\newcommand\prop[4]{
    \gridLine{(#1)}{(#3)}{(#2)}{(#3)};
    \sCircle{(#1)}{(#3)}{#4};
    \sCircle{(#2)}{(#3)}{#4};
    \bCircle{((#1)+(#2))/2}{(#3)}{#4};
}

\newcommand\tproj[4]{
    \gridLine{(#1)}{(#2)}{(#1)}{(#3)};
    \sCircle{(#1)}{(#2)}{#4};
    \sCircle{(#1)}{(#3)}{#4};
    \sCircle{(#1)}{((#2)+(#3))/2}{#4};
}

\newcommand\tprop[4]{
    \gridLine{(#1)}{(#2)}{(#1)}{(#3)};
    \sCircle{(#1)}{(#2)}{#4};
    \sCircle{(#1)}{(#3)}{#4};
    \bCircle{(#1)}{((#2)+(#3))/2}{#4};
}

\newcommand\gridLine[4]{
    \draw [thick,colLines] ({(#1)},{(#2)}) -- ({(#3)},{(#4)});
}

\newcommand\leftHook[2]{
    \draw[thick,colLines] ({(#1)},{(#2)}) arc (90:240:0.2);
}

\newcommand\rightHook[2]{
    \draw[thick,colLines] ({(#1)},{(#2)}) arc (90:-50:0.2);
}

\newcommand\obs[2]{
    \draw [thick,rounded corners=1,colLines,fill=colObs] ({(#1)},{(#2)}) +(0.15,0.15) rectangle 
    +(-0.15,-0.15);
}

\newcommand\ME[2]{
    \draw [thick,colLines,fill=colLines] ({(#1)},{(#2)}) circle (0.15);
}

\newtheorem{theorem}{Theorem}

\makeatletter
\def\@mkboth#1#2{}
\newlength\appendixwidth
\preto\appendix{\addtocontents{toc}{\protect\patchl@section}}
\newcommand{\patchl@section}{%
  \settowidth{\appendixwidth}{\textbf{Appendix }}%
  \addtolength{\appendixwidth}{1.5em}%
  \patchcmd{\l@section}{1.5em}{\appendixwidth}{}{\ddt}%
}
\makeatother

\begin{document}
\title{Rule 54: Exactly solvable model of nonequilibrium statistical mechanics}
\author{Berislav Bu\v ca$^1$, Katja Klobas$^2$, and Toma\v z Prosen$^3$}
\address{$^1$ Clarendon Laboratory, University of Oxford, Parks Road, Oxford OX1 3PU, United Kingdom}
\address{$^2$ Rudolf Peierls Centre for Theoretical Physics, University of Oxford, Parks Road, Oxford OX1 3PU, United Kingdom}
\address{$^3$ Department of Physics, Faculty of Mathematics and Physics, University of Ljubljana, Ljubljana, Slovenia}
\eads{
  \mailto{berislav.buca@physics.ox.ac.uk},
  \mailto{katja.klobas@physics.ox.ac.uk},
  \mailto{tomaz.prosen@fmf.uni-lj.si}}

\begin{abstract}
    We review recent results on an exactly solvable model of nonequilibrium statistical mechanics, specifically the classical Rule 54 reversible cellular automaton and some of its quantum extensions.  We discuss the exact microscopic description of nonequilibrium dynamics as well as the equilibrium and nonequilibrium stationary states. This allows us to obtain a rigorous handle on the corresponding emergent hydrodynamic description, which is treated as well. Specifically, we focus on two different paradigms of Rule 54 dynamics. Firstly, we consider a finite chain driven by stochastic boundaries, where we provide exact matrix product descriptions of the nonequilibrium steady state, most relevant decay modes, as well as the eigenvector of the tilted Markov chain yielding exact large deviations for a broad class of local and extensive observables. Secondly, we treat the explicit dynamics of macro-states on an infinite lattice and discuss exact closed form results for dynamical structure factor, multi-time-correlation functions and inhomogeneous quenches. Remarkably, these results prove that the model, despite its simplicity, behaves like a regular fluid with coexistence of ballistic (sound) and diffusive (heat) transport. Finally, we briefly discuss quantum interpretation of Rule 54 dynamics and explicit results on dynamical spreading of local operators and operator entanglement.

\end{abstract}

\tableofcontents

\section{Introduction}

Exactly solved models are a major cornerstone of statistical mechanics and physics in general. While free (quadratic) models and their perturbations provide some insight, to achieve realistic statistical physical behaviour found in real materials requires strongly interacting models of which only several solvable classes are known. The first and most notable class, in the context of equilibrium physics, are two-dimensional lattice models related to solutions of the Yang-Baxter equation which give crucial exact information about the universality classes of statistical systems in two dimensions~\cite{baxterbook}, and relate to Bethe-ansatz solvable quantum models in one dimension~\cite{bethe1931theorie,sutherland2004beautiful,takahashi1999thermodynamics}. However, for studying out of equilibrium properties, in particular for time-dependent correlation functions or quantum quenches, such Yang-Baxter-Bethe solvable models represent a~much harder challenge. Computation of time correlation functions through the so-called form-factor-expansion is a formidable task, so far accomplished with only partial success in particular models \cite{pozsgay2008formI,pozsgay2008formII,essler2009finite,calabrese2011quantum,calabrese2012quantum,granet2020finite,granet2020systematic}. Nevertheless, integrability techniques resulted in a successful hydrodynamic approach in such models. The most notable of which is the so-called generalised hydrodynamics (GHD)~\cite{bertini2016transport,castroalvaredo2016emergent,alba2021generalizedhydrodynamic}, which has achieved remarkable success in predicting large space-time scale behaviour of observables in integrable systems.

However, hydrodynamics is not a rigorous theory and in particular it relies on the assumption of a clear separation of space-time scales. The fact that sub-ballistic corrections generically result in diffusion, as derived within GHD in Ref.~\cite{denardis2018hydrodynamic} (see also~\cite{denardis2021correlation}), could be at least partly attributed to the central assumption behind the hydrodynamic picture, which is the immediate loss of memory (correlations) in the quasi-particle scattering processes. It is thus of utmost importance to have at our disposal another type of model -- or a class of models -- with generic physical behaviour and for which dynamical physical quantities are accessible by a rigorous analysis free from assumptions.  Within such a class of models we can then achieve the `holy grail' of nonequilibrium statistical physics, which is to derive macroscopic transport laws from reversible and deterministic microscopic equations of motion \cite{lebowitz1999statistical}. 

Indeed, in the last several years it has been recognised that such a model exists and seemingly belongs to a class distinct from regular Bethe-ansatz solvable models. It is the Rule 54 of the family of reversible cellular automata (RCA54) proposed and classified in Ref.~\cite{bobenko1993two}. The model RCA54 can also be understood as a staggered sublattice version of the model 250R of reversible local automata as classified earlier by Takesue~\cite{takesue1987reversible}, or a discrete-time deterministic limit of the Fredrickson-Andersen model~\cite{fredrickson1984kinetic}. It is arguably the simplest microscopic (deterministic) physical theory in 1+1 dimensions with strong local interactions and asymptotically freely propagating excitations -- \emph{quasi-particles}, or \emph{solitons}. RCA54 has been proposed to be \emph{integrable} already in \cite{bobenko1993two} based on mainly qualitative arguments. The first exact solution, however, came in Ref.~\cite{prosen2016integrability}, where the nonequilibrium steady state of the model driven by a pair of chemical baths at the boundaries was analytically found. This led to many other results, such as generalisation to larger families of boundary driving~\cite{inoue2018two,prosen2017exact}, diagonalisation of the Markov propagator~\cite{prosen2017exact,bucaprosenreview}, and exact large-deviation results in the boundary driven setup~\cite{buca2019exact}. Later, many properties of the model on an infinite lattice without stochastic boundaries have been exactly obtained, such as dynamical structure factor and multi-time correlation functions~\cite{klobas2019timedependent,klobas2020matrix,klobas2020space}. The model has also been studied in the quantum context~\cite{gopalakrishnan2018facilitated}, and its simple but non-trivial dynamics was found to provide an ideal setting to study large-scale properties of the physics of local observables and operator spreading~\cite{gopalakrishnan2018operator,gopalakrishnan2018hydrodynamics,friedman2019integrable,alba2019operator,alba2020diffusion,klobas2021exact,klobas2021exactII,klobas2021entanglement}. Furthermore, an intriguing connection with the dynamics of $T\overline{T}$-deformed conformal field theories was recently established~\cite{medenjak2021ttbar,medenjak2021thermal}.

In particular, with respect to transport of quasiparticle excitations (or solitons), Rule 54 provides a model of a generic physical fluid with coexistence of ballistic (convective) and diffusive (conductive) transport. Remarkably, unlike typical Bethe-ansatz solvable systems, this model in many instances allows for fully \emph{closed-form} solutions despite being interacting and thus it allows the understanding of the aforementioned generic transport behaviour on a microscopic level. The purpose of this review is to provide a comprehensive overview of recent results on the Rule 54 model and discuss their comparison with simple predictions of hydrodynamic theory. 

\section{Definition of the model and summary of the results covered}
\begin{figure}
    \centering
    \includegraphics[width=\textwidth]{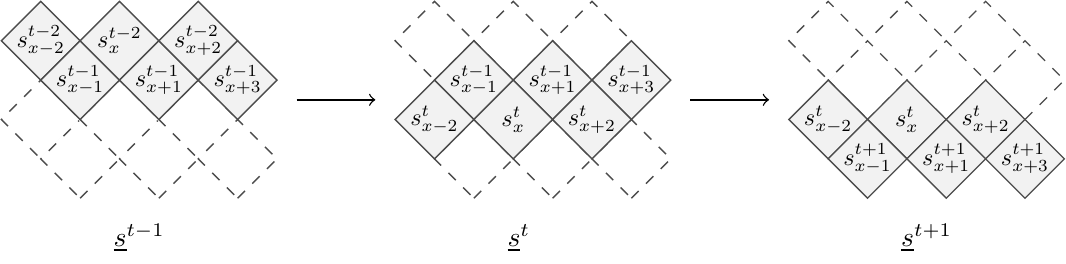}
    \caption{\label{fig:TEGeofig} Schematic representation of the time
    evolution. At time $t$, the sites with the same parity as $t$ get updated
    according to the local time-evolution rules~\eqref{eq:TErules}, while the
    rest stay the same.
    }
\end{figure}

We consider deterministic dynamics defined on a $1+1$ dimensional discrete
lattice with space-time points labelled by $(x,t)\in\mathbb Z^2$ and a field
variable $s_x^t$ taking only binary values $s_x^t\in\mathbb Z_2=\{0,1\}$. We
may restrict the dynamics only to a staggered (diamond, or light-cone)
sublattice of $\mathbb Z^2$ of points $(x,t)$ satisfying $x+t=0\pmod{2}$
and define
\begin{eqnarray}\label{eq:defGeometry1}
s^{t+1}_x = \chi(s^t_{x-1},s^{t-1}_x,s^t_{x+1})    
\end{eqnarray}
for some binary function $\chi:\mathbb Z_2^3\to\mathbb Z_2$. Specifically, the function with the binary code $54$ reads
\begin{eqnarray}\label{eq:TErules}
    s_2^{\prime}=\chi(s_1,s_2,s_3)\equiv s_1+s_2+s_3+s_1 s_3\pmod 2,
\end{eqnarray}
and can be graphically represented as
\begin{eqnarray}\fl \label{eq:TEDrules}
    \begin{tikzpicture}[scale=0.375,baseline={([yshift=-0.5ex]current bounding box.center)}]
        \rectangle{-1}{0}{0};
        \rectangle{0}{1}{0};
        \rectangle{1}{0}{0};
        \redrectangle{0}{-1}{0};
        \node at (-1,0) {\scalebox{1}{$s_1$}};
        \node at (0,1) {\scalebox{1}{$s_2$}};
        \node at (1,0) {\scalebox{1}{$s_3$}};
        \node at (0,-1) {\scalebox{1}{$s_2^{\prime}$}};
    \end{tikzpicture}\quad
    \begin{tikzpicture}[scale=0.375,baseline={([yshift=-0.5ex]current bounding box.center)}]
        \rectangle{-1}{0}{0};
        \rectangle{0}{1}{0};
        \rectangle{1}{0}{1};
        \redrectangle{0}{-1}{1};
    \end{tikzpicture}\quad
    \begin{tikzpicture}[scale=0.375,baseline={([yshift=-0.5ex]current bounding box.center)}]
        \rectangle{-1}{0}{0};
        \rectangle{0}{1}{1};
        \rectangle{1}{0}{0};
        \redrectangle{0}{-1}{1};
    \end{tikzpicture}\quad
    \begin{tikzpicture}[scale=0.375,baseline={([yshift=-0.5ex]current bounding box.center)}]
        \rectangle{-1}{0}{0};
        \rectangle{0}{1}{1};
        \rectangle{1}{0}{1};
        \redrectangle{0}{-1}{0};
    \end{tikzpicture}\quad
    \begin{tikzpicture}[scale=0.375,baseline={([yshift=-0.5ex]current bounding box.center)}]
        \rectangle{-1}{0}{1};
        \rectangle{0}{1}{0};
        \rectangle{1}{0}{0};
        \redrectangle{0}{-1}{1};
    \end{tikzpicture}\quad
    \begin{tikzpicture}[scale=0.375,baseline={([yshift=-0.5ex]current bounding box.center)}]
        \rectangle{-1}{0}{1};
        \rectangle{0}{1}{0};
        \rectangle{1}{0}{1};
        \redrectangle{0}{-1}{1};
    \end{tikzpicture}\quad
    \begin{tikzpicture}[scale=0.375,baseline={([yshift=-0.5ex]current bounding box.center)}]
        \rectangle{-1}{0}{1};
        \rectangle{0}{1}{1};
        \rectangle{1}{0}{0};
        \redrectangle{0}{-1}{0};
    \end{tikzpicture}\quad
    \begin{tikzpicture}[scale=0.375,baseline={([yshift=-0.5ex]current bounding box.center)}]
        \rectangle{-1}{0}{1};
        \rectangle{0}{1}{1};
        \rectangle{1}{0}{1};
        \redrectangle{0}{-1}{0};
    \end{tikzpicture}
\end{eqnarray}
where white/black boxes represent field values $0/1$, and red-framed boxes correspond to the updated values at the next time instance. As the graphical representation suggests, the middle site is changed whenever at least one of the neighbours is black. Dynamics is generated by two sequences of parallel updates (see Fig.~\ref{fig:TEGeofig}), which in two steps maps a {\em configuration} over a zig-zag saw $\underline{s}^{t-1}$ first to $\ul{s}^t$ and then to $\ul{s}^{t+1}$, where $\underline{s}^{t} \equiv (\ldots,s^{t-1}_{x-1},s^{t}_{x},s^{t-1}_{x+1},\ldots)$. An example of a typical trajectory can be seen in Fig.~\ref{MCsnap} (ignoring the boundaries for the time being). The dynamics can be interpreted as a gas of \emph{solitons} travelling with speeds $\pm 1$ which scatter pair-wise, while each scattering produces a time-lag (or shift) of a soliton of one step in time (or space). In particular, the diagrams in~\eqref{eq:TEDrules} can be interpreted as different possible local configurations of solitonic dynamics. The first diagram represents empty space, diagrams $2$ and $7$ show a left mover, diagrams $4$ and $5$ represent a right mover, while the rest of the diagrams correspond to different steps of the scattering event between two oppositely-moving solitons---the two solitons first merge, temporarily disappear and reappear immediately afterwards (given by diagrams $6$, $3$ and $8$ respectively).

The key problem discussed in this review is how to achieve the full understanding of equilibrium and nonequilibrium statistical physics of this model.  This goal can be pursued within two paradigms of statistical mechanics.  In the first, we consider the system defined on a finite lattice $\mathbb Z_n$, of even size $n$, and we couple the left and the right edge of the chain to stochastic baths of solitons. After removing the degrees of freedom of the reservoirs we end up with a perfect Markov chain model, where the cells in the bulk are updated deterministically, while the cells near the boundaries are evolved stochastically. The set of parameters characterizing the boundary driving uniquely determines the nonequilibrium stationary state (NESS) that the system approaches at long times. In Section~\ref{sect:BD} we show how to provide exactly solvable Markovian boundaries and how to construct a correlated NESS in terms of the so-called \emph{patch-state ansatz}. Additionally, we prove a general ergodicity theorem for the boundary driven setup, which guarantees the uniqueness and exponential relaxation to the steady state. In Section~\ref{sect:MPA} we then provide a solution to the steady state cancellation mechanism in terms of a simple cubic algebra whose representation yields a compact matrix product solution to the full steady state.  Remarkably, this form can be extended to subleading eigenvectors of the Markov propagator and allows us to obtain an essential part of its spectrum, which, in particular, gives access to the spectral gap characterising the relaxation to the NESS. Moreover, the same cubic algebra can be further generalised to generate an exact large deviation function describing fluctuations (in the steady state) of a large class of spatially inhomogeneous observables, as elaborated in Section~\ref{sect:largedev}.

The second main paradigm considers the dynamics of the thermodynamic states of the system defined on an infinite lattice, or equivalently, the dynamics of local observables of the large but finite system far from the boundaries.  In Section~\ref{sect:hydro} we provide a basic hydrodynamic description of the model using the two elementary local conserved charges: the densities of left and right movers. In Section~\ref{sec:TS} we provide a space-time dual description of the model. We start by considering an efficient matrix-product-state representation of stationary probability distributions of configurations in time, which determine all multi-time correlations of ultra-local observables, evaluated in equilibrium thermodynamic states. We then proceed by showing that RCA54 allows for a deterministic description also when considering the space-evolution of time-configurations. For facilitating computations we introduce efficient diagrammatics, akin to quantum circuit diagrams, which are used to compactly encode a diversity of rather formal algebraic relations. Arguably the strongest result, elaborated in Section~\ref{sect:timeMPA}, is the construction of an exact time-dependent matrix product ansatz (tMPA) for time evolution of arbitrary local observables. Although the dimension of tMPA is formally infinite, it effectively grows only as $\propto t^2$ when time evolution up to time $t$ is considered. The result allows for a number of explicit computations, such as the dynamical structure factor (two-point space-time correlation function) and the time-dependent density profiles following an inhomogeneous quench. Furthermore, it provides a formal bound on the rate of operator spreading for the quantum interpretation of the model.

\section{Boundary driven cellular automaton and general equilibrium states}
\label{sect:BD}

\begin{figure}
 \centering	
\vspace{-1mm}
\includegraphics[width=0.38\columnwidth]{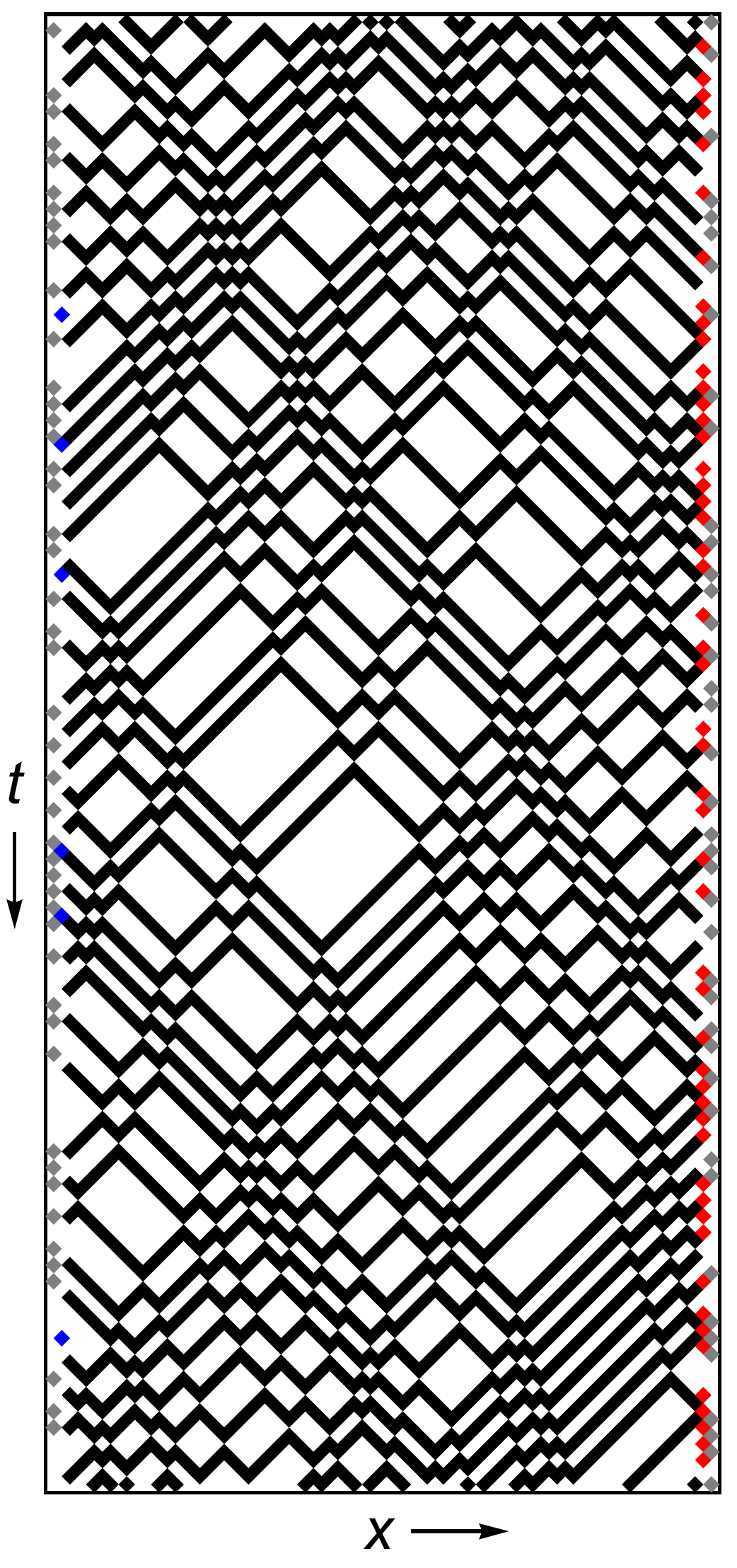}
\vspace{-1mm}
\caption{A snapshot of Monte Carlo dynamics of nonequilibrium stochastically boundary driven deterministic CA, Rule 54, for $n=80$, $\alpha=0.9,\beta=0.1,\gamma=0.4,\delta=0.6$
(Bernoulli driving \eqref{eq:bernoulli} with $\zeta=\eta=1/2$).
Time runs downwards. Grey squares denote occupied environmental cells which are generated by a Bernoulli shift with probability $1/2$, while blue (red) squares are occupied boundary cells determined via ultralocal Markov chains. Note that the right end is ``hotter'' than the left one and that the average (steady-state) soliton current flows to the left.}
\label{MCsnap}
\end{figure}

We begin by considering the model subject to stochastic boundary driving. The system is defined on an \emph{even} number of sites $n$ with instantaneous dynamical variables specified over a zig-zag saw ,
\begin{eqnarray}
    \un{s}=(s_1,s_2,s_3\ldots,s_n)\equiv (s^1_1,s^0_2,s^1_3\ldots,s^0_n)\in\mathcal{C}_n,
\end{eqnarray}
where $\mathcal{C}_n=\mathbb{Z}_2^{\times n}$ is the set of all configurations on the lattice of size $n$.  The \emph{macroscopic} state  of the system $\vec{p}$ is a probability distribution over the set of all configurations and can be interpreted as a $2^n$-dimensional vector
\begin{eqnarray}
\vec{p}=(p_{\ul{s}})_{\ul{s}\in\mathcal{C}_{n}}\in\left(\mathbb{R}^2\right)^{\otimes n},
\end{eqnarray}
with $p_{\ul{s}}\ge0$ denoting the probability of configuration $\ul{s}$. If dynamics were completely deterministic, such a probabilistic description of a finite system would be redundant. However, we allow for stochastic updates of boundary pairs of cells $(s_1,s_2)$, and $(s_{n-1},s_n)$ which cannot be consistently fixed by the deterministic rule~\eqref{eq:TErules}.  We make a minimal assumption that the stochastic updates of boundary cells are local and Markovian, i.e.\ the probability of jumps of $s_1$ (or $s_n$) can only depend on cells $s_1,s_2$ (or $s_{n-1},s_n$) in the immediate space-time neighbourhood.

We thus define the time evolution of the state vector $\vec{p}(t)$ starting from some initial state $\vec{p}(0)$ as,
\begin{eqnarray}
\mathbf{p}(t)=\U^t \mathbf{p}(0),
\label{eq:markov}
\end{eqnarray}
where the one-period propagator $\U$ ($2^n\times 2^n$ matrix)
is factored in two half-time steps, as shown in~Fig.~\ref{stochasticprop},
\begin{eqnarray}
\U=\Uo \Ue. \label{Ustoch}
\end{eqnarray}
Each step is now given in terms of parallel updates of even or odd sites
\begin{eqnarray}
&\Ue=U_{123}U_{345}\cdots U_{n-3\,n-2\,n-1} U^{\rm{R}}_{n-1\,n}, \label{Ue}\\
&\Uo= U^{\rm{L}}_{12} U_{234} U_{456}  \cdots U_{n-2\,n-1\,n}, \label{Uo}
\end{eqnarray}
where $U_{j-1\,j\,j+1}=\one_{2^{j-2}}\otimes U\otimes \one_{2^{n-j-1}}$, with $U$
being an $8\times 8$ matrix, acts non-trivially on a triple of adjacent sites
$(j-1,j,j+1)$ and only affects the cell at position $j$ depending on the values
of cells at positions $(j-1,j,j+1)$ according to the Rule 54. Specifically:
\be \label{eq:defU}
U_{(t,t',t''),(s,s',s'')} =
\delta_{t,s} \delta_{t',\chi(s,s',s'')} \delta_{t'',s''}.
\ee
The boundary updates
\begin{eqnarray}
U^{\rm{L}}_{12}=U^{\rm{L}} \otimes \one_{2^{n-2}}, \qquad
U^{\rm{R}}_{n-1\,n}= \one_{2^{n-2}} \otimes U^{\rm{R}},
\end{eqnarray}
are generated in terms of $4\times 4$ stochastic matrices acting on the left-most (right-most) two sites. As $U^{\rm{L,R}}$ are stochastic, i.e.\ their non-negative matrix elements in each column add to one, the full propagator $\U$ is a stochastic matrix conserving probability (see Fig.~\ref{stochasticprop}). 
\begin{figure}
 \centering	
\vspace{-1mm}
\includegraphics[width=0.45\columnwidth]{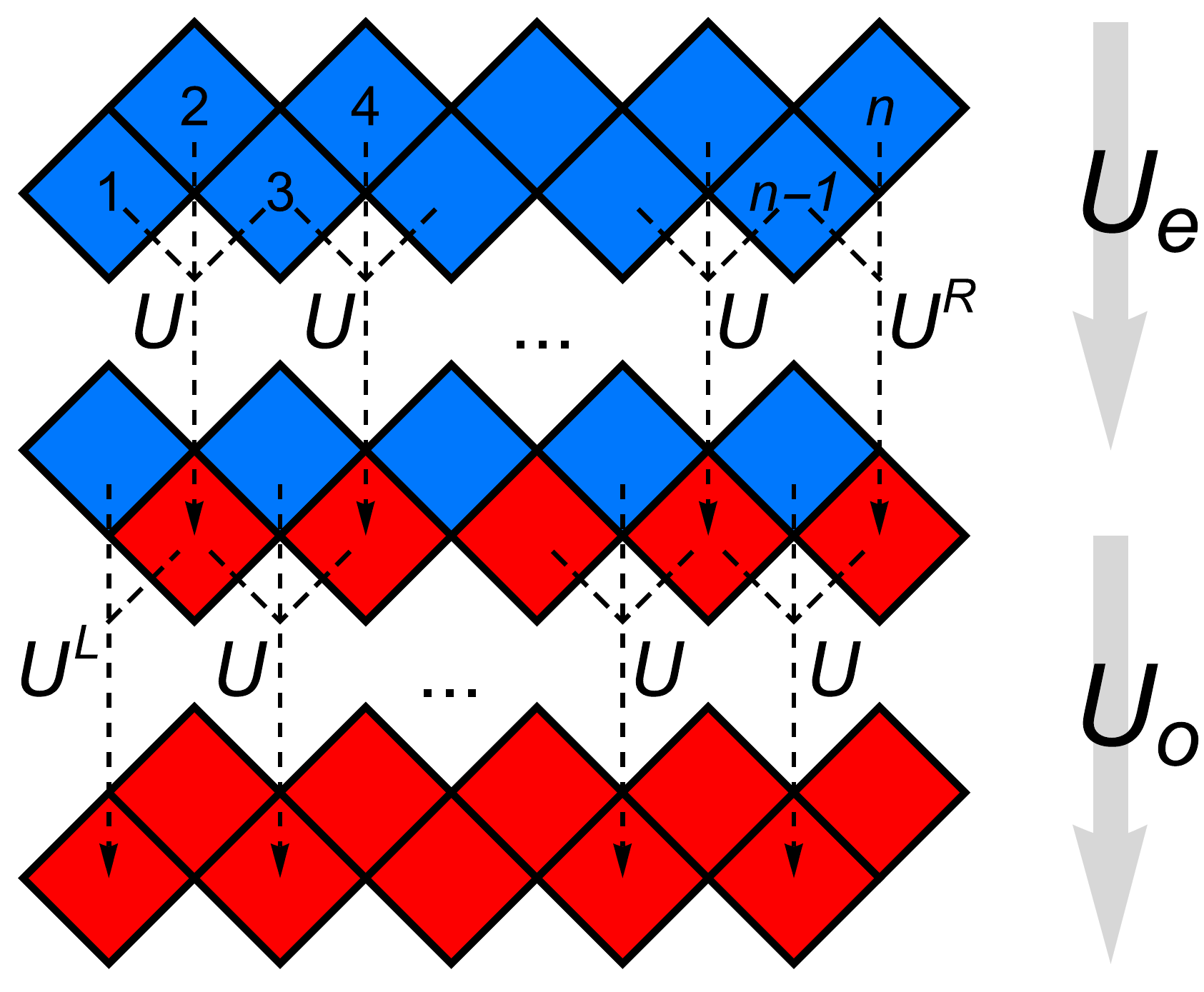}
\vspace{-1mm}
\caption{Illustration of the stochastic boundary driving (\ref{Ustoch}) with two site boundary stochastic maps $U^{\rm L,R}$ (see Eqs.~(\ref{Ue},\ref{Uo})). In blue/red we denote the sites before/after the update.}
\label{stochasticprop}
\end{figure}
The boundary matrices have to satisfy the following compatibility conditions
\begin{eqnarray}
[U^{\rm L}_{12},U_{234}]=0,\qquad
[U_{123},U^{\rm R}_{34}]=0,
 \label{boundcomm}
\end{eqnarray}
essentially implying that only the boundary cell may get stochastically updated, while the site next to it (i.e.\ $n-1$ in $U^{\mathrm{R}}_{n-1\,n}$ and $2$ in $U^{\mathrm{L}}_{1 2}$) acts as a control cell, analogously to the non-central cells in the bulk propagator.

Requiring that the many-body Markov chain process (\ref{eq:markov}) admits an exactly solvable nonequilibrium steady state (NESS) -- an eigenvector $\mathbf{p}_0$ of eigenvalue $1$, $\U \vec{p}_0 =\mathbf{p}_0$ -- puts additional constraints to the boundary operators $U^{\rm L},U^{\rm R}$. A complete classification of exactly solvable {\em local} (2-site) boundaries is composed of two multi-parametric families. The first is the so-called {\em Bernoulli driving} (originally proposed in \cite{prosen2016integrability} and generalised in \cite{inoue2018two}):
\begin{eqnarray}\fl
U^{\rm L}\! = \! 
\begin{bmatrix}
q^{\rm L}_1 & & 1-q^{\rm L}_2 & \cr
& \alpha & & 1-\beta \cr
1-q^{\rm L}_1 & & q^{\rm L}_2 & \cr
& 1-\alpha & & \beta 
\end{bmatrix}\!, \qquad
U^{\rm R}\! = \!
\begin{bmatrix}
q^{\rm R}_1 & 1-q^{\rm R}_2 & & \cr
1-q^{\rm R}_1  & q^{\rm R}_2 & & \cr
& & \gamma & 1-\delta \cr
& & 1-\gamma & \delta 
\end{bmatrix}\!, \label{eq:bernoulli}
\end{eqnarray}
where
\begin{eqnarray}
    \eqalign{
        &q^{\rm L}_1=\zeta + \alpha - 2\zeta\alpha,\qquad
        q^{\rm R}_1 = \eta + \gamma-2\eta\gamma, \\
    &q^{\rm L}_2=\zeta + \beta - 2\zeta\beta,\qquad
    q^{\rm R}_2 = \eta + \delta-2\eta\delta.}
\end{eqnarray}
Each bath is parametrised by a triple of independent parameters: $\zeta,\alpha,\beta\in[0,1]$ for the left boundary, and $\eta,\gamma,\delta\in[0,1]$ for the right boundary.  They admit a simple interpretation in terms of a composition of an ultralocal Markov chain (a $2\times 2$ stochastic matrix parametrised by probabilities $\alpha,\beta$ on the left and $\gamma,\delta$ on the right), and a local Rule 54 map centred around the same site, while the value of the imaginary cell inside the bath is given in terms of a Bernoulli process (with the probability $\zeta$ on the left and $\eta$ on the right).

Another distinct family is the {\em conditional driving}
\cite{prosen2017exact}:
\begin{eqnarray}\fl
U^{\rm{L}}\!=\!
    \begin{bmatrix}
 \alpha  &  & \alpha  &  \\
 0 & \beta  &  & \beta  \\
 1-\alpha  &  & 1-\alpha  &  \\
 0 & 1-\beta  &  & 1-\beta  \\
    \end{bmatrix}\!,\qquad
U^{\rm{R}}\!=\!
    \begin{bmatrix}
 \gamma  & \gamma  &  &  \\
 1-\gamma  & 1-\gamma  &  &  \\
  &  & \delta  & \delta  \\
  &  & 1-\delta  & 1-\delta  \\
    \end{bmatrix}\!,
\label{eq:conditional}
\end{eqnarray}
where $\alpha, \beta, \gamma, \delta \in [0,1]$ are some driving rates parametrising the left and the right bath. We refer to this as conditional driving, since in $U^{\rm{L}}_{12}$ ($U^{\rm{R}}_{n-1\,n}$) the probability of changing the site 1 depends \emph{only} on the state of the neighbouring site 2 (changing the site $n$ depends only on the state of the site $n-1$). For instance, if the site 2 is in state 0 then the site 1 will be stochastically set to state 0 with the rate $\alpha$ or to state 1 with the rate $1-\alpha$. On the other hand, if the site 2 is in the state 1, the site 1 will be set to state 0 or 1 with the rates $\beta$ or $1-\beta$, respectively. The analogous holds for $U^{\rm{R}}_{n-1\,n}$.

These two classes of boundary driving maintain the integrability of the model, which may be intuitively understood as a consequence of the fact that they (with some probabilities) create or destroy solitons at the boundaries. We note that the information about three-site subconfigurations is needed to completely characterize the local particle content, therefore a minimal matrix product ansatz providing such a state should admit a $3$-dimensional auxiliary space. This fact will be important in the sections below, when dealing with solutions to boundary driving.

\subsection{Holographic ergodicity}\label{subsec:holErg}
\begin{figure}
    \centering	
    \vspace{-1mm}
    \includegraphics[width=0.26\columnwidth]{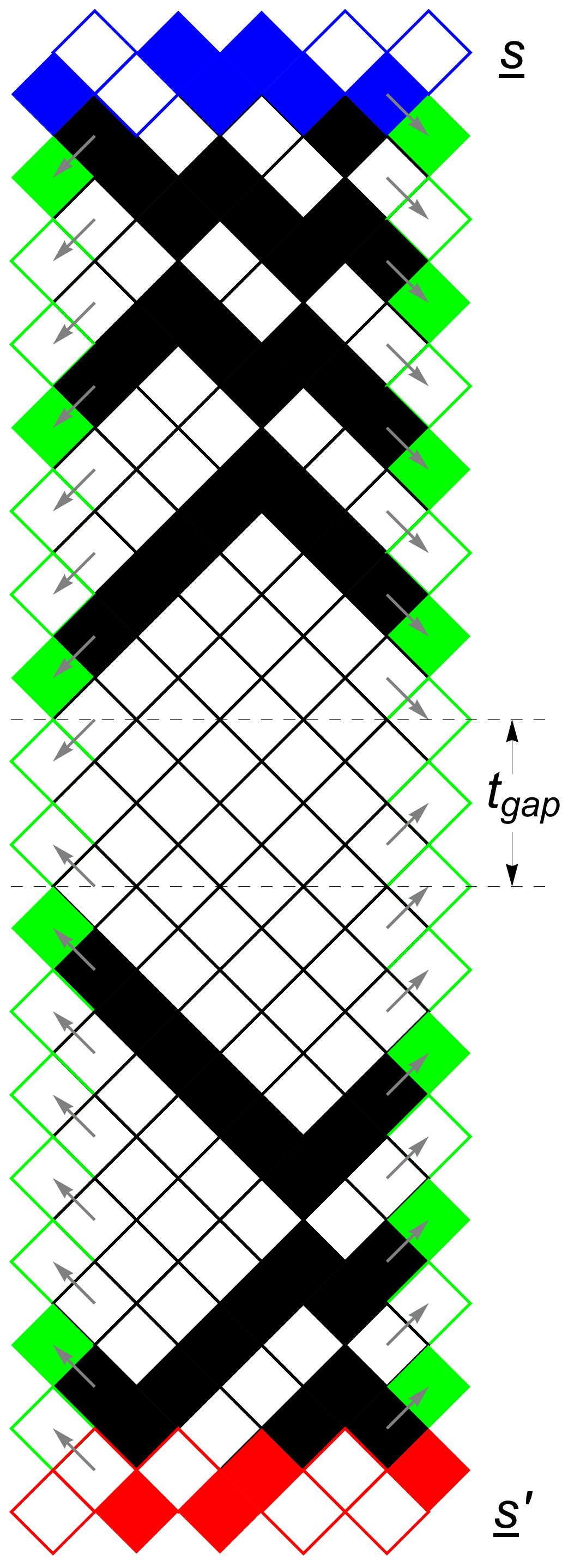}
    \vspace{-1mm}
    \caption{Illustration of the proof of irreducibility and aperiodicity of the Markov matrix $\U$. Blue and red configurations, $\un{s}=(1,0,0,1,1,1,1,0,1,0)$ and $\un{s}'=(0,0,1,0,1,1,0,0,0,1)$, are connected with the Markov-graph walk for generic probabilities $0<\alpha,\beta,\gamma,\delta < 1$ in at least $t_0=15$ time steps where the boundary cells are chosen as indicated by green cells (the boundary conditions are generated by causal/anti-causal absorbing boundaries in the upper/lower part of the walk). The values of the boundary cells are thus determined by copying the values of the near-by bulk cells in the direction of the grey arrows. Consequently $\left[\U^{t_0+t_{\rm gap}}\right]_{\un{s},\un{s}'} > 0$ for any $t_{\rm gap} \ge 0$ (and any other pair of initial/final configurations $\un{s},\un{s}'$ with possibly different $t_0$), which implies irreducibility and aperiodicity of $\U$.
    }
    \label{IrredChart}
\end{figure}
Establishing the existence of a unique NESS and relaxation towards it from an arbitrary initial probability state vector amounts to showing the following statement~\cite{prosen2016integrability}:
\begin{theorem}
The $2^n \times 2^n$ matrix $\U$, Eq.~\eqref{Ustoch}, is irreducible and aperiodic for generic values of driving parameters, more precisely, for an open set $0 < \alpha,\beta,\gamma,\delta < 1$ for conditional driving
\eqref{eq:conditional}, and $0 < \zeta,\alpha,\beta,\eta,\gamma,\delta < 1$ for Bernoulli driving \eqref{eq:bernoulli}.
\end{theorem}
\begin{proof}
    According to the Perron-Frobenius theorem~\cite{perron_frobenius}, a finite, non-negative matrix $\U$ is \emph{irreducible}, if for any pair of configurations $\ul{s},\ul{s}^{\prime}\in\mathcal{C}_n$ there exists a natural number $t_0\in\mathbb{N}$ such that the corresponding matrix element of the $t_0$-th power is nonzero,
\begin{eqnarray}
    \left[\U^{t_0}\right]_{\ul{s}^{\prime},\ul{s}}>0.
\end{eqnarray}
An irreducible matrix $\U$ is \emph{aperiodic} if for any configuration $\ul{s}\in\mathcal{C}_n$, the greatest common divisor of the set of recurrence times $\{t_j\}$ is $1$,
\begin{eqnarray}
    \gcd\left(\{t_j\in\mathbb{N};\ \left[\U^{t_j}\right]_{\ul{s},\ul{s}}>0\}\right)=1.
\end{eqnarray}

Let us first show irreducibility. By the definition of boundary propagators (cf.~\eqref{eq:bernoulli} and~\eqref{eq:conditional}), the Markov matrix $\U$ connects each configuration $\un{s}$ to exactly 4 other configurations $\un{s}'$, which exhibit all the possible configurations of boundary bits, $(s'_1,s'_n)\in\{(0,0),(0,1),(1,0),(1,1)\}$. This holds for all values of parameters $\alpha,\beta,\gamma,\delta$ (and $\zeta,\eta$), except for the marginal case when some of the parameters are equal to $0$ or $1$, but this option is excluded by the assumption of the theorem.  Let us now take a sufficiently large positive integer $t_0$, to be determined below, and fix
$\un{s}(t_0)\equiv \un{s}',\un{s}(0)\equiv\un{s}$.
We shall then construct a walk 
\be
\un{s}(0)\to\un{s}(1)\to\un{s}(2)\to \cdots \to\un{s}(t_0),
\ee 
    which can be understood as a path through the Markov graph defined by positive elements of $\U$ that connects $\un{s}$ and $\un{s}'$ in $t_0$ steps and implies $\left[\U^{t_0}\right]_{\un{s}',\un{s}} > 0$ (see Fig.~\ref{IrredChart} for a `self-contained' graphic illustration of the idea of proof). We are still free to choose the values of the boundary cells $s_{1,n}(t)$ along the walk $t \in\{1,2,\ldots,t_0-1\}$ apart from the ends. For the first part of the walk $t = 1,2\ldots t_+$, up to some $t_+<t_0$, we are fixing them with the rule
\be
s_{1}(t) = s_{2}(t-1),\quad s_{n}(t) = s_{n-1}(t-1).
\label{eq:absorb}
\ee
    The evolution of the interior values of the cells $s_x(t)$ for $1 < x < n$ and $t \le t_+$ is then completely specified by the deterministic RCA54, while (\ref{eq:absorb}) provide the {\em causal absorbing} boundary conditions. Indeed, each time the boundary cell, say $x=1$, gets occupied, $s_{1}(t)=1$, the soliton is absorbed (see Fig.~\ref{IrredChart}).  As the solitons only move ballistically (at speed 1) and scatter pairwise (with time-lag 1), it is clear that a finite time scale $t_+ \in\mathbb N$ exists, surely smaller than $n^2$, after which all the solitons will be absorbed and we end up in a vacuum configuration $\un{s}(t_+) = (0,0\ldots,0)$. 

For the rest of the walk $t\in\{t_++1,\ldots t_0\}$ we need to show that alternative boundary rules exists, which create the configuration $\un{s}'$ out of the vacuum in another $t_-=t_0-t_+$ steps. This is easily achieved by using {\em time-reversibility} of RCA54 and arguing that a vacuum configuration is again generated from $\un{s}'$ in some $t_-$ steps if the {\em anti-causal absorbing} boundary conditions are set, which are equivalent to \eqref{eq:absorb} when the time runs backwards,
\be
s_{1}(t) = s_{2}(t+1),\; s_{n}(t) = s_{n-1}(t+1),\quad {\rm for}\quad t = t_0-1,\ldots,t_0 - t_-\,.
\label{eq:absorb2}
\ee
The entire walk then connects $\un{s}$ to $\un{s}'$ in $t_0=t_++t_-$ steps and implies $\left[\U^{t_0}\right]_{\un{s}',\un{s}} > 0$, for an arbitrary pair $\un{s},\un{s}' \in {\mathcal{C}_n}$, with the minimal possible integer $t_0$ in general depending on the choice of $\un{s},\un{s}'$. This proves irreducibility of \eqref{Ustoch}.

Considering $\un{s}'=\un{s}$, we have just shown that $\left[\U^{t_0}\right]_{\un{s},\un{s}} > 0$ for some $t_0=t_{+}+t_{-}$ depending on $\un{s}$.  However, in between annihilating the configuration $\un{s}$ in $t_+$ time steps and then creating it again in another $t_-$ steps\footnote{Note that in general $t_-\neq t_+$ as a generic configuration $\un{s}$ is not time-reversal invariant.}, we can stay in the vacuum state for an arbitrary additional number of steps $t_{\rm gap} \ge 0$. In this way the walk is increased by a segment of $t_{\rm gap}$ intermediate vacuum configurations, and we still have a non-zero matrix element,
\be
    \left[\U^{t_0 +t_{\rm gap}}\right]_{\un{s},\un{s}} > 0.
\ee
The greatest common divisor of the set $\{t_0+t_{\rm gap};\ t_{\rm gap}\in\mathbb Z_+\}$ is clearly 1, so we have shown aperiodicity.
\end{proof}

In conclusion, the Perron-Frobenius theorem \cite{perron_frobenius} guarantees that the NESS probability state vector $\mathbf{p}_0$, satisfying the fixed point condition $\U\vec{p}_{0}=\vec{p}_{0}$, is \emph{unique} and all other eigenvalues of $\U$ lie strictly inside the unit circle.  As a consequence, the Markov dynamics \eqref{eq:markov} is {\em ergodic and mixing} and an arbitrary initial probability state vector $\mathbf{p}(0)$ converges to NESS exponentially fast in $t$. %Note, however, that the proof is strictly valid only for finite system sizes. In fact, as we will see later, in the thermodynamic limit we can have additional eigenvalues that converge to the unit circle and allow for possibly non-mixing dynamics or non-stationarity. 

\subsection{NESS: Patch State Ansatz} \label{sec:PSA}

\begin{figure}
 \centering	
\vspace{-1mm}
\includegraphics[width=0.5\columnwidth]{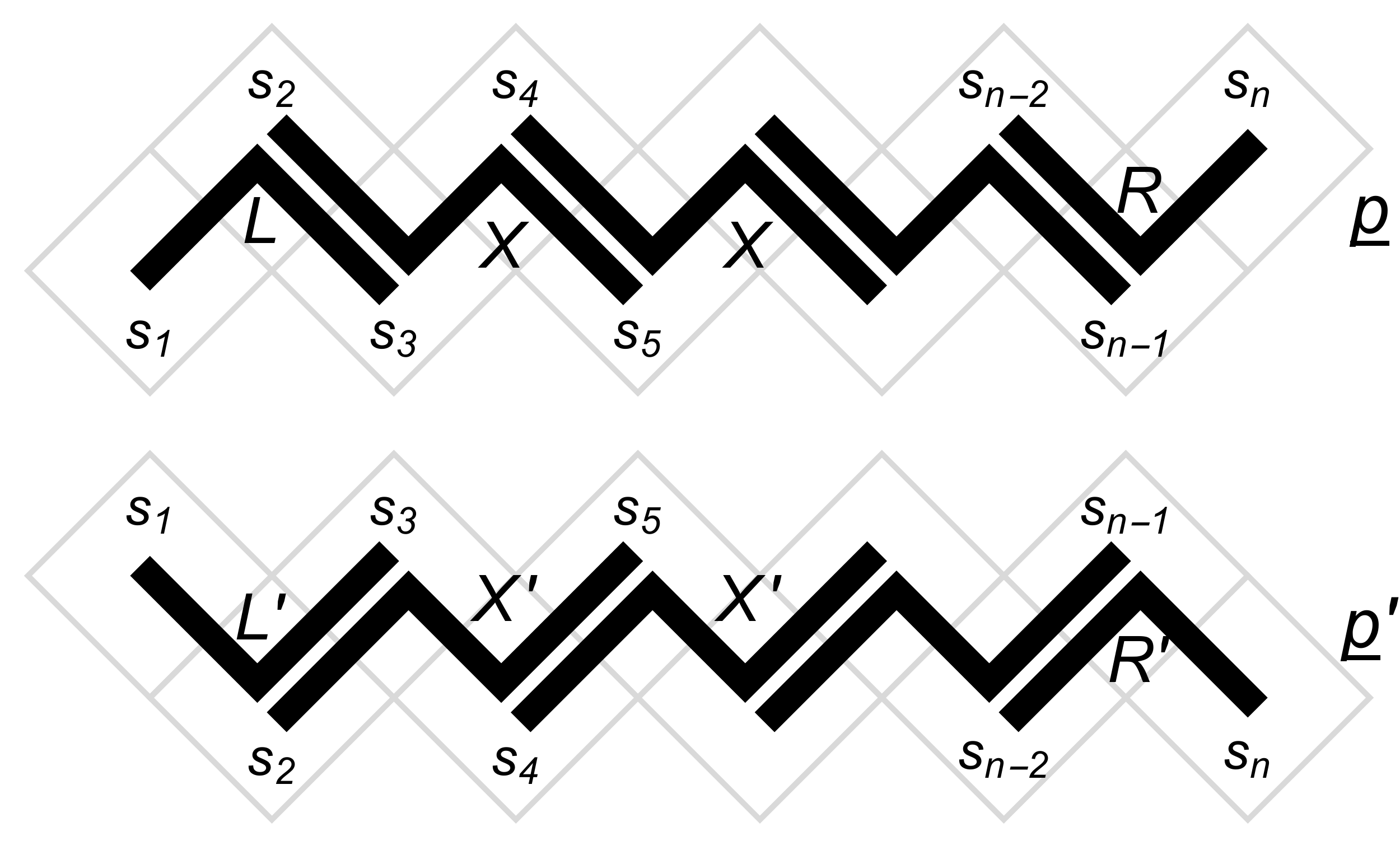}
\vspace{-1mm}
\caption{Illustration of the patch state ansatz \eqref{eq:patch} for NESS probability state vectors $\mathbf{p},\mathbf{p}'$.}
\label{fig:PSAscheme}
\end{figure}

The NESS probability state vector $\mathbf{p}=\mathbf{p}_0$ can be split into even and odd time slice, satisfying a pair of fixed point equations
\begin{eqnarray}
    \mathbf{p}'=\Ue\mathbf{p},\qquad
    \mathbf{p}=\Uo\mathbf{p}'.
    \label{eq:fp}
\end{eqnarray}
We shall now explicitly construct the probability state vectors $\mathbf{p}$ and $\mathbf{p}'$, by solving Eq.~\eqref{eq:fp} in terms of a simple ansatz,  which we term a {\em patch state ansatz} (PSA) [illustrated in Fig.~\ref{fig:PSAscheme}].

For either of the two boundary driving families (\ref{eq:bernoulli},\ref{eq:conditional}), the NESS solution $\mathbf{p},\mathbf{p}'\in\mathbb{R}^{2^{n}}$ of the fixed point condition (\ref{eq:fp}) can be written in the form
\begin{eqnarray}
    \eqalign{
  p_{s_1,s_2,\ldots,s_n} = L_{s_1 s_2 s_3} X_{s_2 s_3 s_4 s_5} X_{s_4 s_5 s_6 s_7}  \cdots X_{s_{n-4} s_{n-3} s_{n-2} s_{n-1}} R_{s_{n-2} s_{n-1} s_n}, \\
    p'_{s_1,s_2,\ldots,s_n} = L'_{s_1 s_2 s_3} X'_{s_2 s_3 s_4 s_5} X'_{s_4 s_5 s_6 s_7} \cdots X'_{s_{n-4} s_{n-3} s_{n-2} s_{n-1}} R'_{s_{n-2} s_{n-1}s_n},}
\label{eq:patch}
\end{eqnarray}
for some rank-4 and rank-3 tensors of strictly positive components $X_{ss'uu'}$, $X'_{ss'uu'}$, $L_{suu'}$, $L'_{suu'}$, $R_{ss'u}$, $R'_{ss'u}$, with binary indices $s,s',u,u'\in\{0,1\}$. To uniquely determine the algebraic expressions for these tensors in terms of the parameters of the model, one has to plug the ansatz~\eqref{eq:patch} into the NESS condition~\eqref{eq:fp}.

First we find a minimal sufficient set of equations that the tensors have to satisfy, by fixing the normalization and taking into account the gauge symmetry of the ansatz. The normalization of the PSA \eqref{eq:patch} can be chosen such that 
\be
X_{0000}=X'_{0000}=1,\quad L_{000}=R'_{000}=1.
\ee
Clearly $X_{0000}=X'_{0000}$, otherwise the probabilities of the vacuum configurations
$p_{00\ldots 0}$ and $p'_{00\ldots0}$ would scale differently with $n$ which is not possible since they are directly connected by both even and odd propagators, i.e.\ 
$\left[{\Uo}\right]_{(00\ldots 0)(00\ldots0)},\left[\Ue\right]_{(00\ldots 0)(00\ldots0)}>0$.
%$U_{\rm o}$ and $U_{\rm e}$ directly connect $(0,0,\ldots,0,0)$ only with configurations $(s_1,0,\ldots,0,s_n)$.

Let us now assume the ansatz (\ref{eq:patch}) and write all components of  Eqs.~(\ref{eq:fp}),
$\left[\Ue\vec{p}-\vec{p}'\right]_{\un{s}} = \left[\Uo\vec{p}'-\vec{p}\right]_{\un{s}}=0$,
pertaining to 4-cluster configurations in the bulk of the form\footnote{Symbol $0^{\{k\}}$ denotes $0$ repeated $k$ times.} $\un{s} = (0^{\{ 2k+1\}},s,s',u,u',0^{\{n-5-2k\}})$, for $k=0,1,\ldots,m-3$, which results in the following two \emph{bulk} equations:
\begin{eqnarray}
    \eqalign{
X'_{0 0 s s'} X'_{s s' u u'} X'_{u u' 0 0} X'_{0000} = \label{eq:bulk1and2}\\
    \qquad X_{0 0 \chi(0 s s') s'} X_{\chi(0 s s') s' \chi(s' u u') u'} X_{\chi(s' u u') u' \chi(u' 0 0) 0} X_{\chi(u' 0 0) 000},\\
X_{0000} X_{0 0 s s'} X_{s s' u u'} X_{u u' 0 0} = \\
    \qquad X'_{000\chi(0 0 s)} X'_{0 \chi(0 0 s) s \chi(s s' u)} X'_{s \chi(s s' u) u \chi(u u' 0)} X'_{u \chi(u u' 0) 0 0}.}
\end{eqnarray}
Similarly, considering $3$-cluster configurations near each boundary, $\un{s}=(v',s,s',0^{\{n-3\}})$ and $\un{s}=(0^{\{n-3\}},s,s',u)$, we obtain a set of $4$ \emph{boundary} equations,
\begin{eqnarray}
\eqalign{
L'_{v' s s'}X'_{s s' 0 0}X'_{0000}R'_{000} = 
L_{v' \chi(s' s s') s'}X_{\chi(v' s s') s' \chi(s' 0 0) 0}X_{\chi(s' 0 0) 000}\frac{R_{000}+R_{001}}{2}, \\
L'_{000}X'_{0 0 s s'} R'_{s s' u} = 
L_{000}\sum_{t',t} P^{\rm R}_{(s',u),(t',t)} X_{00 \chi(0 s t') t'} R_{\chi(0 s t') t' t},\\
L_{v' s s'} X_{s s' 0 0}R_{000} =
 \sum_{t',t} P^{\rm L}_{(v',s),(t',t)} L'_{t' t \chi(t s' 0)} X'_{t \chi(t s' 0) 0 0}R'_{000},\\
L_{000}X_{0000}X_{0 0 s s'} R_{s s' u} = 
    \frac{L'_{000}+L'_{100}}{2}X'_{000 \chi(0 0 s)}X'_{0 \chi(0 0 s) s \chi(s s' u)} R'_{s \chi(s s' u) u}\, . \label{beqs}}
\end{eqnarray}
%\begin{eqnarray}
%&&  X'_{0 0 s s'} X'_{s s' u u'} X'_{u u' 0 0} X'_{0000} = \label{eq:bulk1}\\
%&& \qquad X_{0 0 \chi(0 s s') s'} X_{\chi(0 s s') s' \chi(s' u u') u'} X_{\chi(s' u u') u' \chi(u' 0 0) 0} X_{\chi(u' 0 0) 000},\nonumber\\
%&& X_{0000} X_{0 0 s s'} X_{s s' u u'} X_{u u' 0 0} = \nonumber \\
%&& \qquad X'_{000\chi(0 0 s)} X'_{0 \chi(0 0 s) s \chi(s s' u)} X'_{s \chi(s s' u) u \chi(u u' 0)} X'_{u \chi(u u' 0) 0 0},\label{eq:bulk2}\\
%&& L'_{v' s s'}X'_{s s' 0 0}X'_{0000}R'_{000} = 
%L_{v' \chi(s' s s') s'}X_{\chi(v' s s') s' \chi(s' 0 0) 0}X_{\chi(s' 0 0) 000}\frac{R_{000}+R_{001}}{2}, \nonumber \\
%&& L'_{000}X'_{0 0 s s'} R'_{s s' u} = 
%L_{000}\sum_{t',t} P^{\rm R}_{(s',u),(t',t)} X_{00 \chi(0 s t') t'} R_{\chi(0 s t') t' t}, \nonumber\\
%&&L_{v' s s'} X_{s s' 0 0}R_{000} =
% \sum_{t',t} P^{\rm L}_{(v',s),(t',t)} L'_{t' t \chi(t s' 0)} X'_{t \chi(t s' 0) 0 0}R'_{000},\nonumber\\
%&&L_{000}X_{0000}X_{0 0 s s'} R_{s s' u} = 
%\frac{L'_{000}+L'_{100}}{2}X'_{000 \chi(0 0 s)}X'_{0 \chi(0 0 s) s \chi(s s' u)} R'_{s \chi(s s' u) u}\, . \label{beqs}
%\end{eqnarray}
The total number of $2\times 16+4\times 8 - 4=60$ unknowns can be further reduced by exploiting the following {\em gauge symmetry} 
\begin{eqnarray}
    \eqalign{
 X_{s s' t t'}  \longrightarrow f_{s s'} X_{s s' t t'} f^{-1}_{t t'}, \\
L_{s' t t'}  \longrightarrow L_{s' t t'} f^{-1}_{t t'},  \\
    R_{s s' t}  \longrightarrow f_{s s'} R_{s s' t},}\qquad
\eqalign{
X'_{s s' t t'}  \longrightarrow g_{s s'} X'_{s s' t t'} g^{-1}_{t t'}, \\
L'_{s' t t'}  \longrightarrow L'_{s' t t'} g^{-1}_{t t'},  \\
    R'_{s s' t}  \longrightarrow g_{s s'} R'_{s s' t},}
    \label{gauge}
\end{eqnarray}
which conserves the patch ansatz (\ref{eq:patch}), as well as the defining equations (\ref{eq:bulk1and2},\ref{beqs}) for arbitrary nonzero gauge `fields' $f_{s s'}, g_{s s'}$. 

While a detailed analysis and explicit expressions for the PSA tensors in terms of boundary parameters can be found in references \cite{prosen2016integrability,inoue2018two}, here we only discuss a generic solution form of the bulk part of the equations. Specifically, the bulk equations \eqref{eq:bulk1and2} can be solved independently (before incorporating boundary conditions), and the solution can be parametrised uniquely -- up to a choice of gauge (\ref{gauge}) --
in terms of two free (spectral) parameters $\xi,\omega$
\begin{eqnarray}\fl
X_{s s' t t'} = T_{(ss'),(tt')}(\xi,\omega),\quad
X'_{s s' t t'} = T_{(ss'),(tt')}(\omega,\xi),\quad
    T(\xi,\omega)\!=\!
\begin{bmatrix}
    1 & 1 & \xi & 1 \\
    \xi  \omega  & \xi  \omega  & 1 & \omega  \\
    \omega  &  \omega  & \xi  \omega  & \omega \\
    \xi  & \xi  & \xi & \xi  \omega
\end{bmatrix}\! .
\label{eq:tm}
\end{eqnarray}
The boundary equations \eqref{beqs} then fix the spectral parameters $\xi,\omega$ as functions of the boundary driving parameters $\alpha,\beta,\gamma,\delta$ (and $\xi,\eta$) [see Section~\ref{sect:MPA}]. The remarkable fact is that the solutions do not explicitly depend on the system size $n$, hence it is clear that such a NESS can only describe ballistic transport (i.e.\ net soliton current independent of $n$).  

\subsection{Conserved charges}\label{subsec:conservedCharges}
The $4\times 4$ matrix $T(\xi,\omega)$, defined in Eq.~(\ref{eq:tm}), can be interpreted as a two-parametric transfer matrix generating the local charges of the RCA54 model. To simplify the discussion, let us assume periodic boundary conditions, $n+1\equiv 1$, and consider a steady state vector $\vec{p}(\xi,\omega)$ whose components are given in terms of the transfer matrix as,
%we assume periodic boundary conditions in this section, $n+1\equiv 1$. We can then write a $2^n$ dimensional vector $\mathbf{p}(\xi,\omega)$:
\begin{eqnarray}\fl \label{eq:GibbsStatePSA}
p_{s_1,s_2,\ldots,s_n}(\xi,\omega) = 
T_{(s_1,s_2),(s_3,s_4)}(\xi,\omega)
T_{(s_3,s_4),(s_5,s_6)}(\xi,\omega)\cdots
T_{(s_{n-1},s_n),(s_1,s_2)}(\xi,\omega).
\end{eqnarray}
For any pair of positive real parameters $\xi,\omega > 0$,
$\mathbf{p}(\xi,\omega)$ represents the statistical ensemble -- state that is invariant under time translation (as well as translationally invariant in space)
\begin{eqnarray}
    \Ue \vec{p}(\xi,\omega) = 
    \mathbf{p}(\omega,\xi),\mkern30mu
     \Uo \mathbf{p}(\omega,\xi) = 
    \mathbf{p}(\xi,\omega),\mkern30mu
     \U  \mathbf{p}(\xi,\omega) = 
    \mathbf{p}(\xi,\omega),
\end{eqnarray}
whose partition sum is given simply in terms of powers of the transfer matrix as
\begin{eqnarray}
 \sum_{\un{s}} p_{\un{s}}(\xi,\omega) = \tr [T(\xi,\omega)]^{n/2}.
\end{eqnarray}
Parameters $\log\xi$, $\log\omega$ can be interpreted as chemical potentials corresponding to the conserved numbers of left and right movers, $N_{-}$ and $N_{+}$,
\begin{eqnarray}
  p_{\un{s}}(\xi,\omega) \propto \xi^{N_{-}(\un{s})}
  \omega^{N_{+}(\un{s})}.
\end{eqnarray}
Specifically, the latter can be generated directly from the logarithmic derivatives of the transfer matrix
\begin{eqnarray}\label{eq:numParticles}
    \eqalign{
    N_{\rm +}(\un{s}) = \frac{{\rm d}}{{\rm d}\log \xi} 
    \log \prod_{j=1}^{n/2}
    T_{(s_{2j-1},s_{2j}),(s_{2j+1},s_{2j+2})}(\xi,\omega)
    \Big\vert_{\xi=\omega=1},\\
    N_{\rm -}(\un{s}) = \frac{{\rm d}}{{\rm d}\log \omega} 
    \log \prod_{j=1}^{n/2} T_{(s_{2j-1},s_{2j}),(s_{2j+1},s_{2j+2})}(\xi,\omega)\Big\vert_{\xi=\omega=1},}
\end{eqnarray}
which can clearly be written as extensive sums of local densities
\begin{eqnarray}\label{eq:localDensities}
    N_{\pm} = \sum_{j=1}^{n/2} n^j_{\pm}, \qquad
    \eqalign{
    n^j_{+}(\un{s}) =
    \frac{{\rm d}}{{\rm d}\xi}
    T_{(s_{2j-1},s_{2j}),(s_{2j+1},s_{2j+2})}
    \Big\vert_{\xi=\omega=1}, \\
    n^j_{-}(\un{s}) =
    \frac{{\rm d}}{{\rm d}\omega}
    T_{(s_{2j-1},s_{2j}),(s_{2j+1},s_{2j+2})}
    \Big\vert_{\xi=\omega=1}.}
\end{eqnarray}
These are the elementary local charges of RCA54 and will be discussed later in the context of hydrodynamics of the model. We note that $N_{\pm}$ by no means represent a complete set of local charges. Specifically, searching for all local charges with densities of support size $\ell > 3$, we find of the order of Fibonacci number $F_\ell=F_{\ell-1}+F_{\ell-2}$ of trivially conserved charges corresponding to $F_\ell$ non-interacting unidirectional soliton configurations. See the discussion in~\cite{klobas2020exactPhD} for the details.

\section{Matrix product ansatz and Markovian excitations} \label{sect:MPA}
In this section we will recast and generalise the exact solution for the NESS in terms of
a more standard \emph{matrix product ansatz} (MPA).
%using the concept of the auxiliary space (representation space of objects in terms of products of which we encode the probabilities $p_{\un{s}}$.
From now on we fix the following convention: notation for row/column vectors in the auxiliary space $\mathcal{V}_a$ will be Dirac bras/kets. Quantities that are vectors in physical space will be denoted by bold-face Roman letters as before, e.g.\ the probability state vector $\mathbf{p}(t)$. The numeral subscript of an operator or vector in the physical space is the site position in the tensor product physical space $(\mathbb{R}^2)^{\otimes n}$ on which it acts nontrivially. The physical space components will be labelled by a binary index and written in corresponding non-bold font, e.g.\ $p_{s_1,\ldots,s_n}$. Matrices (operators that act nontrivially either in auxiliary or physical space) will be denoted  by capital Roman letters.

We are interested in solving the eigenvalue problem for the Markov propagator,
\begin{eqnarray}
\U \mathbf{p}=\Lambda \mathbf{p}, \label{eigeneq}
\end{eqnarray}
which can be conveniently split into two coupled linear equations for the even and odd half-steps of the period of time propagation (generalising \eqref{eq:fp}),
\be
    \Ue \vec{p}=\Lambda_{\rm{R}} \vec{p'},
    \qquad \Uo \vec{p'}=\Lambda_{\rm{L}} \vec{p},
    \qquad \Lambda=\Lambda_{\rm{R}} \Lambda_{\rm{L}}. \label{eigensplit}
\ee
Once we solve this, we have access to the full dynamics of the probabilities at time $t$ in terms of eigenvalues $\Lambda_j$ ($\Lambda_0=1$) and the corresponding eigenvectors $\mathbf{p}_j$ (assuming that $\U$ is not defective),
\be
\mathbf{p}(t) = c_0 \mathbf{p}_0 + \sum_{j=1}^{2^n-1} c_j \Lambda_j^t \mathbf{p}_j,
    \label{timeeigvecs}
\ee
where $c_j$ are constants given by the initial probability distribution $\mathbf{p}(0)$. Using this labelling, the eigenvector $\vec{p}_0$ denotes the NESS as it does not decay in time, while the components along the rest of the eigenvectors $\vec{p}_j$, $j\ge 1$, decay exponentially with the rate $-\log\left|\Lambda_j\right|$. If all multiplicative rates are bounded away from the unit circle $\left|\Lambda_j\right|<1$, any initial state $\vec{p}(0)$ will relax to the NESS. In principle, one may find eigenvalues that lie on the unit circle, but are not $1$. These would correspond to persistently oscillating eigenvectors analogous to the ones for quantum Markov processes \cite{Buca2019nonstationary,Baumgartner_2008,Albert,wolf2012quantum}.
In our case, however, their existence is prohibited by the Perron-Frobenius theorem (see the discussion in Subsection~\ref{subsec:holErg}).

\subsection{NESS: Matrix Product Ansatz}\label{sec:NESS_mpa}
Following \cite{buca2019exact,prosen2017exact} we take a \emph{staggered} matrix product ansatz (MPA) for the NESS in the form,
\begin{eqnarray}
\mathbf{p}=\bra*{\mathbf{l}_1}\mmu{W}{2}\mmup{W}{3} \mmu{W}{4} \mmup{W}{5}\cdots\mmup{W}{n-3}\mmu{W}{n-2}\ket{\mathbf{r}_{n-1,n}}, \label{pness} \\
\mathbf{p'}=\bra*{\mathbf{l}'_{12}}\mmu{W}{3}\mmup{W}4\mmu{W}{5}\mmup{W}{6}\cdots\mmup{W}{n-2}\mmu{W}{n-1}\ket{\mathbf{r}'_{n}}. \label{ppness}
\end{eqnarray}
We will use this as an ansatz to solve the split eigenvalue equation \eqref{eigensplit}. We require that the operators $\mmu{W}{s},\mmup{W}{s}$ and some operator $S$ satisfy the following \emph{bulk} algebraic relation,
\begin{eqnarray}
U_{123} \mathbf{W}_1 S\, \mathbf{W}_2 \mathbf{W}'_3=\mathbf{W}_1 \mathbf{W}'_2 \mathbf{W}_3 S.  \label{bulk1W} 
\end{eqnarray}
Component-wise the above vector equation reads 
\be
W_{s} S W_{\chi(s,s',s'')} W'_{s''} = W_s W'_{s'} W_{s''} S,
\quad s,s',s''\in\{0,1\}.
\ee
A representation of this algebra can be found on a 3-dimensional auxiliary vector space $\mathcal{V}_a =\mathbb{C}^3$ in terms of matrices 
\begin{eqnarray}\fl\label{eq:matWWp}
  W_0(\xi,\omega) = 
    \begin{bmatrix}
 1 & 0 & 0 \\
 \xi  & 0 & 0 \\
 1 & 0 & 0 \\
\end{bmatrix},
 \qquad 
  W_1(\xi,\omega) =
    \begin{bmatrix}
 0 & \xi  & 0 \\
 0 & 0 & 1 \\
 0 & 0 & \omega  \\
\end{bmatrix},\qquad
 S= \begin{bmatrix}
 1 & 0 & 0 \\
 0 & 0 & 1 \\
 0 & 1 & 0 \\
\end{bmatrix},
\end{eqnarray}
while the other set of auxiliary space matrices $\V_s(\xi,\omega)$ is obtained from
$\W_s(\xi,\omega)$ by swapping the parameters,
\begin{eqnarray}
    \V_s(\xi,\omega)=\W_{s}(\omega,\xi),\quad s\in\{0,1\}.
\end{eqnarray}
The representation is parametrised in terms of two (so far) free variables $\xi$ and $\omega$, which (as we argue later) correspond precisely to the variables parametrising the transfer matrix introduced in~\eqref{eq:tm}, therefore we shall refer to them as \emph{spectral parameters}.  We note that $S^2=\one$, which together with $U^2=\one$ and the bijection between the two sets of matrices implies a \emph{dual} bulk relation,
\begin{eqnarray}
U_{123} \mathbf{W}'_1 \mathbf{W}_2 \mathbf{W}'_3 S = \mathbf{W}'_1 S\, \mathbf{W}'_2 \mathbf{W}_3.  \label{bulk1bW}
\end{eqnarray}

We now turn to the boundaries. We demand that the spectral parameters $\xi,\omega$, eigenvalue parameters $\Lambda_{\rm{L}},\Lambda_{\rm{R}}$ and the boundary vectors $\bra*{\mathbf{l}_1}$, $\bra*{\mathbf{l}'_{12}}$, $\ket{\mathbf{r}_{12}}$, $\ket{\mathbf{r}'_1}$, satisfy the following set of equations,
\begin{eqnarray}
U^{\rm{L}}_{12} \bra*{\mathbf{l}'_{12}}=\Lambda_{\rm{L}} \bra*{\mathbf{l}_1} \mmu{W}{2} S, \label{bound1a} \\
U_{123} \bra*{\mathbf{l}_{1}} \mmu{W}{2} \mmup{W}{3}=\bra*{\mathbf{l}'_{12}} \mmu{W}{3}S,\label{bound2a} \\
    U^{R}_{12} \ket{\mathbf{r}_{12}}=\Lambda_{\rm{R}}\mmup{W}{1}S\ket{\mathbf{r}'_2}, \label{bound3a} \\
    U_{123} \mmup{W}{1} \mmu{W}{2}\ket{\mathbf{r}'_{3}} = \mmup{W}{1}S\ket{\mathbf{r}_{23}}. \label{bound4a}
%U^{R}_{12} \ket{\mathbf{r}_{12}}=\mmup{W}{1}S\ket{\mathbf{r}'_2}, \label{bound3a} \\
%U_{123} \mmup{W}{1} \mmu{W}{2}\ket{\mathbf{r}'_{3}} =\Lambda_{\rm{R}} \mmup{W}{1}S\ket{\mathbf{r}_{23}}. \label{bound4a}
\end{eqnarray}
If these equations are satisfied, the pair of states $\vec{p}$, $\vec{p}^{\prime}$ written in terms of MPAs~\eqref{pness} and \eqref{ppness}, solves the staggered eigenvalue equations~\eqref{eigensplit}.
This can be straightforwardly verified by plugging the MPAs into e.g.\ the first
equation of the set~\eqref{eigensplit}, and apply the appropriate relations to transform $\vec{p}$ into $\vec{p}^{\prime}$,
\begin{eqnarray}\fl
    \eqalign{
        \Ue \vec{p} &= 
    U_{1 2 3}\cdots 
    U_{n-5\,n-4\,n-3}U_{n-3\,n-2\,n-1}
    U^{\mathrm{R}}_{n-1\,n}
    \mel{\vec{l}_{1}}{\vW_2 \vV_3 \cdots \vW_{n-4}\vV_{n-3}\vW_{n-2}}{\vec{r}_{n-1\,n}}\\
    &=\Lambda_{\mathrm{R}} U_{1 2 3}\cdots 
    U_{n-5\,n-4\,n-3}
    U_{n-3\,n-2\,n-1}
    \mel{\vec{l}_{1}}{\vW_2 \vV_3 \cdots \vW_{n-4}\vV_{n-3}\vW_{n-2}\vV_{n-1}S}{\vec{r}^{\prime}_{n}}\\
    &=\Lambda_{\mathrm{R}}U_{1 2 3}\cdots 
    U_{n-5\,n-4\,n-3}
    \mel{\vec{l}_{1}}{\vW_2 \vV_3 \cdots \vW_{n-4}\vV_{n-3}S\vV_{n-2}\vW_{n-1}}{\vec{r}^{\prime}_{n}}\\
    &=\cdots = \Lambda_{\mathrm{R}}
    U_{1 2 3} \mel{\vec{l}_1}{\vW_2 \vV_3 S \vV_4 \vW_5\cdots \vW_{n-1}}{\vec{r}^{\prime}_n}\\
    &= \Lambda_{\mathrm{R}}
    \mel{\vec{l}^{\prime}_{12}}{\vW_3 \vV_4 \vW_5\cdots \vW_{n-1}}{\vec{r}^{\prime}_n}
    =\Lambda_{\mathrm{R}}\vec{p}^{\prime}.
}
\end{eqnarray}
To get to the second line, we apply the boundary propagator $U^{\mathrm{R}}_{n-1\,n}$, which, according to relation~\eqref{bound3a}, replaces the right boundary vector $\ket*{\vec{r}_{n-1\,n}^{\prime}}$ with $\ket*{\vec{r}_{n}}$, and at the same time introduces the delimiter matrix $S$.  We continue by applying the dual bulk relation~\eqref{bulk1bW}, which moves $S$ to the left, while the roles of $\vV_s$ and $\vW_s$ are exchanged (fourth line). We finish by absorbing the matrix $S$ at the left and finally change $\bra*{\vec{l}_{1}}$ into~$\bra*{\vec{l}^{\prime}_{12}}$. The second part of~\eqref{eigensplit} is proven analogously.

Solving the left boundary equations (\ref{bound1a},\ref{bound2a})
for $\xi,\omega$ and boundary vectors $\bra{l_s},\bra{l'_{ss'}}$ (given explicitly in \ref{app:boundaryvectorsMPS}) we obtain
\begin{eqnarray}\label{nessspc1}
\xi=\frac{(\alpha +\beta -1)-\beta  \Lambda_{\rm{L}}}{(\beta -1) \Lambda_{L}^{2}},\qquad
\omega =\frac{\Lambda_{\rm{L}} (\alpha -\Lambda_{\rm{L}})}{(\beta -1)}.
\end{eqnarray}
Similarly, solving the right boundary equations
(\ref{bound3a},\ref{bound4a}) for $\xi,\omega$ and the
$\ket{r'_s},\ket{r_{ss'}}$ (again listed in \ref{app:boundaryvectorsMPS})
we get
\begin{eqnarray}\label{nessspc2}
&\xi=\frac{\Lambda_{\rm{R}} (\gamma -\Lambda_{\rm{R}})}{(\delta -1) },\qquad
&\omega =\frac{ (\gamma +\delta -1)-\delta  \Lambda_{\rm{R}}}{(\delta -1) \Lambda_{\rm{R}}^2}.
\end{eqnarray}
These and all explicit result below are written for the specific case of conditional driving 
(\ref{eq:conditional}), while results for the Bernoulli family are obtained fully analogously \cite{prosen2017exact}.
Equating $\xi$ and $\omega$ from these pairs of equations yields an eigenvalue equation, which expressed in terms of the eigenvalue $\Lambda=\Lambda_{\rm{L}} \Lambda_{\rm{R}}$ reads
\begin{eqnarray}\fl\label{char1}
(\Lambda-1)
    \Big(\Lambda ^3+(1-\alpha\gamma)\Lambda ^2 +
    \beta\delta\,\Lambda  - (\alpha +\beta -1) (\gamma +\delta -1)\left(\Lambda+1\right)
    \Big)=0. 
\end{eqnarray}
%\begin{eqnarray}\fl\label{char1}
%&(\Lambda-1)\times \\ \fl
%&\Big(\Lambda ^3+\Lambda ^2 (1-\alpha  \gamma )+\Lambda  [\beta  \delta -(\alpha +\beta -1) (\ga%mma +\delta -1)]-(\alpha +\beta -1) (\gamma +\delta -1)\Big)=0. \nonumber
%\end{eqnarray}
Importantly, $\Lambda=1$ is always a solution. The corresponding eigenvector is the NESS. The remainder is a cubic polynomial giving three other eigenvalues corresponding to three decay modes. These eigenvalues do not depend on the system size $n$. We refer to these four eigenvalues as the \emph{NESS orbital}.

\subsection{Markovian excitations}

\label{sect:markexc}

In order to generalise the results from the previous subsection to other eigenvalues and eigenvectors (decay modes) of the Markov matrix $\U$ we will follow the spirit of the standard coordinate Bethe ansatz (see e.g.\ \cite{bethe1}) for finding solutions to eigenstates containing interacting particles. However, we need to adapt it to a coordinate \emph{matrix product} ansatz picture similar to the one used for the asymmetric simple exclusion processes \cite{ASEPMatrix}. As we will see the leading decay modes can be understood in terms of localized particle excitations of the NESS, the latter being analogous to the vacuum state. The leading decay modes will thus be a superposition of single-particle excitations. 

The leading decay mode is the eigenvector of $\U$ corresponding to the eigenvalue $\Lambda_1$ with the real part closest to $1$,
\begin{eqnarray}
    \min_{j\ge 1}\big(1-\left|\Re(\Lambda_j)\right|\big)
    = 1-\left|\Re(\Lambda_1)\right|.
\end{eqnarray}
Apart from the eigenvectors with eigenvalues on the unit circle, it is the most dynamically relevant eigenvector, since it dominates the long-time asymptotic relaxation. In this subsection, we will provide a compact MPA representation of the eigenvectors corresponding to the orbital that contains the leading decay mode. The generalisation to other orbitals remains an open question, although the complete set of eigenvalues can be predicted by a simple conjecture.

The starting point is to double the auxiliary space, which now consists of two copies of the original auxiliary space ${\mathcal V}'_a = {\mathcal V}_a\oplus{\mathcal V}_a = \mathbb C^2 \otimes {\mathcal V}_a$. We generalise the $\vec{W}$ operator to the doubled auxiliary space $\mathcal{V}^{\prime}_a$ as,
\begin{eqnarray}
    \eqalign{
\hat{W}_s= e_{11}\otimes W_s(\xi z, \omega/z)+ e_{22}\otimes W_s(\xi/ z, \omega z), \\
\hat{W}'_s=e_{11}\otimes W'_s(\xi z, \omega/z)+ e_{22}\otimes W'_s(\xi/ z, \omega z),
    \label{diagpart}
}
\end{eqnarray}
where $e_{i j} = \ketbra{i}{j}$, $i,j\in \{1,2\}$ is the Weyl basis of $2\times 2$ matrices.  We define a set of operators
$\hat{F}^{(k)}$, $\hat{F}^{\prime(k)}$, $\hat{G}^{(k)}$, $\hat{G}^{\prime(k)}$, $\hat{K}^{(k)}$, $\hat{L}^{\prime(k)}$, acting on ${\mathcal V}'_a$,
\begin{eqnarray}\fl
    \eqalign{
    \eqalign{
\hat{F}^{(k)}=\one+ e_{1 2} \otimes \frac{c_+ z^k F_+ +c_- z^{-k} F_-}{\xi \omega -1},\qquad
&\hat{F}'^{(k)}=\one+ e_{1 2} \otimes \frac{c_+ z^k F'_+ +c_- z^{-k} F'_-}{\xi \omega -1},\\
\hat{G}^{(k)}=\one+ e_{1 2} \otimes \frac{c_+ z^k G_+ +c_- z^{-k} G_-}{\xi \omega -1},
&\hat{G}'^{(k)}=\one+ e_{1 2} \otimes \frac{c_+ z^k G'_+ +c_- z^{-k} G'_-}{\xi \omega -1},\\
    \hat{K}^{(k)}=\one+ e_{1 2} \otimes \frac{c_+ z^k K_+ +c_- z^{-k} K_-}{\xi \omega -1},\\
}\\
    \hat{L}^{(k)}=(z e_{11}+z^{-1} e_{22})\otimes \one+ e_{1 2} \otimes \frac{c_+ z^k L_+ +c_- z^{-k} L_-}{\xi \omega -1},}
    \label{FLhat} 
\end{eqnarray} 
 where we introduced the operators $F^{(\prime)}_{\pm}$, $G^{(\prime)}_{\pm}$, $K_{\pm}$, $L_{\pm}$ that act on ${\mathcal V}_a$, i.e. only on a \emph{single} copy of the original auxiliary space.  The explicit representation of these operators is given in Ref.~\cite{prosen2017exact} in terms of a four dimensional auxiliary space ${\mathcal V}_a=\mathbb C^4$.  Note that unlike the MPA solution corresponding to the NESS orbital, we have so far been unable to find a $3$-dimensional representation of auxiliary space operators, so a $4$-dimensional single-copy auxiliary space is assumed here. The operators~\eqref{FLhat} linearly depend on two complex coefficients $c_+,c_-$, and in the inhomogeneous MPA (introduced below), the complex variable $z$ plays the role of the multiplicative quasi-momentum of the quasi-particle excitation on top of the NESS, which are created by the operators~\eqref{FLhat}.

We now introduce the following inhomogeneous (site-dependent) MPA
\begin{eqnarray}\fl
 \mathbf{p}=\bra*{\mathbf{\hat{l}}_1}\hat{L}^{(0)}\WW'_2\hat{S}\hat{F}^{(1)}\mathbf{\WW}'_3\hat{G}^{(2)}\mathbf{\WW}_4 \cdots 
\hat{F}^{(n-5)}\mathbf{\WW}'_{n-3}\hat{G}^{(n-4)}\mathbf{\WW}_{n-2}\ket{\mathbf{\hat{r}}_{n-1,n}}, \label{p} \\ \fl
  \mathbf{p'}=\bra*{\mathbf{\hat{l}}'_{12}}\hat{F}^{'(1)}\mathbf{\WW}_3\hat{G}^{'(2)}\mathbf{\WW}'_4\hat{F}^{'(3)}\mathbf{\WW}_5 \cdots \hat{G}^{'(n-4)}\mathbf{\WW}'_{n-2}\hat{K}^{(n-3)}\mathbf{\WW}_{n-1}\hat{S}\ket{\mathbf{\hat{r}}'_{n}}, \label{pp}
\end{eqnarray}
and $\mathbf{p}$ is the eigenvector of $\U=\Ue\Uo$ (\ref{Ue},\ref{Uo}) with eigenvalue $\Lambda=\Lambda_{\rm L}\Lambda_{\rm R}$, like before. 
The interpretation of $z$ as the momentum of a single excitation should now be clear. Due to the presence of the off-diagonal nilpotent $e_{12}$ in the site dependent matrices $\hat{F}^{(k)},\hat{G}^{(k)}$ these can only create at most one excitation on the NESS in the MPA \eqref{p} and \eqref{pp}. Furthermore, since the site dependent matrices contain both terms with $z$ and $z^{-1}$, this MPA is like a standing wave.
We deform the bulk relation and its dual, (\ref{bulk1W},\ref{bulk1bW}), to an inhomogeneous form -- where a free integer $k$ denotes a spatial coordinate -- which our newly introduced operators have to satisfy
\begin{eqnarray}
    \eqalign{
U_{123}  \hat{K}^{(k-1)}\mathbf{\hat{W}}_1 \hat{S} \hat{G}^{(k)}\mathbf{\hat{W}}_2 \hat{F}^{(k+1)}\mathbf{\hat{W}}'_3 =\hat{F}'^{(k-1)}\mathbf{\hat{W}}_1 \hat{G}'^{(k)} \mathbf{\hat{W}}'_2  \hat{K}^{(k+1)}\mathbf{\hat{W}}_3 \hat{S}, \\
    U_{123} \hat{G}'^{(k-1)}\mathbf{\hat{W}}'_1 \hat{F}'^{(k)} \mathbf{\hat{W}}_2  \hat{L}^{(k+1)}\mathbf{\hat{W}}'_3 \hat{S}=\hat{L}^{(k-1)}\mathbf{\hat{W}}'_1 \hat{S} \hat{F}^{(k)}\mathbf{\hat{W}}'_2 \hat{G}^{(k+1)}\mathbf{\hat{W}}_3, \label{bulk2}}
\end{eqnarray}
where $\hat{S}=\one \otimes S$. 
Analogously as before we demand the boundary vectors and the other parameters to provide a nontrivial solution to a set of generalised boundary equations,
\begin{eqnarray}
&U^{\rm{L}}_{12} \bra*{\mathbf{\hat{l}}'_{12}}=\Lambda_{\rm{L}} \bra*{\mathbf{\hat{l}}_1} \hmp{G}0 \hmup{W}{2}, \label{bound1} \\
&U_{123} \bra*{\mathbf{\hat{l}}_{1}} \hm{L}{0}\hmup{W}{2}\hat{S} \hm{F}{1}\hmup{W}{3}=\bra*{\mathbf{\hat{l}}'_{12}} \hm{K}{1}\hmu{W}{3}\hat{S},\label{bound2} \\
&U^{R}_{12} \ket{\mathbf{\hat{r}}_{12}}=\hm{F}{n-3}\hmup{W}{1}\hat{S}\ket{\mathbf{\hat{r}}'_2}, \label{bound3} \\
&U_{123} \hm{G}{n-4}\hmup{W}{1} \hm{K}{n-3}\hmu{W}{2}\hat{S} \ket{\mathbf{\hat{r}}'_{3}} =\Lambda_{\rm{R}} \hm{L}{n-4}\hmup{W}{1}\hat{S}\ket{\mathbf{\hat{r}}_{23}}. \label{bound4}
\end{eqnarray}

To understand why $\vec{p}$ is indeed an eigenvector, let us for clarity consider an example with fixed $n=8$ spin sites,
\begin{eqnarray}\nonumber \fl
\Ue \mathbf{p}=
 U_{123}U_{345}U_{567} U^{\rm{R}}_{78} 
\bra*{\mathbf{\hat{l}}_1}\hat{L}^{(0)}\WW'_2\hat{S}\hat{F}^{(1)}\mathbf{\WW}'_3\hat{G}^{(2)}\mathbf{\WW}_4 \hat{F}^{(3)}\mathbf{\WW}'_{5}\hat{G}^{(4)}\mathbf{\WW}_{6}\ket{\mathbf{\hat{r}}_{78}}=\ldots
\end{eqnarray}
We start by first acting with $U_{123}$ (note that all the local operators commute, cf.~\eqref{boundcomm}, and can be applied in any order),
%The first operator in $\Ue$ that acts is $U_{123}$ and it does so on the left boundary vector,
which using \eqref{bound2} creates $\hm{K}{1}\hmu{W}{3}\hat{S}$ giving,
\begin{eqnarray}\nonumber
    &\Ue\vec{p}= U_{345}U_{567} U^{\rm{R}}_{78}
\bra*{\mathbf{\hat{l}}'_{12}} \hm{K}{1}\hmu{W}{3}\hat{S}\hat{G}^{(2)}\mathbf{\WW}_4 \hat{F}^{(3)}\mathbf{\WW}'_{5}\hat{G}^{(4)}\mathbf{\WW}_{6}\ket{\mathbf{\hat{r}}_{78}}.
\end{eqnarray}
The subsequent $U$'s in $\Ue$ transfer $\hm{K}{1}\hmu{W}{3}\hat{S}$ to the penultimate right hand side via the bulk relation \eqref{bulk2}, 
\begin{eqnarray}\nonumber
\ldots= U_{567} U^{\rm{R}}_{78}\bra*{\mathbf{\hat{l}}'_{12}} \hat{F}'^{(1)}\mathbf{\hat{W}}_1 \hat{G}'^{(2)} \mathbf{\hat{W}}'_2  \hat{K}^{(3)}\mathbf{\hat{W}}_3 \hat{S}\hat{G}^{(4)}\mathbf{\WW}_{6}\ket{\mathbf{\hat{r}}_{78}}=\ldots
\end{eqnarray}
Prior to acting with the final $U_{n-3\,n-2\,n-1}$, we apply $U^{\rm{R}}$ so it creates $\hm{F}{n-3}\hmup{W}{n-1}\hat{S}$ at the right end,
\begin{eqnarray} \nonumber
\ldots= U_{567} 
\bra*{\mathbf{\hat{l}}'_{12}} \hat{F}'^{(1)}\mathbf{\hat{W}}_1 \hat{G}'^{(2)} \mathbf{\hat{W}}'_2  \hat{K}^{(3)}\mathbf{\hat{W}}_3 \hat{S}\hat{G}^{(4)}\mathbf{\WW}_{6}\hm{F}{5}\hmup{W}{7}\hat{S}\ket{\mathbf{\hat{r}}'_8}=\ldots
\end{eqnarray}
We then use the top Eq.~\eqref{bulk2} for the final $U_{n-3,n-2,n-1}$ to move $\hm{K}{n-5}\hmu{W}{n-3}\hat{S}$ to the far end into $\hm{K}{n-3}\hmu{W}{n-1}\hat{S}$ and we thus finally obtain the form \eqref{pp} of $\mathbf{p'}$,
\begin{eqnarray}
&\ldots=  \bra*{\mathbf{\hat{l}}'_{12}} \hat{F}'^{(1)}\mathbf{\hat{W}}_1 \hat{G}'^{(2)} \mathbf{\hat{W}}'_2 \hat{F}'^{(3)}\mathbf{\hat{W}}_1 \hat{G}'^{(4)} \mathbf{\hat{W}}'_2  \hat{K}^{(5)}\mathbf{\hat{W}}_3 \hat{S}^2\ket{\mathbf{\hat{r}}'_8}= \mathbf{p'}.
    \label{eq:exampleInh}
\end{eqnarray}
The odd-part of the propagator~$\Uo$ acts analogously in reverse, hence we obtain the second 
part of the eigenvalue equation,
\be
\Uo\Ue \vec{p}=\Uo\vec{p}^{\prime}=\Lambda_L\Lambda_R\vec{p}.
\ee

Note that for $z=1$, $c_+=c_-=0$, one recovers the NESS bulk relations (\ref{bulk1W},\ref{bulk1bW}). The boundary equations given by~(\ref{bound1}-\ref{bound4}) can be solved to obtain the parameters $\xi,\omega,z,c_+,c_-\in\CC$, together with the eigenvectors and eigenvalues of the first (leading decay-mode) orbital. Explicitly, from the left boundary equations we obtain
\begin{eqnarray}\fl\label{paraleft}
\xi=\frac{z (\alpha +\beta -1)-\beta  \Lambda_{\rm{L}}}{(\beta -1) \Lambda_{L}^{2}},\qquad
\omega =\frac{\Lambda_{\rm{L}} (\alpha  z-\Lambda_{\rm{L}})}{(\beta -1) z},\qquad
\frac{c_-}{c_+}=\frac{\Lambda_{L}^{4}}{z^4 (\alpha +\beta -1)},
\end{eqnarray}
and from the right ones,
 \begin{eqnarray}\fl\label{pararight}
\xi=\frac{\Lambda_{\rm{R}} (\gamma  z-\Lambda_{\rm{R}})}{(\delta -1) z},\qquad
\omega =\frac{z (\gamma +\delta -1)-\delta  \Lambda_{\rm{R}}}{(\delta -1) \Lambda_{\rm{R}}^2},
\qquad
\frac{c_+}{c_-}=\frac{\Lambda_{\rm{R}}^{4} z^{4 m+2}}{\gamma +\delta -1},
\end{eqnarray}
where $m=\frac{n}{2}-2$. Pairwise identifying expressions for $\xi$, $\omega$ and $c_+/c_-$ from the left and right boundary equations, we obtain a triple of coupled equations for the unknowns $\Lambda_{\rm L},\Lambda_{\rm R}$ and $z$.
For the details of the solution procedure and the form of the boundary vectors we refer the reader to \cite{prosen2017exact}. Using this solution it can be shown that the eigenvalue of the leading decay mode
$\Lambda_1$ scales with large $n$ as 
$\Lambda_1 = 1 - \Lambda'_1/n + {\mathcal O}(1/n^2)$, with
\begin{eqnarray} 
\Lambda'_1=
\frac{[1-(\alpha +\beta -1) (\gamma +\delta -1)] \log [(\alpha +\beta -1) (\gamma+\delta -1)]}{2 (\alpha +\beta -1) (\gamma +\delta -1)+\alpha  \gamma -\beta  \delta-2}. \label{ldmth}
\end{eqnarray}
The $1/n$ scaling of the gap implies that the relaxation time is proportional to $n$, which is consistent with the ballistic transport observed in~\cite{prosen2016integrability}.
\begin{figure}
 \centering	
\vspace{-1mm}
\includegraphics[width=0.8\columnwidth]{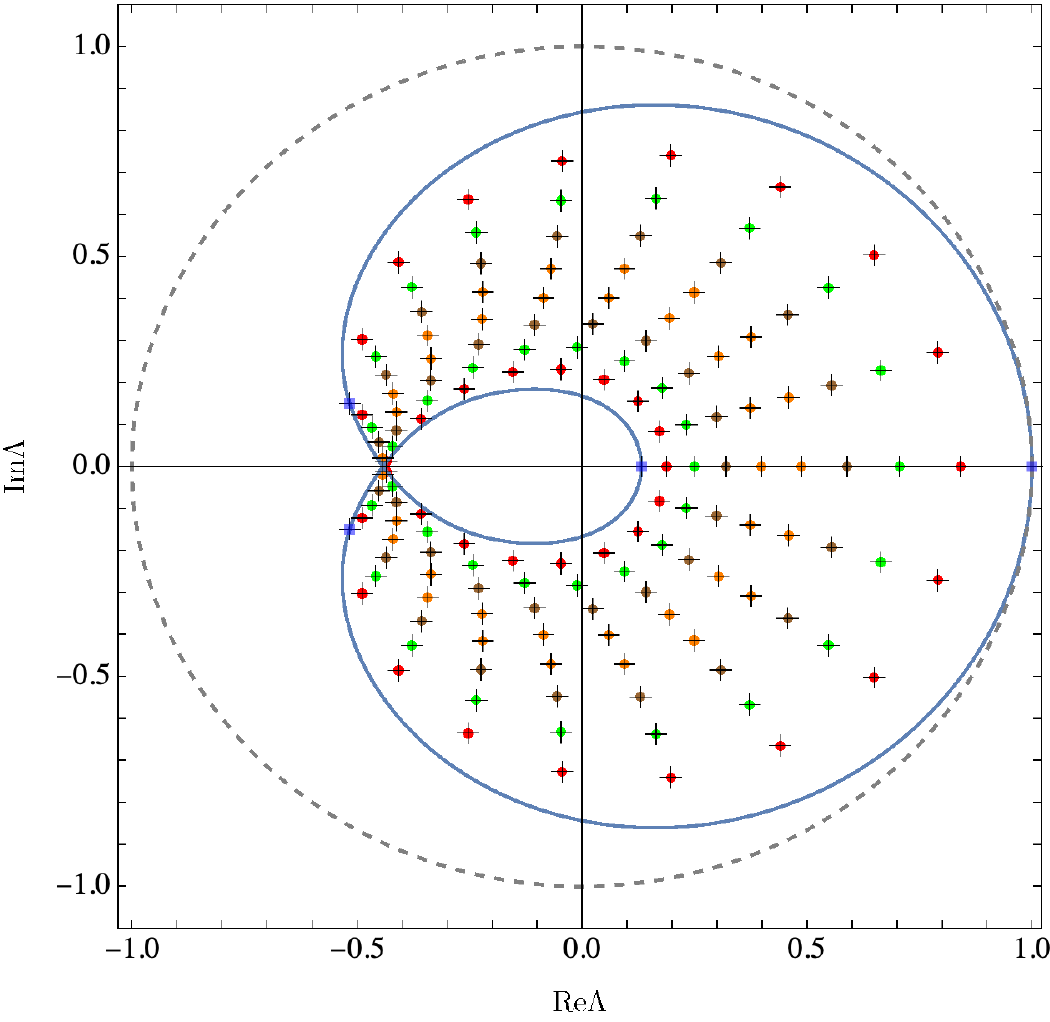}
\vspace{-1mm}
\caption{The spectrum of $\U$ for $n=12$ and $\alpha=1/2,\beta=1/3,\gamma=1/5,\delta=4/7$. The black crosses show numerical results. The red ($p=1$), green ($p=2$), brown ($p=3$), orange ($p=4$) points are solutions for the $p$-th orbital eigenvalues~(\ref{conj}). The blue squares are the roots of characteristic polynomial for the NESS orbital. Note that eigenvalues have to be largely degenerate on average as the number of distinct eigenvalues is $4+18\times 8 = 148 \propto n^2$, while the total number of eigenvalues is $4096= 2^n$. The dashed grey curve is the unit circle. The blue curve is an algebraic curve to which the leading decay mode orbital converges in the thermodynamic limit.}
\label{fig:orbitals}
\end{figure}

Based on these results one can formulate a conjecture for the higher orbital eigenvalues~\cite{prosen2017exact,bucaprosenreview}. Let $p$ denote the orbital number, $1\le p \le m=n/2-2$. Then the \emph{complete} set of eigenvalues $\Lambda=\Lambda_{\mathrm{L}}\Lambda_{\mathrm{R}}$ satisfy the following set of Bethe-like equations,
\begin{eqnarray}
    \eqalign{
\frac{z (\alpha +\beta -1)-\beta  \Lambda_{\rm{L}}}{(\beta -1) \Lambda_{L}^{2}}=\frac{\Lambda_{\rm{R}} (\gamma  z-\Lambda_{\rm{R}})}{(\delta -1) z},\\
 \frac{\Lambda_{\rm{L}} (\alpha  z-\Lambda_{\rm{L}})}{(\beta -1) z}=\frac{z (\gamma +\delta -1)-\delta  \Lambda_{\rm{R}}}{(\delta -1) \Lambda_{\rm{R}}^2}, \\
    (\alpha +\beta -1)^p(\gamma +\delta -1)^p=(\Lambda_{\rm{L}} \Lambda_{\rm{R}})^{4 p} z^{4(m-p)+2},}
    \label{conj}
\end{eqnarray}
where $p=1$ corresponds to the first orbital obtained analytically via Eqs.~(\ref{paraleft}--\ref{pararight}). 
The eigenvalue solutions of \eqref{conj} are shown in Fig.~\ref{fig:orbitals} and agree perfectly with numerical diagonalization of the Markov matrix $\U$. In contrast to the usual Bethe ansatz \cite{bethe1931theorie}, this conjecture implies that we have only one momentum parameter ($z$) even for many-particle excitations. This is consistent with the Bethe-ansatz description of Rule 54 proposed in~\cite{friedman2019integrable},
and can intuitively be understood
%Intuitively, this could be understood
as a consequence of a very singular dispersion relation, namely all quasi-particles (solitons) move with the same speed $\pm 1$. Another many-body effect is a large degeneracy of each eigenvalue  in the orbital $p$ which is conjectured~\cite{prosen2017exact} to be exponential, specifically  $2^{p-1}$.

Expanding the last equation \eqref{conj} for $p=1$ in $1/m$ (with $m=\frac{n}{2}-2$) in the leading order of $1/n$, we get that $z$ becomes unimodular, $z=e^{\ii \kappa}$, where the momentum $\kappa$ is restricted to the interval $[0, \pi)$. Note that if some eigenvalues $\Lambda_L \Lambda_R$ in the $n\to\infty$ limit converged to points on the unit circle different from $1$,
%this would indicate that in the thermodynamic limit the model may have persistent oscillations and/or non-decaying fluctuations in time.
this would imply persistent oscillations and/or non-decaying fluctuations in time for the model in the thermodynamic limit. 
Although one may naively expect that the effect of the boundaries becomes thermodynamically irrelevant, pushing all eigenvalues to collapse onto the unit circle (similarly to the closed model studied in Sections~\ref{sect:hydro}--\ref{sect:timeMPA}), this is not the case. Rather, we have eigenvalues collapsing onto an algebraic curve (Fig.~\ref{fig:orbitals}), which is \emph{inside} the unit disk and only tangential to a unit circle at $\kappa=0$ ($z=1$). This indicates that the effect of the boundaries is thermodynamically nontrivial. Note that this implies that the thermodynamic limit and the long-time limit do not commute. More specifically, the stationary state is the only asymptotic state for any finite system size $n$, but with increasing $n$ the time it takes to reach it diverges.

\section{Exact large deviations}
\label{sect:largedev}
The Markovian matrix $\U$ describes the ensemble averaged dynamics of a stochastic process. In order to gain further insight into the detailed dynamics of the individual realizations of this stochastic process (also referred to as \emph{trajectories}) we need to go beyond the ensemble averaged properties. Fortunately there exists an elegant framework~\cite{Touchette,Hedges2009,lecomte2007thermodynamic}
for calculating the fluctuations of \emph{time-integrated} local (or extensive) observables,
\begin{eqnarray}\fl
 \cO_T &= \sum_{t=0}^{T-1} \sum_{x=1}^{n-1}\Big[
     f_x(s^{2t+\mathrm{mod}(x,2)}_x,s^{2t+\mathrm{mod}(x+1,2)}_{x+1})+
     g_x(s^{2t+2-\mathrm{mod}(x,2)}_x,s^{2t+2-\mathrm{mod}(x+1,2)}_{x+1})\Big],
 %f_{2x}(s^{2t}_{2x},s^{2t+1}_{2x+1})
 %+f_{2x+1}(s^{2t+1}_{2x+1},s^{2t}_{2x+2})
 %+g_{2x}(s^{2t+2}_{2x},s^{2t+1}_{2x+1})\\
 %&+g_{2x+1}(s^{2t+1}_{2x+1},s^{2t+2}_{2x+2})\Big],}
 \label{integratedobs}
\end{eqnarray}
%\begin{eqnarray}\fl
%\eqalign{
 %\cO_T &= \sum_{t=0}^{T-1} \sum_{x=1}^{\frac{n}{2}-1}\Big[
 %f_{2x}(s^{2t}_{2x},s^{2t+1}_{2x+1})
 %+f_{2x+1}(s^{2t+1}_{2x+1},s^{2t}_{2x+2})
 %+g_{2x}(s^{2t+2}_{2x},s^{2t+1}_{2x+1})\\
 %&+g_{2x+1}(s^{2t+1}_{2x+1},s^{2t+2}_{2x+2})\Big],}
 %\label{integratedobs}
%\end{eqnarray}
%\begin{eqnarray}
% \cO_T = \sum_{t=0}^{T-1} \sum_{x=1}^{n-1}\Big[f_x(s^t_x,s^t_{x+1})+g_x(s^{t+1}_x,s^{t+1}_{x+1})\Big],
% \label{integratedobs}
%\end{eqnarray}
in the large time $T$ limit.
The ansatz above encodes arbitrary two-site observables that are local (and extensive) both in space and time and are parametrised by functions $f_x$ and $g_x$.
Note again that the dynamics is defined on a diamond sublattice $x+t=0\pmod{2}$ (see the discussion around Eq.~\eqref{eq:defGeometry1}).
This framework is referred to as \emph{large deviations} (LD), also known as \emph{full counting statistics}. Before we define it, we just emphasise that the observables of the type \eqref{integratedobs} depend on the \emph{full} trajectory realization
%$(\bm{s}^0,\bm{s}^{1},\bm{s}^2,\dots,\bm{s}^{2T-1})$,
$(\ul{s}^0,\ul{s}^{1},\ul{s}^2,\ldots,\ul{s}^{2T-1})$,
where $\ul{s}^t$
is an $n$-point configuration on the $t$-th space-time saw.
For instance, one choice may be the time integrated number of occupied cells  $f_x(s,s')=(s+s')/2$, $g_x(s,s')=0$. 

It can be shown~\cite{Touchette} that for $T \to \infty$ the probability distribution of $\cO_T$ has a large deviation form
\be
P_T(\cO) = \ave{\delta( \cO - \cO_T)} \sim_{T \to \infty} e^{-T \varphi_n(\cO_T / T)},
\ee
where $\varphi_n(x)$ is called the {\em rate function}. 
This implies that the moment generating function also has the large deviation form 
\be
Z_T(\cf) = \ave{e^{-\cf \cO_T}} \sim e^{T \theta_n(\cf)}.
\ee
The derivatives of the {\em scaled cumulant generating function} (SCGF) $\theta_n(\cf)$ at $\cf=0$ correspond to the cumulants of $\cO_T$ divided by time. 
The LD functions may be understood intuitively as free-energies of the trajectories and are related to rate functions by a Legendre transform, $\theta_n(\cf) = - \min_x \left[ \cf x + \varphi_n(x) \right]$. 

In order to compute the SCGF one can show~\cite{Touchette,buca2019exact} that we can deform the propagator $\U$ with a \emph{tilt} or \emph{counting field} $\cf$,
\be
\U(\cf) = \Uo \, G(\cf) \, \Ue \, F(\cf),
\label{eq:tiltU} 
\ee
where we introduced diagonal matrices
\be
F_{\underline{s},\underline{s}'}(\cf) = \delta_{\underline{s},\underline{s}'} \prod_{x=1}^{n-1} f_{s_x,s_{x+1}}^{(x)},\qquad
G_{\underline{s},\underline{s}'}(\cf) = \delta_{\underline{s},\underline{s}'} \prod_{x=1}^{n-1} g_{s_x,s_{x+1}}^{(x)},
\ee
with 
\be \label{eq:tiltOps}
f_{s,s'}^{(x)}=e^{-\cf f_x(s,s')}\quad\text{and}\quad  g_{s,s'}^{(x)}=e^{-\cf g_x(s,s')}.
\ee
We then have that $\theta_n(\cf) = \log \lambda(\cf)$,   
where $\lambda(\cf)$ is the eigenvalue of $\U(\cf)$ with the largest real part. Therefore, by finding this eigenvalue we have access to $P_T(\cO)$ and its cumulants. 

The quantity $\cO_T$ explicitly depends on the various values of the trajectory at all times, which makes its exact evaluation a formidable task. 
%At first sight, the explicit dependence of the quantity $\cO_T$ on the various values of the trajectory at all times.
Remarkably, we can find an exact solution for its statistics in the long-time limit for \emph{arbitrary} two-site observables $f_x(s,s')$, $g_x(s,s')$. We again take an ansatz similar to the one we used for the NESS in Subsection~ \ref{sec:NESS_mpa}. We write an ansatz for the components of the eigenvector of~\eqref{eq:tiltU} as,
\begin{eqnarray}
    \eqalign{
        p_{s_1,\dots,s_n} =  \mel*{l_{s_1}}{W_{s_2}^{(2)}W_{s_3}^{(3)} \cdots W_{s_{n-3}}^{(n-3)}W_{s_{n-2}}^{(n-2)}}{r_{s_{n-1}s_n}},\\
    p'_{s_1,\dots,s_n} =  \mel*{l'_{s_1 s_2}}{V_{s_3}^{(3)}V_{s_4}^{(4)} \cdots V_{s_{n-2}}^{(n-2)}V_{s_{n-1}}^{(n-1)}}{r'_{s_n}}.
    }
\end{eqnarray}
The matrices in the MPA are now site-dependent and we demand that they satisfy the following inhomogeneous bulk algebra, that is generalised to include the tilt operators~\eqref{eq:tiltOps}, written for \emph{even} $x$:
\begin{eqnarray}
    \eqalign{
f^{(x-1)}_{ss'} f^{(x)}_{s's'\!'} W^{(x-1)}_s W^{(x)}_{s'} Z^{(x+1)}_{s'\!'} = 
Z^{(x-1)}_{s} V^{(x)}_{\chi(ss's'\!')} V^{(x+1)}_{s'\!'}, \\
    g^{(x-2)}_{ss'} g^{(x-1)}_{s's'\!'} Z^{(x-2)}_s V^{(x-1)}_{s'} V^{(x)}_{s'\!'} = W^{(x-2)}_s W^{(x-1)}_{\chi(s,s',s'\!')} Z^{(x)}_{s'\!'},}
\end{eqnarray}
where $Z^{(x)}_{s'\!'}$ are site dependent matrices generalising the delimiter matrix $S$ from the homogeneous bulk algebra~\eqref{eq:matWWp}. As before, these matrices admit a 3-dimensional representation,
\begin{eqnarray}\fl
    \eqalign{
    W_0^{(x)} = \begin{bmatrix}
               1 & 0 & 0 \\
               w_1^{(x)} & 0 & 0 \\
               1 & 0 & 0
              \end{bmatrix},\\
    W_1^{(x)} = \begin{bmatrix}
               0 & w_2^{(x)} & 0 \\
               0 & 0 & 1 \\
               0 & 0 & w_3^{(x)}
    \end{bmatrix},}
    \quad
    \eqalign{
    V_0^{(x)} = \begin{bmatrix}
               1 & 0 & 0 \\
               v_1^{(x)} & 0 & 0 \\
               1 & 0 & 0
              \end{bmatrix}, \\
    V_1^{(x)} = \begin{bmatrix}
               0 & v_2^{(x)} & 0 \\
               0 & 0 & 1 \\
               0 & 0 & v_3^{(x)}
    \end{bmatrix},}\quad
    \eqalign{
        Z_0^{(x)} = \begin{bmatrix}
               z_1^{(x)} & 0 & 0 \\
               z_2^{(x)} & 0 & 0 \\
               z_3^{(x)} & 0 & 0
              \end{bmatrix},\\
    Z_1^{(x)} = \begin{bmatrix}
               0 & 0 & z_4^{(x)} \\
               0 & z_5^{(x)} & 0 \\
               0 & z_6^{(x)} & 0
  \end{bmatrix},}
  \label{DevMat}
 \end{eqnarray}
with the precise values of the coefficients $v_j^{(x)},w_j^{(x)},z_j^{(x)}$ given in~\ref{app:Reps}.  We impose the following set of boundary equations for boundary vectors,
\begin{eqnarray}\label{eq:boundRelsLDs}
    \eqalign{
    f^{(1)}_{ss'}f^{(2)}_{s's'\!'} \bra*{l_s}W^{(2)}_{s'} Z^{(3)}_{s'\!'} = \bra*{l'_{s\chi(s,s',s'\!')}}V^{(3)}_{s'\!'},\\
    g^{(n-2)}_{ss'}\!g^{(n-1)}_{s's'\!'} Z^{(n-2)}_{s} V^{(n-1)}_{s'}\!\ket{r'_{s'\!'}} = W^{(n-2)}_{s}\!\ket{r_{\!\chi(s,s',s'\!')s'\!'}},\\
    \sum_{m,m'=0,1} \big[U^{\rm R}\big]_{ss'}^{mm'}f^{(n-1)}_{mm'}\ket{r_{mm'}} = \lambda_{\rm R} Z^{(n-1)}_{s}\ket{r'_{s'}}, \\
    \sum_{m,m'=0,1} \big[U^{\rm L}\big]_{ss'}^{mm'} g^{(1)}_{mm'} \bra*{l'_{mm'}} = \lambda_{\rm L} \bra*{l_s}Z^{(2)}_{s'},
    }
\end{eqnarray}
where the sub-(super-)scripts on $U^{\rm R,L}$ denote the row (column) of the corresponding matrix elements. If these boundary equations are satisfied, the MPAs solve the following eigenvector conditions
\begin{eqnarray}
    \Ue F(\cf) \bm{p} = \lambda_{R}(\cf) \bm{p'},\qquad
    \Uo G(\cf) \bm{p'} = \lambda_{L}(\cf) \bm{p}.
\end{eqnarray}
We then have that $\lambda(\cf) = \lambda_{R}(\cf) \lambda_{L}(\cf)$ is the dominant eigenvalue giving the SCGF. For simplicity we focus on conditional boundary driving~\eqref{eq:conditional}.
The solution to boundary equations~\eqref{eq:boundRelsLDs} consists of boundary vectors (given in~\cite{buca2019exact}) and a consistency condition for the NESS orbital, which takes the form of a quartic polynomial,
%The solutions to the boundary vectors are given in~\cite{buca2019exact}. The consistency condition satisfied by the solution for the NESS orbital gives a polynomial of degree $4$, 
\begin{eqnarray}
\lambda^4 - \alpha \gamma \cA_n \lambda^3 - \omega \cA_n^2 \lambda^2 - \beta\delta\xi \cA_n^3 \lambda + \eta \cA_n^4=0,
\label{lambdacf}
\end{eqnarray}
with  
\begin{eqnarray}\label{eq:LDsolSpectrPar}
    \fl
    \eqalign{
\omega = \cB_n(1-\alpha)(1-\delta)\beta'\gamma'+\cC_n(1-\beta)(1-\gamma)\alpha'\delta',\qquad
    \eta = (\alpha\beta-\widetilde\alpha\widetilde\beta)(\gamma\delta-\widetilde\gamma\widetilde\delta)\xi,\\
    \xi = 
\cB_n\cC_n(1-\alpha)(1-\beta)(1-\gamma)(1-\delta)\frac{\alpha'\beta'\gamma'\delta'}{\widetilde{\alpha}\widetilde{\beta}\widetilde{\gamma}\widetilde{\delta}},
}
\end{eqnarray}
and where
\begin{eqnarray}
    \eqalign{
  \cA_n = \prod_{x=1}^{n-1} (f_{00}^{(x)}g_{00}^{(x)}), \\
  \cB_n = \prod_{x=1}^{n/2} \frac{f_{01}^{(2x-1)}f_{10}^{(2x-1)}g_{11}^{(2x-1)}}{\big(f_{00}^{(2x-1)}\big)^2 g_{00}^{(2x-1)}}
  \prod_{x=1}^{n/2-1} \frac{g_{01}^{(2x)}g_{10}^{(2x)}f_{11}^{(2x)}}{\big(g_{00}^{(2x)}\big)^2 f_{00}^{(2x)}},\\
  \cC_n = \prod_{x=1}^{n/2} \frac{g_{01}^{(2x-1)}g_{10}^{(2x-1)}f_{11}^{(2x-1)}}{\big(g_{00}^{(2x-1)}\big)^2 f_{00}^{(2x-1)}}
  \prod_{x=1}^{n/2-1} \frac{f_{01}^{(2x)}f_{10}^{(2x)}g_{11}^{(2j)}}{\big(f_{00}^{(2x)}\big)^2 g_{00}^{(2x)}},
    }
\end{eqnarray}
and, for $\alpha'=\alpha+\widetilde{\alpha}$, $\beta'=\beta+\widetilde{\beta}$, $\gamma'=\gamma+\widetilde{\gamma}$, $\delta'=\delta+\widetilde{\delta}$,
 \begin{eqnarray}\label{eq:LDsolPartilde}
     \eqalign{
 \widetilde\alpha = \frac{f_{10}^{(1)}g_{11}^{(1)}}{f_{00}^{(1)}g_{01}^{(1)}}(1-\alpha), \qquad 
     &\widetilde\beta = \frac{f_{11}^{(1)}g_{10}^{(1)}}{f_{01}^{(1)}g_{00}^{(1)}}(1-\beta), \\
\widetilde\gamma = \frac{f_{11}^{(n-1)}g_{01}^{(n-1)}}{f_{10}^{(n-1)}g_{00}^{(n-1)}}(1-\gamma), \qquad 
     &\widetilde\delta = \frac{f_{01}^{(n-1)}g_{11}^{(n-1)}}{f_{00}^{(n-1)}g_{10}^{(n-1)}}(1-\delta).}
 \end{eqnarray}
Crucially, when $\sigma=0$ the eigenvalue equation \eqref{lambdacf} reduces to the eigenvalue equation for the NESS orbital \eqref{char1}, as expected. 

\begin{figure}
\vspace{0mm}
\begin{center}
\includegraphics[width=0.8\columnwidth]{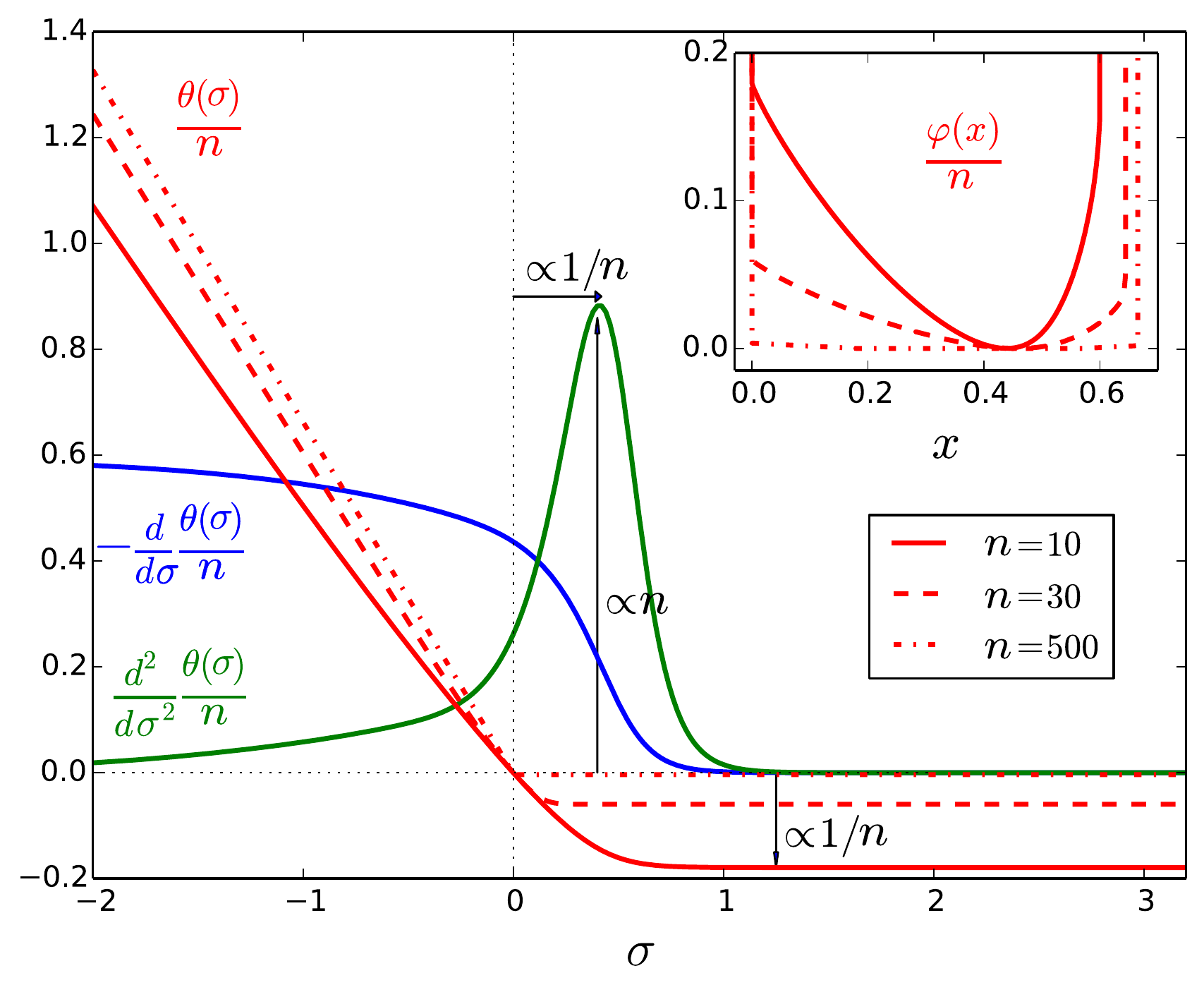}
\end{center}
\vspace{-5mm}
\caption{
The red curves show SCGF $\theta(\sigma)$ 
for the time-integrated number of occupied cells (i.e. $f_{x}(s,s')=\frac{1}{2}(s+s')$ and $g_{x}=0$) from \eqref{integratedobs} for $n=10,30,500$. It approaches the asymptotic form \eqref{theta} in the thermodynamic limit. The blue curve is the order parameter $\lim_{T \to \infty} \langle \cO_T e^{-\sigma \cO_T} \rangle / (T n Z_T(\sigma)) = -\theta'(\sigma)/n$, which has a first-order phase transition between phases $\sigma<0$ and $\sigma>0$. The green susceptibility curve $\theta''(\sigma)/n$ diverges as $\propto n$. The inset is the rate function $\varphi(x)$ with $x=\cO_T / (T n)$. The conditional boundary driving parameters are $(\alpha,\beta,\gamma,\delta)=(1/3,1/8,1/2,2/5)$.
}
\label{fig:cumulants}
\end{figure} 
 
Equation~\eqref{lambdacf} can be solved exactly for arbitrary observables and very large system sizes $n$. However, we find remarkable universality in the scaling of all observables in the thermodynamic limit $n\to\infty$.  As the observables \eqref{integratedobs} are extensive in the system size $n$, we have that the limits
\be 
a:=-\lim_{n \to \infty}\frac{\log \cA_n}{n \cf}, \quad
b:=-\lim_{n \to \infty}\frac{\log \cB_n}{n \cf}, \quad
c:=-\lim_{n \to \infty}\frac{\log \cC_n}{n \cf},
\ee
exist. The SCGF then takes the following simple form
\begin{eqnarray} \label{eq:scaling}
 \theta_n(\cf) = \cE(n \cf)+\cO\Big(\frac{1}{n}\Big),
\end{eqnarray}
where the function $\cE(\varsigma)$ is such that $\exp[\cE(\varsigma)]$ is the maximum real root of the polynomial 
\begin{eqnarray} \label{eq:poly_scaling}
    \eqalign{
    0 &=  e^{4[\cE(\varsigma)+\varsigma a]}-\alpha\gamma e^{3[\cE(\varsigma)+\varsigma a]} \\
    &- \Big[e^{-\varsigma b} (1-\alpha)(1-\delta)+e^{-\varsigma c}(1-\beta)(1-\gamma)\Big]e^{2[\cE(\varsigma)+\varsigma a]} \\
    &-e^{-(b+c)\varsigma}\beta\delta e^{\cE(\varsigma)+\varsigma a} 
  + e^{-(b+c)\varsigma}(\alpha+\beta-1)(\gamma+\delta-1).
    }
\end{eqnarray}
Remarkably \eqref{eq:scaling} gives the large system size form of long-time cumulants of $\cO_T$
\begin{eqnarray}
 \lim_{T \to \infty} \frac{1}{T} \langle\!\langle \cO_T^k \rangle\!\rangle = \left. (-1)^k \frac{d^k}{d\cf^k} \theta_n \right|_{\cf=0} \propto n^{k},
\end{eqnarray}
where $\langle\!\langle \cdot \rangle\!\rangle$ denotes the cumulant.
Crucially, for $k \geq 2$ we see a divergence for $n \to \infty$ at $\cf=0$. 

We can extract explicitly a \emph{universal} form from \eqref{eq:poly_scaling} of the first few cumulants. From $k=1$ we get the average per unit time, 
\begin{eqnarray}
\lim_{T \to \infty} \frac{1}{Tn} \langle \cO_T \rangle =
 a+\frac{\mu b+ \nu c}{2(\mu+\nu)+\alpha\gamma-\beta \delta}+ \mathcal{O}\Big(\frac{1}{n}\Big),
 \end{eqnarray}
while from $k=2$ we get the corresponding variance per unit time
\begin{eqnarray}\fl
    \eqalign{
\lim_{T \to \infty} \frac{1}{Tn} {\rm var}~\cO_T  
    &= n \bigg[ 
-\frac{2bc(1-\alpha\gamma)+\mu b^2 +\nu c^2}{2(\mu+\nu)+\alpha\gamma-\beta \delta} 
    +\frac{3(\mu b + \nu c)^2}{(2(\mu+\nu)+\alpha\gamma-\beta \delta)^2} \\
    &+\!\frac{2(b+c)(\mu b + \nu c)(2-\alpha\gamma)}{(2(\mu+\nu)+\alpha\gamma-\beta \delta)^2} 
    \!-\!
    \frac{2(\mu b + \nu c)^2(4+\mu+\nu-\alpha\gamma)}{(2(\mu+\nu)+\alpha\gamma-\beta \delta)^3}
 \bigg]
    \!+\!\mathcal{O}(1)},
\end{eqnarray}
where $\mu = \gamma(1-\alpha) + \beta(1-\gamma)$ and $\nu = \delta(1-\alpha) + \alpha(1-\gamma)$. The divergence of the variance (see Fig.~\ref{fig:cumulants}) signals a presence of a \emph{dynamical phase transition} in the statistics of the trajectories. 
 
\begin{figure}
\vspace{3mm}
\begin{center}
\includegraphics[width=0.8\columnwidth]{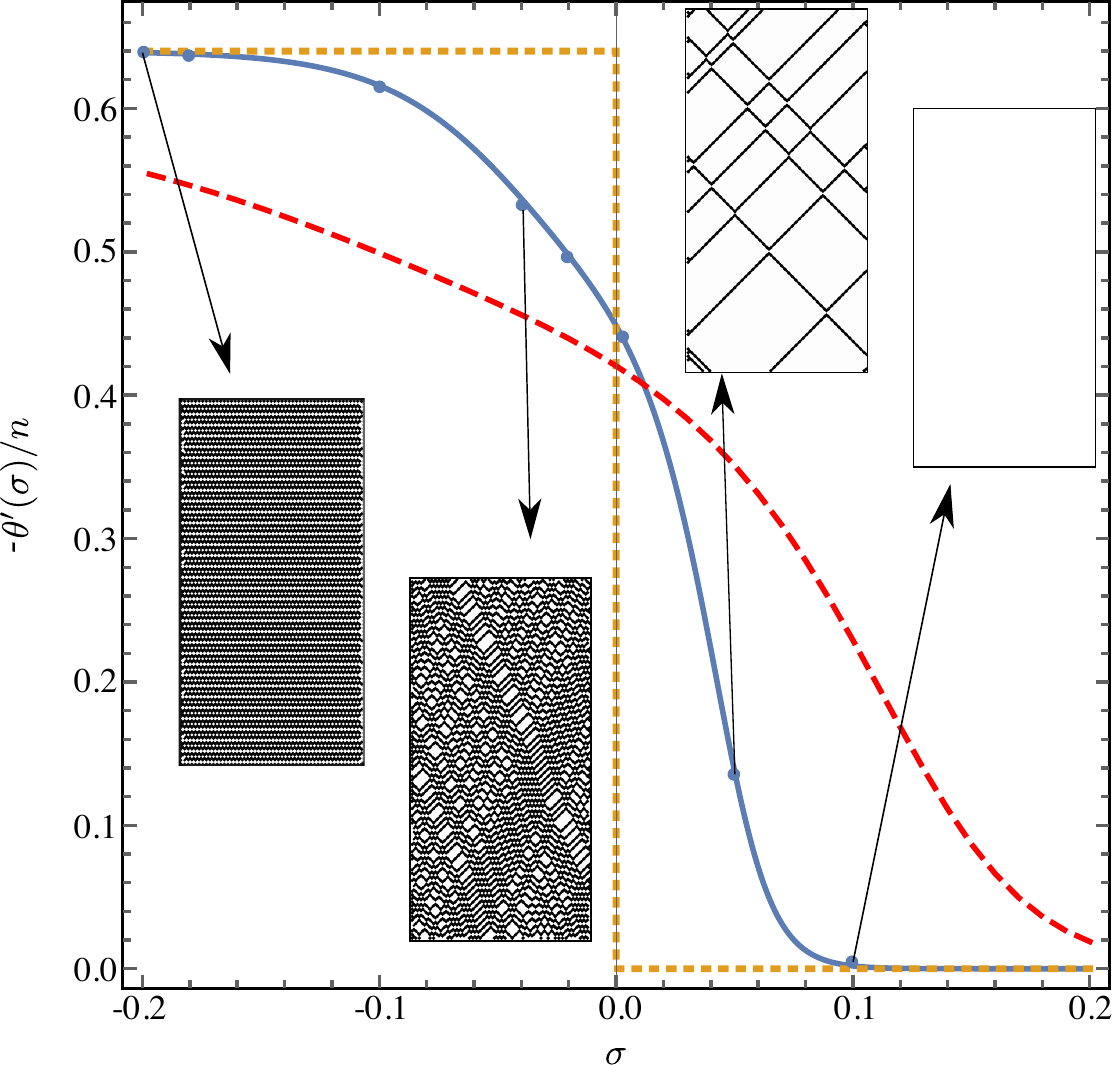}
\end{center}
\vspace{-6mm}
\caption{
Plot of the order parameter $-\theta'(\sigma)/n$ (blue curve) for $N=100$ showing typical active and inactive trajectories for various values of $\sigma$. These are obtained by weighting the probabilities of trajectories using a Doob transform~\cite{Jack2010,chetrite2015nonequilibrium,garrahan2018aspects}.
%\cite{Hedges2009,garrahan2018aspects,Jack2010}.
The red and orange curve correspond to the respective order parameter for $N=40$ and $N\to\infty$. For details see Ref.~\cite{buca2019exact}. The observables and parameters are the same as in Fig.~\ref{fig:cumulants}.
}
\label{fig:active_inactive}
\end{figure}

More specifically the scaling function for $b+c>0$ has the form,
\begin{eqnarray}
\cE(\varsigma) = 
 \begin{cases}
     -a\varsigma + \log(\alpha\gamma) + o(1), \quad &\varsigma \to \infty, \\
     -\frac{1}{3}(b+c+3a)\varsigma + \frac{1}{3}\log(\beta\delta) + o(1), \quad
     &\varsigma \to -\infty.
 \end{cases}
\end{eqnarray}
For $(b+c)<0$ the asymptotic behaviour is obtained by $\varsigma \to -\varsigma$. We deduce that the SCGF converges to,
\begin{eqnarray}
\lim_{n\to\infty}\frac{1}{n}\theta_n(\cf) = 
 \begin{cases}
     -a \cf , \quad &\cf>0, \\
     -\frac{1}{3}(b+c+3a)\cf, \quad &\cf<0,
 \end{cases}
 \label{theta}
\end{eqnarray}
when $b+c>0$. The singularity at $\cf=0$ corresponds to the first-order phase transition. The behaviour should be contrasted with facilitated models~\cite{garrahan2018aspects} as here there are no fluctuations within each dynamical phase. 

In order to demonstrate the phase transition we show in Fig.~\ref{fig:active_inactive}, alongside the order parameter as a function of $\sigma$, typical active and inactive trajectories (on each side of the transition point $\sigma=0$). 
 
\section{Hydrodynamics of Rule 54}
\label{sect:hydro}
Let us now turn away from the boundary driven setup and consider the dynamics
of the system far from the edges, or equivalently, dynamics of an infinite
system defined on $\mathbb Z$ lattice. In this section we start the discussion
of the bulk physics by providing the effective hydrodynamic equations that
govern the dynamics on the large scales. The starting point is the description
of \emph{macrostates}, from which one extracts the information sufficient to
characterize the dynamics on the Euler scale. As an example we provide the
parametrisation of the state after the bipartitioning quench and we finish the
section by discussing diffusive corrections. 

\subsection{Macroscopic description}
We start by introducing $n_{+}$ and $n_{-}$ to denote
densities of left and right movers, given as
\begin{eqnarray}\label{def:hydroDens}
    n_{\nu}=\frac{2 N_{\nu}}{n},
\end{eqnarray}
where $N_{\nu}$ for $\nu\in\{\pm\}$ are the numbers of particles of both types
on the lattice of length $n$ (cf.~\eqref{eq:numParticles}).
%Note that for convenience
%Note that the density is defined
%with an additional factor of $2$, since this provides a more natural scaling of the
%quantities introduced later.
To account for the
entropic contribution to the macrostate,
we introduce $n_{\nu}^{\rm tot}$ as the total density of ``slots'' that can be occupied by particles of
type $\nu\in\{\pm\}$. In the context of thermodynamic Bethe ansatz
(TBA)~\cite{yang1969thermodynamics,takahashi1999thermodynamics} (see~\cite{TBAIntro} for a pedagogical review), $n_{\nu}^{\rm tot}$
corresponds to the \emph{total density} (i.e.\ the sum of the densities of particles and
holes). Classically we interpret it as the density of the effective \emph{free space} in which
the particles with a finite width can move~\cite{doyon2017dynamics,doyon2018soliton}.
As a result of the interactions between particles, the total densities
$n_{\nu}^{\rm tot}$  exhibit a nontrivial dependence on densities $n_{\pm}$ described
by the following relation,
\begin{eqnarray}\label{eq:hydro1}
    n_{\nu}^{\rm tot}=1-n_{\nu}+n_{-\nu},\qquad \nu\in\{+,-\}.
\end{eqnarray}
This identity follows from the TBA description of the quantum generalisation
of the model (see~\cite{friedman2019integrable} for the details). Intuitively, we can motivate
it by realizing that two solitons of the same kind are either separated at least by two sites, or there is an oppositely moving soliton in between them, as is
demonstrated by the following two diagrams,
\begin{eqnarray}
    \begin{tikzpicture}[scale=0.3,baseline={([yshift=-0.5ex]current bounding box.center)}]
        \rectangle{0}{0}{1};
        \rectangle{1}{1}{1};
        \rectangle{2}{0}{0};
        \rectangle{3}{1}{0};
        \rectangle{4}{0}{1};
        \rectangle{5}{1}{1};
        \begin{scope}[shift={(8,0)}]
        \rectangle{0}{0}{1};
        \rectangle{1}{1}{1};
        \rectangle{2}{0}{1};
        \rectangle{3}{1}{1};
        \end{scope}
    \end{tikzpicture}.
\end{eqnarray}
Solitons of the same kind therefore behave as hard rods with a finite length (hence
the term $-n_{\nu}$),
but due to the attractive interaction with the other species (i.e.\ the shift after scattering is $-2$ sites) their total density increases and we obtain the term $+n_{-\nu}$.
This is consistent with the
effective description in terms of the \emph{flea-gas}~\cite{doyon2018soliton},
where the interaction between the particles of the same kind corresponds to the positive
jump, while the scattering of oppositely-moving solitons induces the negative jump.

We restrict the discussion to the simplest class of generalised Gibbs ensembles (GGE), 
that are characterized by two chemical potentials, $\mu_{+}$ and $\mu_{-}$,
corresponding to densities of right and left movers, so that the probability of
a configuration $\ul{s}$ is given by 
\begin{eqnarray}
\label{GGE}
    p^{\text{GGE}}_{\ul{s}} = \frac{1}{Z_{\text{GGE}}}
    \mathrm{e}^{-\mu_{+}N_{+}(\ul{s})-\mu_{-}N_{-}(\ul{s})}.
\end{eqnarray}
This GGE is equivalent to the state~\eqref{eq:GibbsStatePSA} given by the PSA
with periodic boundary conditions (see Section~\ref{subsec:conservedCharges}
for the details), and can be alternatively given in terms of the
MPS~\eqref{pness} if one replaces the left and right boundary vectors with the
trace (see the discussion in Section~\ref{sec:PBmps}).

In the large system size limit, the partition function is expressed in terms of
the following phase integral involving densities of particles,
\begin{eqnarray}
    \lim_{n\to\infty}Z_{\text{GGE}}=
    \lim_{n\to\infty}
    \int\!d n_{+}\, dn_{-}\,
    \mathrm{e}^{-\frac{n}{2}\left(\mu_{+}n_{+}+\mu_{-}n_{-}-s[n_{+},n_{-}]\right)},
\end{eqnarray}
where the configurational entropy density $s[n_{+},n_{-}]$ is given by,
\begin{eqnarray}\label{eq:YYent}
    s[n_{+},n_{-}]=-\sum_{\nu\in\{+,-\}}\left(
    n_{\nu}\log\left(\frac{n_{\nu}}{n_{\nu}^{\rm tot}}\right)
    +(n_{\nu}^{\rm tot}-n_{\nu})\log\left(1-\frac{n_{\nu}}{n_{\nu}^{\rm tot}}\right)
    \right).
\end{eqnarray}
In the TBA description of the quantum model, this form of entropy density is precisely
the Yang-Yang entropy~\cite{friedman2019integrable}. From the classical point of view, 
this form can at first sight be surprising as one could naively expect the second term to vanish~\cite{doyon2017dynamics}.
However, the particles in the model
do not behave exactly as hard rods, which makes the entropy density more complicated.
One can also verify the validity of~\eqref{eq:YYent} independently, without relying
on the TBA description of the quantum model, by studying the partition sum in finite
systems (see~\ref{appsubsec:partitionSum} for details).
In the thermodynamic limit the expectation values of densities are given by evaluating this integral using the saddle-point
method, which gives the following relation between $n_{\nu}$ and $\mu_{\nu}$
\begin{eqnarray}\label{eq:TBAdefEq1}
    \frac{n_{\nu}^{\rm tot}-n_{\nu}}{n_{\nu}}=
    \mathrm{e}^{\epsilon_{\nu}}
    ,\qquad
    \epsilon_{\nu}=\mu_{\nu}
    +\log\frac{1+\mathrm{e}^{-\epsilon_{\nu}}}{1+\mathrm{e}^{-\epsilon_{-\nu}}},
\end{eqnarray}
where $\epsilon_{\nu}$ can be interpreted as a single-particle (quasi-)energy~\cite{yang1969thermodynamics}.
Additionally, for later convenience, we introduce filling functions $\vartheta_{\nu}$ as the
ratios between the density and the total density,
\begin{eqnarray}\label{eq:TBAdefEq2}
    \vartheta_{\nu}=\frac{n_{\nu}}{n_{\nu}^{\rm tot}}.
\end{eqnarray}
Note that the GGE is completely determined either by the pair of chemical potentials
$\{\mu_{+},\mu_{-}\}$, or by the pair of filling functions $\{\vartheta_{+},\vartheta_{-}\}$.

Excitations on top of these macrostates do not simply move with velocities $\pm1$,
but rather they slow down due to the interactions with other particles, and
their velocities get renormalized into
\begin{eqnarray}\label{eq:hydroDressedVelocities}
    v_{\nu}=\nu\left(1-\frac{2 n_{-\nu}}{1+n_{\nu}+n_{-\nu}}\right)=
    \frac{\nu}{1+2\vartheta_{-\nu}}.
\end{eqnarray}
There are many equivalent ways of finding this result: one can derive it from
the definition of the dressed
velocity~\cite{bertini2016transport,castroalvaredo2016emergent} by using the
TBA description put forward in~\cite{friedman2019integrable}; one can obtain it
by using the appropriate result for hard
rods~\cite{doyon2017dynamics,boldrighini1983one}, or an appropriate soliton
gas~\cite{doyon2018soliton,el2003thermodynamic,el2011kinetic,el2021soliton};
or alternatively the dressed velocity can be identified in terms of expectation
values in the GGE~\cite{gopalakrishnan2018hydrodynamics}. Additionally, the
simplicity of the result allows us to obtain it through an elementary argument
introduced in~\cite{gopalakrishnan2018hydrodynamics}: let us imagine that we
are tracing a right-moving soliton that starts at time $t=0$ at position $x=0$
and at time $t$ it arrives to the position $x$. The soliton moves with the bare
velocity $1$ and gets displaced for two sites every time it scatters, therefore
$t-x=2 m_{\textrm{s}}$, where $m_{\textrm{s}}$ is the number of
\emph{left}-movers encountered by the tagged particle. The left-movers are
moving with the dressed velocity $v_{-}$, therefore the tagged soliton will
encounter all the left-movers that were at time $0$ between $0$ and $x-v_{-}t$,
which by identifying $v_{+}=x/t$ gives us,
\begin{eqnarray}
    1-v_{+}=(v_{+}-v_{-})n_{-},\qquad
    1+v_{-}=(v_{+}-v_{-})n_{+}.
\end{eqnarray}
The second relation follows from an analogous procedure involving a tagged
left-mover. Note that an additional factor of $\frac{1}{2}$ comes from the
definition of left/right densities~\eqref{def:hydroDens}. This set of equations
provides exactly the dressed velocities given by~\eqref{eq:hydroDressedVelocities}.

\begin{figure}
    \centering
    \includegraphics[width=0.125\textwidth]{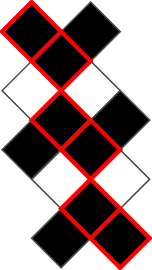}
    \caption{\label{fig:hydrofig1} A trajectory of the slowest right-mover. The red-bordered
    sites denote the right-mover that undergoes the largest possible number of scatterings,
    which slow it down to the velocity $\frac{1}{3}$.}
\end{figure}

For any choice of the parameters $\{\mu_{+},\mu_{-}\}$, the velocities
$v_{\nu}$ are restricted to
\begin{eqnarray}\label{eq:admVels}
    \abs{v_{\nu}}\in\left[\tfrac{1}{3},1\right].
\end{eqnarray}
%interval $\abs{v_{\pm}}\in\left[\tfrac{1}{3},1\right]$.
The velocity cannot be
larger than $1$ (or smaller than $-1$), since the interactions can only slow
the particles down. Alternatively, we can also explain the upper bound by the
fact that the dynamics is given in terms of local mutually-commuting operators
(cf.~\eqref{Ue} and~\eqref{Uo}), which restricts the correlation spreading to
the causal light-cone spreading with the speed $1$ (see the upcoming discussion
in Section~\ref{subsec:circuits}). The lower bound is at first sight less
clear, but can be easily understood when we recall that two consecutive
scatterings of a given soliton should be at least $3$ time-steps apart, and
therefore its effective speed cannot drop to $0$. An example of a configuration
corresponding to the soliton moving with the smallest velocity is shown in
Fig.~\ref{fig:hydrofig1}.

\subsection{Example: Inhomogeneous quench}
Now we have all the necessary ingredients to obtain predictions for the
dynamics on the Euler scale. We consider an example of the bipartitioning
protocol, where at time $t=0$, the left half of the system is prepared in the
state determined by $(\vartheta_{+}^{L},\vartheta_{-}^L)$, while the
state in the right half is $(\vartheta_{+}^{R},\vartheta_{-}^R)$. At long times,
the state is locally described by a GGE parametrised by $\vartheta_{\pm,\zeta}$ that slowly
changes with the ray $\zeta=x/t$,
\begin{eqnarray}
    \lim_{\substack{x,t\to \infty\\ x/t=\zeta}}\vartheta_{\nu}(x,t) =\vartheta_{\nu,\zeta}.
\end{eqnarray}
The limiting values far from the junction are given by the left and right GGEs,
\begin{eqnarray}\label{eq:hydroInhQ1}
    \lim_{\zeta\to-\infty} \vartheta_{\nu,\zeta}= \vartheta_{\nu}^{L},\qquad
    \lim_{\zeta\to\infty} \vartheta_{\nu,\zeta}= \vartheta_{\nu}^{R},
\end{eqnarray}
while for the intermediate values of $\zeta$, the filling functions satisfy
the following consistency condition~\cite{bertini2016transport,castroalvaredo2016emergent},
\begin{eqnarray}
    \vartheta_{\nu,\zeta}=\begin{cases}
        \vartheta_{\nu}^L,& v_{\nu}(\zeta)>\zeta,\\
        \vartheta_{\nu}^R,& v_{\nu}(\zeta)<\zeta.
    \end{cases}
\end{eqnarray}
Here $v_{\nu}(\zeta)=\frac{\nu}{1+2\vartheta_{-\nu,\zeta}}$ denotes the dressed
velocity in the state $(\vartheta_{+,\zeta},\vartheta_{-,\zeta})$, as given by
Equation~\eqref{eq:hydroDressedVelocities}. The consistency condition describes
the profile consisting of two steps, determined by the velocity of left-movers in the state
on the right and the velocity of right-movers on the left,
\begin{eqnarray}\label{eq:inhQpred}
    (\vartheta_{+,\zeta},\vartheta_{-,\zeta})=
    \begin{cases}
        (\vartheta_{+}^L,\vartheta_{-}^L),& \zeta<-\frac{1}{1+2\vartheta_{+}^L},\\
        (\vartheta_{+}^L,\vartheta_{-}^R),& -\frac{1}{1+2\vartheta_{+}^L}<\zeta<
        \frac{1}{1+2\vartheta_{-}^R},\\
        (\vartheta_{+}^R,\vartheta_{-}^R),& \zeta>\frac{1}{1+2\vartheta_{-}^R}.
    \end{cases}
\end{eqnarray}

\subsection{Diffusive corrections}
%Diffusion matrix
%\begin{eqnarray}
%    \mathcal{D}_{\nu_1 \nu_2} =
%    \frac{2\nu_1 \nu_2 \vartheta_{+}\vartheta_{-}(1-\vartheta_{+})(1-\vartheta_{-})}
%    {(1+2\vartheta_{+})(1+2\vartheta_{-})(1+\vartheta_{+}+\vartheta_{-})}
%\end{eqnarray}
Hydrodynamics up to the diffusive scale is described by the Navier-Stokes 
equation for filling functions $\vartheta_{+},\vartheta_{-}$,
\cite{denardis2018hydrodynamic,denardis2019diffusion}
\begin{eqnarray}\label{eq:NavierStokes}
    \partial_{t} \vartheta_{\nu} + v_{\nu}\partial_{x} \vartheta_{\nu} =
    \smashoperator{\sum_{\nu_1}}
    D_{\nu \nu_1}\partial_x^2\vartheta_{\nu_1}
    +\underbrace{\smashoperator{\sum_{\nu_1,\nu_2}}
    A_{\nu\nu_1}\left(\partial_x B_{\nu_1 \nu_2}\right)\partial_x\vartheta_{\nu_2}
    }_{\mathcal{O}\left((\partial_x\vartheta)^2\right)}.
\end{eqnarray}
The simplicity of the TBA for RCA54~\cite{friedman2019integrable} allows us
to find the following explicit form of the coefficients $D_{\nu_1,\nu_2}$,
\begin{eqnarray}\label{eq:defD}
    D_{\nu_1 \nu_2} =
    \frac{2\nu_1 \nu_2 \vartheta_{-\nu_2}(1-\vartheta_{-\nu_2})}
    {(1+2\vartheta_{-\nu_1})^2(1+2\vartheta_{-\nu_2})}.
\end{eqnarray}
Here we obtain an additional factor of $2$ w.r.t.\ the expressions introduced
in~\cite{denardis2019diffusion,gopalakrishnan2018hydrodynamics,friedman2019integrable},
due to the convention used in the definition of particle densities~\eqref{def:hydroDens}.
The diagonal coefficients $D_{\nu,\nu}$ can be understood as the variance of the
position of the particle of type $\nu$ due to the fluctuations in the particle
density in the quasi-stationary state, and agree with the prediction
of~\cite{gopalakrishnan2018hydrodynamics,friedman2019integrable}.

For completeness, we conclude the section by reporting the coefficients~$A_{\nu_1,\nu_2}$,
$B_{\nu_1,\nu_2}$, that characterize the nonlinear term, as given
by~\cite{denardis2019diffusion},
\begin{eqnarray}
    \eqalign{
        A_{\nu_1 \nu_2}=
        \frac{(1+\vartheta_{+}+\vartheta_{-})(\delta_{\nu_1,\nu_2}-\nu_1\nu_2\vartheta_{\nu_1})}
        {1+2\vartheta_{-\nu_1}},\\
        B_{\nu_1 \nu_2}=D_{\nu_1 \nu_2}
        \frac{(1+2\vartheta_{-\nu_1})(1-\vartheta_{-\nu_1}+3\vartheta_{\nu_1})}
        {(1+2\vartheta_{+})(1+2\vartheta_{-})(1-(\vartheta_{-}+\vartheta_{+})^2)}.
    }
\end{eqnarray}
Even thought the expressions are very explicit, we are not able to independently verify
the validity of the quadratic term in the Navier-Stokes equation~\eqref{eq:NavierStokes},
since it is hard to solve it to obtain predictions in this limit.

\section{Space-time duality: time-states and space dynamics}\label{sec:TS}
The starting point of this section are multi-point correlation functions of
local one-site observables at the same spatial coordinate but at different
times, while the system is assumed to be in a stationary state $\vec{p}$.
Such correlation functions can be interpreted
as expectation values of observables that are ultra-local in space, but extended 
in time, and, intuitively, one can imagine to introduce a corresponding probability
vector $\vec{q}$ (to be defined precisely below), that contains expectation values
of all such observables.
The construction of an explicit representation of $\vec{q}$ can be schematically
understood as a problem of keeping track of solitons that pass through the origin,
as shown in Fig.~\ref{fig:TSfig1}. Generically, there is no guarantee that this is feasible.
However, in the case of GGEs introduced in the previous section, the solitons that reach
the central site are uncorrelated, which implies an efficient MPA representation of $\vec{q}$
in terms of $3\times3$ matrices constructed in Subsection~\ref{sec:TimeStates}.

The probability distribution $\vec{q}$ has short-range correlations, therefore
a natural question that arises is whether it can be interpreted as a
fixed-point (steady state) of another dynamical system, which corresponds to exchanging the roles
of space and time in RCA54.  In Subsection~\ref{subsec:spaceEvol} we show that
the evolution in space can indeed be given in terms of local and deterministic
maps, however, their support is  bigger than that of the time-evolution.
Equivalently, space evolution can be formulated as a composition of
non-deterministic $3$-site maps and projectors to a subspace of configurations,
as we show in Subsection~\ref{subsec:circuits}, where we also introduce a
convenient tensor-network representation of dynamics that provides an algebraic
interpretation of the MPA encoding multi-time correlation functions. We finish
the section by a short discussion of the physics of one-site correlations.

\begin{figure}
  \centering
  \includegraphics[width=0.75\textwidth]{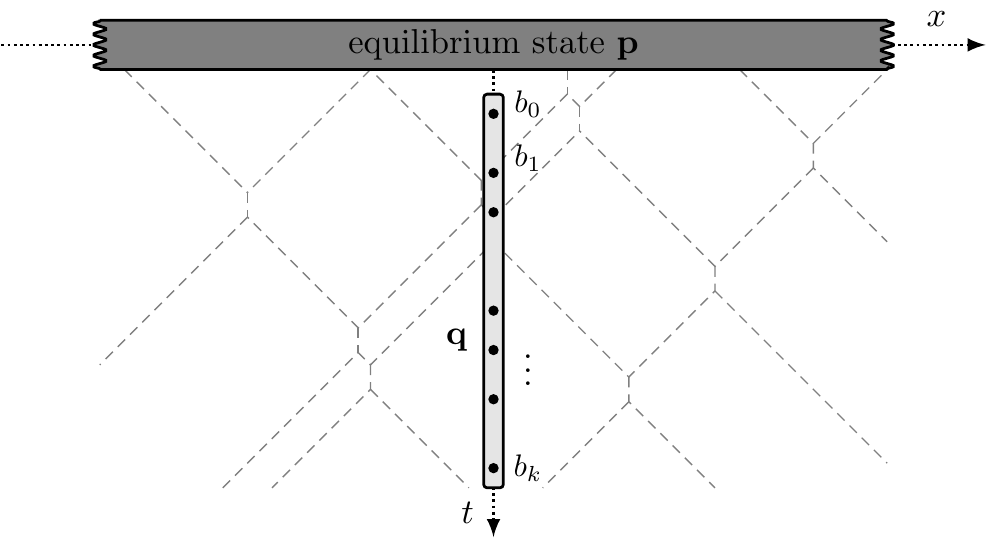}
    \caption{\label{fig:TSfig1} Schematic picture of the setup. We are interested in
    probability distribution $\vec{q}$ of configurations $(b_0,b_1,\ldots,b_k)$ at
    the center of the chain $x=0$ and different times $t$, while the system is initialized
    in the equilibrium state $\vec{p}$. Intuitively, this can be understood as keeping
    track of particles that pass through the origin. Since they are interacting with each
    other, this is generally a nontrivial task and the computational complexity of $\vec{q}$
    might still be exponentially growing with time. In our case, however, the statistical
    independence of solitons implies an efficient MPA representation of $\vec{q}$.
    }
\end{figure}

\subsection{Stationary states of the closed system}\label{sec:PBmps}
We start by recalling that the NESS introduced in Sections~\ref{sect:BD} and~\ref{sect:MPA}
in the absence of boundary driving (and by assuming periodic boundaries) reduces to
a GGE from the previous section -- see the discussion in Section~\ref{subsec:conservedCharges}.
We will therefore use the parametrisation of this class of stationary states in
terms of the matrices~$\vW\!(\xi,\omega)$, $\vV(\xi,\omega)$
introduced in~\eqref{eq:matWWp}. In particular, for a periodic system of \emph{even}
size $n$,   
\begin{eqnarray}
  \vec{p}^{(n)}=\frac{1}{Z_{n}} \tr\left(\vW_1\vV_2\vW_3\cdots \vV_{n}\right),
  \qquad
  Z_{n}=\tr T^{n/2},
\end{eqnarray}
where the partition function~$Z_{n}$ is expressed in terms of the transfer
matrix~$T$
\be
T=\left(\W_0+\W_1\right)\left(\V_0+\V_1\right).
\ee
We recall that the chemical potentials associated to the densities of left and
right movers are given by $\mu_{+}=-\log\xi$ and $\mu_{-}=-\log\omega$.

In this section we are interested in the large system-size limit, therefore we
introduce the following graphical notation for the asymptotic (i.e.\ $n\to\infty$)
probability of finite block configurations,
\begin{eqnarray}\label{eq:asymptPropEven}
    \eqalign{
        p\Big(\!\begin{tikzpicture}
            [scale=0.3,baseline={([yshift=-0.6ex]current bounding box.center)}]
            \textrectangle{0}{-0.5}{$s_1$};
            \textrectangle{1}{0.5}{$s_2$};
            \textrectangle{2}{-0.5}{$\cdots$};
            \textrectangle{3}{0.5}{};
            \textrectangle{4}{-0.5}{};
            \textrectangle{5}{0.5}{};
            \textrectangle{6}{-0.5}{$\cdots$};
            \textrectangle{7}{0.5}{$s_{2k}$};
        \end{tikzpicture}\!\Big)=
        \lim_{n\to\infty} \frac{\tr\left(\W_{s_1}\V_{s_2}\cdots \V_{s_{2k}}T^{n/2-k}\right)}
        {\tr\left(T^{n/2}\right)},\\
        p\Big(\!\begin{tikzpicture}[scale=0.3,baseline={([yshift=-0.6ex]current bounding box.center)}]
            \textrectangle{0}{0.5}{$s_1$};
            \textrectangle{1}{-0.5}{$s_2$};
            \textrectangle{2}{0.5}{$\cdots$};
            \textrectangle{3}{-0.5}{};
            \textrectangle{4}{0.5}{};
            \textrectangle{5}{-0.5}{};
            \textrectangle{6}{0.5}{$\cdots$};
            \textrectangle{7}{-0.5}{$s_{2k}$};
        \end{tikzpicture}\!\Big)=
        \lim_{n\to\infty} \frac{\tr\left(\V_{s_1}\W_{s_2}\cdots \W_{s_{2k}}
        T^{\prime\,n/2-k}\right)}
        {\tr\left(T^{n/2}\right)},}
\end{eqnarray}
where the transfer matrix $T^{\prime}=(\V_0+\V_1)(\W_0+\W_1)$ is obtained from~$T$ by
the exchange of parameters~$\xi$, $\omega$. Analogously, the probabilities of
configurations with odd block length can be obtained by summing over the value of the
last bit,
\begin{eqnarray}\label{eq:asymptPropOdd}
    \eqalign{
        p\Big(\!\begin{tikzpicture}[scale=0.3,baseline={([yshift=-0.6ex]current bounding box.center)}]
            \textrectangle{0}{-0.5}{$s_1$};
            \textrectangle{1}{0.5}{$s_2$};
            \textrectangle{2}{-0.5}{$\cdots$};
            \textrectangle{3}{0.5}{};
            \textrectangle{4}{-0.5}{};
            \textrectangle{5}{0.5}{$\cdots$};
            \textrectangle{6}{-0.5}{\scalebox{0.8}{$s_{2k-1}$}};
        \end{tikzpicture}\!\Big)=
        \lim_{n\to\infty} \frac{\tr\left(\W_{s_1}\V_{s_2}\cdots \W_{s_{2k-1}}
        (\V_{0}+\V_{1})T^{n/2-k}\right)}
        {\tr\left(T^{n/2}\right)},\\
        p\Big(\!\begin{tikzpicture}[scale=0.3,baseline={([yshift=-0.6ex]current bounding box.center)}]
            \textrectangle{0}{0.5}{$s_1$};
            \textrectangle{1}{-0.5}{$s_2$};
            \textrectangle{2}{0.5}{$\cdots$};
            \textrectangle{3}{-0.5}{};
            \textrectangle{4}{0.5}{};
            \textrectangle{5}{-0.5}{$\cdots$};
            \textrectangle{6}{0.5}{\scalebox{0.8}{$s_{2k-1}$}};
        \end{tikzpicture}\!\Big)=
        \lim_{n\to\infty} \frac{\tr\left(\V_{s_1}\W_{s_2}\cdots \V_{s_{2k-1}}
        (\W_0+\W_1)
        T^{\prime\,n/2-k}\right)}
        {\tr\left(T^{n/2}\right)}.}
\end{eqnarray}
These asymptotic probabilities can be formulated in a more explicit and
convenient form by diagonalizing $T$ and $T^{\prime}$. However, the precise
form is not essential for the upcoming discussion, and one can find it
in~\ref{appsubsec:thermodynamicEVs}.

These asymptotic probabilities exhibit a nontrivial factorization property: the
conditional probability of observing $k$-th bit to be in the state~$s_k$, given
the configuration on previous~$k-1$ sites $(s_1,s_2,\ldots,s_{k-1})$  only
depends on the last three bits. Similarly, the conditional probability of~$s_1$ given
the configuration $(s_2,s_3,\ldots, s_{k})$ only depends on $(s_1,s_2,s_3)$. Explicilty,
the following holds,
\begin{eqnarray}\fl
    \label{eq:condProb2}
    \frac{
        p\Big(\!\begin{tikzpicture}[scale=0.3,baseline={([yshift=-0.6ex]current bounding box.center)}]
            \textrectangle{0}{-0.5}{$s_1$};
            \textrectangle{1}{0.5}{$s_2$};
            \textrectangle{2}{-0.5}{\scalebox{0.9}{$\cdots$}};
            \textrectangle{3}{0.5}{};
            \textrectangle{4}{-0.5}{\scalebox{0.9}{$\cdots$}};
            \textrectangle{5}{0.5}{\scalebox{0.85}{$s_{k-2}$}};
            \textrectangle{6}{-0.5}{\scalebox{0.85}{$s_{k-1}$}};
            \textrectangle{7}{0.5}{{$s_{k}$}};
        \end{tikzpicture}\!\Big)}
    {
        p\Big(\!\begin{tikzpicture}[scale=0.3,baseline={([yshift=-0.6ex]current bounding box.center)}]
            \textrectangle{0}{-0.5}{$s_1$};
            \textrectangle{1}{0.5}{$s_2$};
            \textrectangle{2}{-0.5}{\scalebox{0.9}{$\cdots$}};
            \textrectangle{3}{0.5}{};
            \textrectangle{4}{-0.5}{\scalebox{0.9}{$\cdots$}};
            \textrectangle{5}{0.5}{\scalebox{0.85}{$s_{k-2}$}};
            \textrectangle{6}{-0.5}{\scalebox{0.85}{$s_{k-1}$}};
        \end{tikzpicture}\!\Big)
    }
    =
    \frac{
        p\Big(\!\begin{tikzpicture}[scale=0.3,baseline={([yshift=-0.6ex]current bounding box.center)}]
            \textrectangle{5}{0.5}{\scalebox{0.85}{$s_{k-2}$}};
            \textrectangle{6}{-0.5}{\scalebox{0.85}{$s_{k-1}$}};
            \textrectangle{7}{0.5}{{$s_{k}$}};
        \end{tikzpicture}\!\Big)}
    {
        p\Big(\!\begin{tikzpicture}[scale=0.3,baseline={([yshift=-0.6ex]current bounding box.center)}]
            \textrectangle{5}{0.5}{\scalebox{0.85}{$s_{k-2}$}};
            \textrectangle{6}{-0.5}{\scalebox{0.85}{$s_{k-1}$}};
        \end{tikzpicture}\!\Big)
    },\qquad
    \frac{
        p\Big(\!\begin{tikzpicture}[scale=0.3,baseline={([yshift=-0.6ex]current bounding box.center)}]
            \textrectangle{0}{-0.5}{$s_1$};
            \textrectangle{1}{0.5}{$s_2$};
            \textrectangle{2}{-0.5}{$s_3$};
            \textrectangle{3}{0.5}{\scalebox{0.9}{$\cdots$}};
            \textrectangle{4}{-0.5}{};
            \textrectangle{5}{0.5}{};
            \textrectangle{6}{-0.5}{\scalebox{0.9}{$\cdots$}};
            \textrectangle{7}{0.5}{{$s_{k}$}};
        \end{tikzpicture}\!\Big)}
    {
        p\Big(\!\begin{tikzpicture}[scale=0.3,baseline={([yshift=-0.6ex]current bounding box.center)}]
            \textrectangle{1}{0.5}{$s_2$};
            \textrectangle{2}{-0.5}{$s_3$};
            \textrectangle{3}{0.5}{\scalebox{0.9}{$\cdots$}};
            \textrectangle{4}{-0.5}{};
            \textrectangle{5}{0.5}{};
            \textrectangle{6}{-0.5}{\scalebox{0.9}{$\cdots$}};
            \textrectangle{7}{0.5}{{$s_{k}$}};
        \end{tikzpicture}\!\Big)
    }
    =
    \frac{
        p\Big(\!\begin{tikzpicture}[scale=0.3,baseline={([yshift=-0.6ex]current bounding box.center)}]
            \textrectangle{0}{-0.5}{$s_1$};
            \textrectangle{1}{0.5}{$s_2$};
            \textrectangle{2}{-0.5}{$s_3$};
        \end{tikzpicture}\!\Big)}
    {
        p\Big(\!\begin{tikzpicture}[scale=0.3,baseline={([yshift=-0.6ex]current bounding box.center)}]
            \textrectangle{1}{0.5}{$s_2$};
            \textrectangle{2}{-0.5}{$s_3$};
        \end{tikzpicture}\!\Big)
    }.
\end{eqnarray}
where we remark that this is satisfied independently of parity of the first
site and the length~$k$ (i.e.\ it also holds when all the configurations are
flipped upside-down). Note that this property, although straightforward to
prove (the proof is provided in~\cite{klobas2020matrix,klobas2020exactPhD}), is
nontrivial and for example does not hold for the stationary state of a related
cellular automaton~\cite{wilkinson2020exact}
(rule 201; cf.\ Section~\ref{sec:relModels}) with
otherwise similar properties.

A direct physical consequence of the factorization of asymptotic conditional
probabilities is \emph{statistical independence of solitons}: in the stationary
state, the probability of finding a left- (or right-) moving soliton is
independent of the positions of other solitons, as long as there is no
left-mover (or right-mover) on the neighbouring sites.  The conditional
probability $p_l$ of observing a left-mover, if we know that the neighbouring
left ray does not contain a left-moving soliton is therefore a well-defined
quantity and can be expressed in terms of asymptotic probabilities as follows,
\begin{eqnarray}
    p_l=\frac{p\big(\!
    \begin{tikzpicture}[baseline={([yshift=-0.6ex]current bounding box.center)},
        scale=0.15] 
        \rectangle{0}{-0.5}{0};
        \rectangle{1}{0.5}{0};
        \rectangle{2}{-0.5}{1};
    \end{tikzpicture}
    \!\big)}{p\big(\! 
    \begin{tikzpicture}[baseline={([yshift=-0.6ex]current bounding box.center)},
        scale=0.15] 
        \rectangle{0}{-0.5}{0};
        \rectangle{1}{0.5}{0};
    \end{tikzpicture}
    \!\big)}=
    \frac{p\big(\!
    \begin{tikzpicture}[baseline={([yshift=-0.6ex]current bounding box.center)},
        scale=0.15] 
        \rectangle{0}{-0.5}{0};
        \rectangle{1}{0.5}{1};
        \rectangle{2}{-0.5}{0};
    \end{tikzpicture}
    \!\big)}{p\big(\! 
    \begin{tikzpicture}[baseline={([yshift=-0.6ex]current bounding box.center)},
        scale=0.15] 
        \rectangle{0}{-0.5}{0};
        \rectangle{1}{0.5}{1};
    \end{tikzpicture}
    \!\big)}
    =\frac{\xi(\lambda+\omega(1-\xi))}{\lambda(1+\xi)+\xi(1-\xi\omega)},
\end{eqnarray}
where~$\lambda$ is the leading eigenvalue of~$T=(W_0+W_1)(W'_0+W'_1)$
(see~\cite{klobas2020matrix} for the details). Analogously, the conditional
probability of finding a right-mover at a given site, if the neighbouring ray
on the right is empty, is given by
\begin{eqnarray}
    p_r=\frac{p\big(\!
    \begin{tikzpicture}[baseline={([yshift=-0.6ex]current bounding box.center)},
        scale=0.15] 
        \rectangle{0}{-0.5}{1};
        \rectangle{1}{0.5}{0};
        \rectangle{2}{-0.5}{0};
    \end{tikzpicture}
    \!\big)}{p\big(\! 
    \begin{tikzpicture}[baseline={([yshift=-0.6ex]current bounding box.center)},
        scale=0.15] 
        \rectangle{0}{0.5}{0};
        \rectangle{1}{-0.5}{0};
    \end{tikzpicture}
    \!\big)}=
    \frac{p\big(\!
    \begin{tikzpicture}[baseline={([yshift=-0.6ex]current bounding box.center)},
        scale=0.15] 
        \rectangle{0}{-0.5}{0};
        \rectangle{1}{0.5}{1};
        \rectangle{2}{-0.5}{0};
    \end{tikzpicture}
    \!\big)}{p\big(\! 
    \begin{tikzpicture}[baseline={([yshift=-0.6ex]current bounding box.center)},
        scale=0.15] 
        \rectangle{0}{0.5}{1};
        \rectangle{1}{-0.5}{0};
    \end{tikzpicture}
    \!\big)}
    =\frac{\omega(\lambda+\xi(1-\omega))}{\lambda(1+\omega)+\omega(1-\xi\omega)}.
\end{eqnarray}
Note that the set of probabilities~$(p_l,p_r)$ provides an equivalent parametrisation
of the steady state, since the  mapping $(\xi,\omega)\to(p_l,p_r)$ can be inverted,
\begin{eqnarray}
\label{xiomega}
    \xi=\frac{p_l(1-p_r)}{(1-p_l)^2},\qquad \omega=\frac{p_r(1-p_l)}{(1-p_r)^2}.
\end{eqnarray}
Comparing these expressions with the TBA equations~\eqref{eq:TBAdefEq1} and~\eqref{eq:TBAdefEq2}
we realize that $p_l$ and $p_r$ coincide with the filling functions~$\vartheta_{+}$ and
$\vartheta_{-}$, respectively.

\subsection{Time-states}\label{sec:TimeStates}
We define \emph{time-states} as probability distributions of
\emph{time-configurations}, i.e.\ configurations observed in time (as
schematically shown in Fig.~\ref{fig:TSfig2}), under the assumption that the
system is in a stationary state~$\vec{p}$. Explicitly, the components of the
time-state $\vec{q}$,
\begin{eqnarray}
    q_{b_0\,b_1\,\ldots b_{m-1}}=\smashoperator{\sum_{\ul{s}\equiv(s_{-m},s_{-m+1},\ldots,s_{m-1})}}p_{\ul{s}}\,
    \delta_{b_0,s^0_{0}}
    \delta_{b_1,s^1_{1}}
    \delta_{b_2,s^2_{0}}
    \delta_{b_3,s^3_{1}}
    \cdots
    \delta_{b_{m-1},s^{m-1}_{{\mathrm{mod}}(m-1,2)}},
\end{eqnarray}
where $s^t_k$ is the value at site $k$ of the configuration
$\ul{s}=(s_{-m},s_{-m+1},\ldots s_{m-1})$ propagated for $t$ time steps
started from initial data $s^t_k = s_k$, and $p_{\ul{s}}$ is the asymptotic
stationary probability (in our case given by Equations~\eqref{eq:asymptPropEven}
and~\eqref{eq:asymptPropOdd}). 

\begin{figure}
    \centering
    \includegraphics[width=0.5\textwidth]{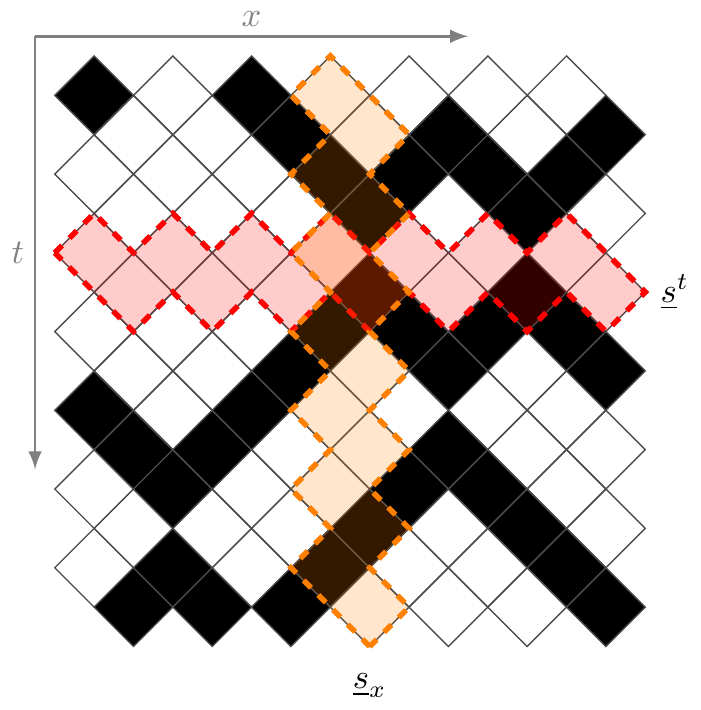}
    \caption{\label{fig:TSfig2} Example of time evolution. The red strip denotes a
    configuration of bits at a fixed time $\ul{s}^t$, while the orange strip
    represents an example of a \emph{time-configuration}, i.e.\ a set of bits observed
    at the same position in consecutive time-steps. We assume $\ul{s}^t$ to
    be distributed according to a stationary probability distribution~$\vec{p}$, while
    the probabilities of time-configurations $\ul{s}_x$ are given by the
    corresponding \emph{time-state} (probability distribution) $\vec{q}$.
    }
\end{figure}

To understand why for RCA54 the probability distribution~$\vec{q}$ can be
explicitly obtained, we imagine the following setup (schematically illustrated
in Fig.~\ref{fig:TSfig1}): at time~$t=0$ we pick a random configuration, distributed
according to the stationary probability distribution $\vec{p}$, and let it evolve in
time, while keeping track only of the bit at position $x=0$ or $x=1$ (depending on
the parity of the time-step, cf.~Fig.~\ref{fig:TSfig2}). Generically, the probability
of observing a soliton at time~$t$ depends on the full observed time-configuration.
In our case, however, we recall that the solitons in the underlying stationary 
state~$\vec{p}$ are statistically independent. This implies that at any time,
the probability of a left (or right) mover reaching the observed site, is constant
and given by~$p_l$ (or $p_r$), as long as there was no soliton in the previous time-step.

Explicitly, the conditional probability of observing a time-configuration
$(b_0,b_1,\ldots,b_{k-1},b_{k})$ given the previous configuration
$(b_0,\ldots,b_{k-1})$ therefore depends only on the last $4$ bits, which
implies a set of relations analogous to~\eqref{eq:condProb2},
\begin{eqnarray}\fl
    \frac{q_{b_0 b_1\ldots b_{2 k-1}}}{q_{b_0 b_1 \ldots b_{2 k-2}}}=
    f(b_{2k-4},b_{2k-3},b_{2k-2},b_{2k-1}),\quad
    \frac{q_{b_0 b_1\ldots b_{2 k}}}{q_{b_0 b_1 \ldots b_{2 k-1}}}=
    f^{\prime}(b_{2k-3},b_{2k-2},b_{2k-1},b_{2k}),
\end{eqnarray}
where $f,f^{\prime}$ are two yet unknown functions of the last four bits.
Conditional probabilities $f^{(\prime)}$ can be straightforwardly determined
by observing that short $3$-site time-configurations can be divided into
three groups.
\begin{enumerate}
    \item A configuration can be \emph{inaccessible}, i.e.\ does not appear in
        a time-configuration $\ul{s}_x$ corresponding to any initial
        configuration~$\ul{s}^t$. To identify them, we only have to note
        that the following $4$ configurations never appear in the local time-evolution
        rules~\eqref{eq:TEDrules},
        \begin{eqnarray}\label{eq:confForbidden}
            \begin{tikzpicture}[baseline={([yshift=-0.6ex]current bounding box.center)},scale=0.2]
                \rectangle{0.5}{1}{0};
                \rectangle{-0.5}{0}{1};
                \rectangle{0.5}{-1}{0};
            \end{tikzpicture}\qquad
            \begin{tikzpicture}[baseline={([yshift=-0.6ex]current bounding box.center)},scale=0.2]
                \rectangle{0.5}{1}{1};
                \rectangle{-0.5}{0}{1};
                \rectangle{0.5}{-1}{1};
            \end{tikzpicture}\qquad
            \begin{tikzpicture}[baseline={([yshift=-0.6ex]current bounding box.center)},scale=0.2]
                \rectangle{-0.5}{1}{0};
                \rectangle{0.5}{0}{1};
                \rectangle{-0.5}{-1}{0};
            \end{tikzpicture}\qquad
            \begin{tikzpicture}[baseline={([yshift=-0.6ex]current bounding box.center)},scale=0.2]
                \rectangle{-0.5}{1}{1};
                \rectangle{0.5}{0}{1};
                \rectangle{-0.5}{-1}{1};
            \end{tikzpicture}.
        \end{eqnarray}
        The probability of obtaining such configurations is $0$, which means that the conditional
        probabilities $f^{(\prime)}(0,1,0,s)$ and $f^{(\prime)}(1,1,1,s)$ are not well defined
        for any~$s$, while
        \begin{eqnarray}
            f^{(\prime)}(s,0,1,0)=f^{(\prime)}(s,1,1,1)=0.
        \end{eqnarray}
    \item A $3$-site configuration can have a unique extension into an \emph{accessible}
        $4$-site configuration,
        \begin{eqnarray}\fl
            \begin{tikzpicture}[baseline={([yshift=-0.6ex]current bounding box.center)},scale=0.2]
                \rectangle{0.5}{-1}{1};
                \rectangle{-0.5}{0}{0};
                \rectangle{0.5}{1}{0};
                \draw[thick,->] (1.5,0) -- (3.5,0);
                \rectangle{5}{-2}{1};
                \rectangle{6}{-1}{1};
                \rectangle{5}{0}{0};
                \rectangle{6}{1}{0};
            \end{tikzpicture},\quad
            \begin{tikzpicture}[baseline={([yshift=-0.6ex]current bounding box.center)},scale=0.2]
                \rectangle{0.5}{-1}{1};
                \rectangle{-0.5}{0}{1};
                \rectangle{0.5}{1}{0};
                \draw[thick,->] (1.5,0) -- (3.5,0);
                \rectangle{5}{-2}{0};
                \rectangle{6}{-1}{1};
                \rectangle{5}{0}{1};
                \rectangle{6}{1}{0};
            \end{tikzpicture},\quad
            \begin{tikzpicture}[baseline={([yshift=-0.6ex]current bounding box.center)},scale=0.2]
                \rectangle{0.5}{-1}{1};
                \rectangle{-0.5}{0}{0};
                \rectangle{0.5}{1}{1};
                \draw[thick,->] (1.5,0) -- (3.5,0);
                \rectangle{5}{-2}{1};
                \rectangle{6}{-1}{1};
                \rectangle{5}{0}{0};
                \rectangle{6}{1}{1};
            \end{tikzpicture},\quad
            \begin{tikzpicture}[baseline={([yshift=-0.6ex]current bounding box.center)},scale=0.2]
                \rectangle{-0.5}{-1}{1};
                \rectangle{0.5}{0}{0};
                \rectangle{-0.5}{1}{0};
                \draw[thick,->] (2,0) -- (4,0);
                \rectangle{6}{-2}{1};
                \rectangle{5}{-1}{1};
                \rectangle{6}{0}{0};
                \rectangle{5}{1}{0};
            \end{tikzpicture},\quad
            \begin{tikzpicture}[baseline={([yshift=-0.6ex]current bounding box.center)},scale=0.2]
                \rectangle{-0.5}{-1}{1};
                \rectangle{0.5}{0}{1};
                \rectangle{-0.5}{1}{0};
                \draw[thick,->] (2,0) -- (4,0);
                \rectangle{6}{-2}{0};
                \rectangle{5}{-1}{1};
                \rectangle{6}{0}{1};
                \rectangle{5}{1}{0};
            \end{tikzpicture},\quad
            \begin{tikzpicture}[baseline={([yshift=-0.6ex]current bounding box.center)},scale=0.2]
                \rectangle{-0.5}{-1}{1};
                \rectangle{0.5}{0}{0};
                \rectangle{-0.5}{1}{1};
                \draw[thick,->] (2,0) -- (4,0);
                \rectangle{6}{-2}{1};
                \rectangle{5}{-1}{1};
                \rectangle{6}{0}{0};
                \rectangle{5}{1}{1};
            \end{tikzpicture}.
        \end{eqnarray}
        These configurations represent examples, where the other choice to continue upwards
        would result in an inaccessible configuration. The corresponding conditional probabilities
        must therefore be $1$,
        \begin{eqnarray}
            f^{(\prime)}(s,0,1,1)=f^{(\prime)}(0,1,1,0)=1,\qquad s\in\{0,1\}.
        \end{eqnarray}
    \item A configuration $(s_1,s_2,s_3)$ is valid and both $(s_1,s_2,s_3,0)$ and
        $(s_1,s_2,s_3,1)$ are also valid, in which case the two conditional probabilities
        are between $0$ and $1$. Explicitly, the configurations falling in this class are
        the following,
        \begin{eqnarray}\fl
            \begin{tikzpicture}[baseline={([yshift=-0.6ex]current bounding box.center)},scale=0.2]
                \rectangle{0.5}{1}{0};
                \rectangle{-0.5}{0}{0};
                \rectangle{0.5}{-1}{0};
                \draw[thick,->] (2,0.5) -- (3.5,2);
                %node [midway,above,xshift=-1ex,yshift=1ex] {\scalebox{0.75}{$1-p_r$}};
                \draw[thick,->] (2,-0.5) -- (3.5,-2);
                %node [midway,below,yshift=-1ex] {\scalebox{0.75}{$p_r$}};
                \draw[thick,-latex,red] (3.75,-2.25) -- (4.75,-3.25);
                \rectangle{5.5}{1}{0};
                \rectangle{6.5}{2}{0};
                \rectangle{5.5}{3}{0};
                \rectangle{6.5}{4}{0};
                \rectangle{5.5}{-4}{1};
                \rectangle{6.5}{-3}{0};
                \rectangle{5.5}{-2}{0};
                \rectangle{6.5}{-1}{0};
            \end{tikzpicture},\quad\!
            \begin{tikzpicture}[baseline={([yshift=-0.6ex]current bounding box.center)},scale=0.2]
                \rectangle{0.5}{-1}{0};
                \rectangle{-0.5}{0}{0};
                \rectangle{0.5}{1}{1};
                \draw[thick,->] (2,0.5) -- (3.5,2);
                %node [midway,above,xshift=-1ex,yshift=1ex] {\scalebox{0.75}{$1-p_r$}};
                \draw[thick,->] (2,-0.5) -- (3.5,-2);
                %node [midway,below,yshift=-1ex] {\scalebox{0.75}{$p_r$}};
                \draw[thick,-latex,red] (3.75,-2.25) -- (4.75,-3.25);
                \rectangle{5.5}{1}{0};
                \rectangle{6.5}{2}{0};
                \rectangle{5.5}{3}{0};
                \rectangle{6.5}{4}{1};
                \rectangle{5.5}{-4}{1};
                \rectangle{6.5}{-3}{0};
                \rectangle{5.5}{-2}{0};
                \rectangle{6.5}{-1}{1};
            \end{tikzpicture},\quad\!
            \begin{tikzpicture}[baseline={([yshift=-0.6ex]current bounding box.center)},scale=0.2]
                \rectangle{0.5}{-1}{0};
                \rectangle{-0.5}{0}{1};
                \rectangle{0.5}{1}{1};
                \draw[thick,->] (2,0.5) -- (3.5,2);
                %node [midway,above,xshift=-1ex,yshift=1ex] {\scalebox{0.75}{$1-p_r$}};
                \draw[thick,->] (2,-0.5) -- (3.5,-2);
                %node [midway,below,yshift=-1ex] {\scalebox{0.75}{$p_r$}};
                \draw[thick,-latex,red] (3.75,-0.25) -- (4.75,-1.25);
                \rectangle{5.5}{1}{0};
                \rectangle{6.5}{2}{0};
                \rectangle{5.5}{3}{1};
                \rectangle{6.5}{4}{1};
                \rectangle{5.5}{-4}{1};
                \rectangle{6.5}{-3}{0};
                \rectangle{5.5}{-2}{1};
                \rectangle{6.5}{-1}{1};
            \end{tikzpicture},\quad\!
            \begin{tikzpicture}[baseline={([yshift=-0.6ex]current bounding box.center)},scale=0.2]
                \rectangle{-0.5}{-1}{0};
                \rectangle{0.5}{0}{0};
                \rectangle{-0.5}{1}{0};
                \draw[thick,->] (2,0.5) -- (3.5,2);
                %node [midway,above,xshift=-1ex,yshift=1ex] {\scalebox{0.75}{$1-p_r$}};
                \draw[thick,->] (2,-0.5) -- (3.5,-2);
                %node [midway,below,yshift=-1ex] {\scalebox{0.75}{$p_r$}};
                \draw[thick,-latex,red] (8.25,-2.25) -- (7.25,-3.25);
                \rectangle{6.5}{1}{0};
                \rectangle{5.5}{2}{0};
                \rectangle{6.5}{3}{0};
                \rectangle{5.5}{4}{0};
                \rectangle{6.5}{-4}{1};
                \rectangle{5.5}{-3}{0};
                \rectangle{6.5}{-2}{0};
                \rectangle{5.5}{-1}{0};
            \end{tikzpicture}\!,\quad\!
            \begin{tikzpicture}[baseline={([yshift=-0.6ex]current bounding box.center)},scale=0.2]
                \rectangle{-0.5}{-1}{0};
                \rectangle{0.5}{0}{0};
                \rectangle{-0.5}{1}{1};
                \draw[thick,->] (2,0.5) -- (3.5,2);
                %node [midway,above,xshift=-1ex,yshift=1ex] {\scalebox{0.75}{$1-p_r$}};
                \draw[thick,->] (2,-0.5) -- (3.5,-2);
                %node [midway,below,yshift=-1ex] {\scalebox{0.75}{$p_r$}};
                \draw[thick,-latex,red] (8.25,-2.25) -- (7.25,-3.25);
                \rectangle{6.5}{1}{0};
                \rectangle{5.5}{2}{0};
                \rectangle{6.5}{3}{0};
                \rectangle{5.5}{4}{1};
                \rectangle{6.5}{-4}{1};
                \rectangle{5.5}{-3}{0};
                \rectangle{6.5}{-2}{0};
                \rectangle{5.5}{-1}{1};
            \end{tikzpicture}\!,\quad\!
            \begin{tikzpicture}[baseline={([yshift=-0.6ex]current bounding box.center)},scale=0.2]
                \rectangle{-0.5}{-1}{0};
                \rectangle{0.5}{0}{1};
                \rectangle{-0.5}{1}{1};
                \draw[thick,->] (2,0.5) -- (3.5,2);
                %node [midway,above,xshift=-1ex,yshift=1ex] {\scalebox{0.75}{$1-p_r$}};
                \draw[thick,->] (2,-0.5) -- (3.5,-2);
                %node [midway,below,yshift=-1ex] {\scalebox{0.75}{$p_r$}};
                \draw[thick,-latex,red] (8.25,-0.25) -- (7.25,-1.25);
                \rectangle{6.5}{1}{0};
                \rectangle{5.5}{2}{0};
                \rectangle{6.5}{3}{1};
                \rectangle{5.5}{4}{1};
                \rectangle{6.5}{-4}{1};
                \rectangle{5.5}{-3}{0};
                \rectangle{6.5}{-2}{1};
                \rectangle{5.5}{-1}{1};
            \end{tikzpicture}.
        \end{eqnarray} 
        Here the red arrows show the position in which the previously untracked
        (left or right moving) particle should appear for the new bit in the
        configuration to be equal to $1$.  As a consequence of statistical
        independence of solitons, the configurations on the bottom all
        arise with the conditional probability $p_r$ or $p_l$, while the
        probabilities associated with the top configurations are $1-p_r$ and
        $1-p_l$. Explicitly,
        \begin{eqnarray}\fl
            \eqalign{
                f(s,0,0,0)=f(1,1,0,0)=1-p_r,\qquad
                &f(s,0,0,1)=f(1,1,0,1)=p_r,\\
                f^{\prime}(s,0,0,0)=f^{\prime}(1,1,0,0)=1-p_l,
                &f^{\prime}(s,0,0,1)=f^{\prime}(1,1,0,1)=p_l.
            }
        \end{eqnarray}
\end{enumerate}
This exhausts all the possible conditional probabilities, which allows us to express 
the full probability distribution $\vec{q}$ as
\begin{eqnarray}\fl
    q_{b_0 b_1\ldots b_{m-1}} = q_{b_0 b_1 b_2} f(b_0,b_1,b_2,b_3) f^{\prime}(b_1,b_2,b_3,b_4)
    \cdots f^{(\prime)}(b_{m-4},b_{m-3},b_{m-2},b_{m-,1}),
\end{eqnarray}
where the last term in the product is $f$ or $f^{\prime}$ for even or odd $m$, respectively.
To completely determine~$\vec{q}$, we have to take into account the stationarity of the underlying
state~$\vec{p}$, which implies the following consistency condition for any length $m$,
\begin{eqnarray}
    q_{b_0 b_1\ldots b_{m-1}} = \sum_{b_{-2},b_{-1}\in\{0,1\}}
    q_{b_{-2}b_{-1}b_0 b_1\ldots b_{m-1}}.
\end{eqnarray}
This finally provides (up to normalization) the explicit values of~$q_{b_0 b_1 b_2}$ and we 
obtain an explicit representation of the time-state in a form similar to
the \emph{patch-state ansatz} introduced in Section~\ref{sec:PSA}.

The time-state $\vec{q}$ can be equivalently represented in an MPA form, by encoding
the appropriate factors $f^{(\prime)}$ in the $3\times 3$ matrices $A^{(\prime)}_{s}$,
\begin{eqnarray}\fl
    \A_0=\begin{bmatrix}
        1-p_r&0&0\\
        0&0&0\\
        1&0&0
    \end{bmatrix},\qquad
    \A_1=\begin{bmatrix}
        0&p_r&0\\
        0&0&1\\
        0&0&0
    \end{bmatrix},\qquad
    \B_s(p_r,p_l)=\A_s(p_l,p_r),
\end{eqnarray}
and defining the appropriate boundary vectors $\bra*{L}$ and $\ket{R}$,
\begin{eqnarray}
    \bra{L} = \frac{1}{1+p_r+p_l}
    \begin{bmatrix}
        1&p_l & p_r
    \end{bmatrix},\qquad
    \ket{R}=\begin{bmatrix}
        1 \\ 1 \\ 1
    \end{bmatrix}.
\end{eqnarray}
Using these definitions, the full time-state can be succinctly summarized as,
\begin{eqnarray}\label{eq:timeStateFinal}
    \vec{q}=\bra{L}\vA_0\vB_1\vA_2\cdots \mathbf{A}_{m-1}^{(\prime)}\ket{R}.
\end{eqnarray}
Note that the right boundary vector~$\ket{R}$ is the right eigenvector corresponding
to the eigenvalue $1$ of both $\A_0+\A_1$ and $\B_0+\B_1$,
which immediately gives us access to the normalization of the
state,
\begin{eqnarray}\fl
    \smashoperator[r]{\sum_{b_0,b_1,\ldots,b_{m-1}}} q_{b_0 b_1 \ldots b_{m-1}}=
    \bra{L}
    \underbrace{(\A_0+\A_1)(\B_0+\B_1)\cdots (A^{(\prime)}_0+A^{(\prime)}_1)}_{m}
    \ket{R} =\braket{L}{R}=1.
\end{eqnarray}

\subsection{Space evolution}\label{subsec:spaceEvol}
The discussion in the previous subsection brings us to a related question. Let
us imagine we know that for a given trajectory the time-configuration at position
$x$ is equal to $\ul{s}_x$. Is it possible to deterministically map it
into the neighbouring configuration $\ul{s}_{x+1}$ (schematically 
represented in Fig.~\ref{fig:TSfig3})? In other words: we wish to understand
whether we can define a local and deterministic model in the space-direction, with
respect to which the class of time-states introduced above is stationary.
Generally, there is no guarantee that the space evolution can be expressed as a
composition of local deterministic maps. In our case, however, we expect this
to be possible due to the solitonic description of the model: the dynamics in
space can be again understood as solitons moving in one of the two directions
with a fixed velocity $1$, and undergoing pairwise scattering (see e.g.\
Fig.~\ref{fig:TSfig2}). The main difference with respect to the usual dynamics
in the time direction, is the nature of the interaction, which in this case
temporarily speeds the particles up, rather than slowing them down.

\begin{figure}
    \centering
    \includegraphics[width=0.75\textwidth]{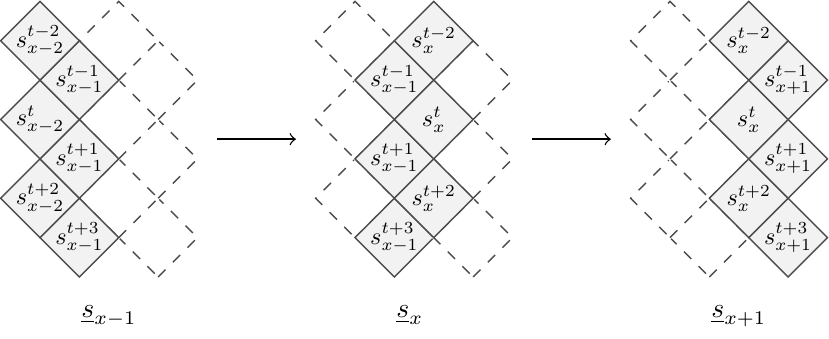}
    \caption{\label{fig:TSfig3} Schematic representation of the mapping between
    time-configurations. Evolution in space is analogous to time-evolution (see
    e.g.\ Fig.~\ref{fig:TEGeofig}): at each point $x$, the bits are updated in
    rows with the same parity as $x$, while the others are unchanged.
    }
\end{figure}

We therefore expect that it is possible to find a local deterministic map
\begin{eqnarray}
    \phi:\underbrace{\mathbb{Z}_2\times \mathbb{Z}_2 \times \cdots \times \mathbb{Z}_2}_{2r+1}
    \to
    \mathbb{Z}_2,\quad r\in\mathbb{N},
\end{eqnarray}
that gives the new value for a bit in the middle of the configuration of $2r+1$ sites,
\begin{eqnarray} \label{eq:defSEmap}
    s_{x+1}^t=\phi\left(s_{x/x-1}^{t-r},\ldots,s_{x-1}^{t-2},
    s_{x}^{t-1},s_{x-1}^t,s_{x}^{t+1},\ldots,s_{x/x-1}^{t+r}
    \right).
\end{eqnarray}
Note that $\chi$ given by~\eqref{eq:TErules} is of this form with $r=1$. At the moment
it is not clear what the support for the local space-evolution maps should be, but
by examining the time-evolution rules~\eqref{eq:TEDrules}, we quickly realize that the
support should be larger than $3$ (i.e.\ $r\ge 2$), since e.g.\ the last two diagrams
suggest two different updated values corresponding to the time-configuration $(0,1,1)$.
Nonetheless, three-site time-configurations provide a starting point for
construction of space-evolution maps with larger support. Again we realize that
three-site configurations can be divided into three types.
\begin{enumerate}
    \item As we already noted earlier, time-configurations
        cannot contain substrings $(0,1,0)$ and $(1,1,1)$, as these are
        forbidden by the dynamic rules. The space map $\phi$ is therefore
        defined on the reduced space consisting only of allowed configurations.
    \item The first four diagrams in Eq.~\eqref{eq:TEDrules} already provide the
        deterministic space updates for $3$-site configurations $(0,0,0)$, $(1,0,0)$,
        $(1,0,1)$ and $(0,1,1)$,
        \begin{eqnarray}\label{eq:SEdia1}
            \begin{tikzpicture}
                [scale=0.3,baseline={([yshift=-0.6ex]current bounding box.center)}]
                \rectangle{-1}{0}{0};
                \rectangle{0}{-1}{0};
                \rectangle{0}{1}{0};
                \redrectangle{1}{0}{0};
            \end{tikzpicture}\quad
            \begin{tikzpicture}
                [scale=0.3,baseline={([yshift=-0.6ex]current bounding box.center)}]
                \rectangle{-1}{0}{0};
                \rectangle{0}{-1}{0};
                \rectangle{0}{1}{1};
                \redrectangle{1}{0}{1};
            \end{tikzpicture}\quad
            \begin{tikzpicture}
                [scale=0.3,baseline={([yshift=-0.6ex]current bounding box.center)}]
                \rectangle{-1}{0}{0};
                \rectangle{0}{-1}{1};
                \rectangle{0}{1}{1};
                \redrectangle{1}{0}{0};
            \end{tikzpicture}\quad
            \begin{tikzpicture}
                [scale=0.3,baseline={([yshift=-0.6ex]current bounding box.center)}]
                \rectangle{-1}{0}{0};
                \rectangle{0}{-1}{1};
                \rectangle{0}{1}{0};
                \redrectangle{1}{0}{1};
            \end{tikzpicture}.
        \end{eqnarray}
    \item In the case of $(0,1,1)$ and $(1,1,0)$ one needs more information to be able
        to find the updated value of the central bit, since in both cases the
        transitions to $0$ and to $1$ are possible. However, it suffices to extend the
        configurations for two sites to top and bottom to find a deterministic update.
        To demonstrate it, let us first focus on $(0,1,1)$. By avoiding the forbidden
        configurations, there are only two ways to extend the time-configuration for
        two-sites to the bottom, $(0,1,1,0,0)$ and $(0,1,1,0,1)$,
        \begin{eqnarray}
            \begin{tikzpicture}
                [scale=0.3,baseline={([yshift=-0.6ex]current bounding box.center)}]
                \rectangle{0}{1}{0};
                \rectangle{-1}{0}{1};
                \rectangle{0}{-1}{1};
                
                \draw [thick, ->] (-0.5,2.25) -- (-0.5,3) -- (9.5,3) -- (9.5,2.25);
                \draw [thick, ->] (5.5,3) -- (5.5,2.25);
                
                \rectangle{6}{1}{0};
                \rectangle{5}{0}{1};
                \rectangle{6}{-1}{1};
                \rectangle{5}{-2}{0};
                \rectangle{6}{-3}{0};

                \rectangle{10}{1}{0};
                \rectangle{9}{0}{1};
                \rectangle{10}{-1}{1};
                \rectangle{9}{-2}{0};
                \rectangle{10}{-3}{1};
            \end{tikzpicture}.
        \end{eqnarray}
        Now we are able to find a unique way to update these $5$-site
        configurations, using the rules we already have. For either one of these
        configurations the update rules can be applied to the bottom three sites,
        after which we try to set the new value at the top to either $0$ and $1$ and in
        both cases, only one of them obeys the restriction to allowed configurations,
        \begin{eqnarray}\fl
            \begin{tikzpicture}
                [scale=0.3,baseline={([yshift=-0.6ex]current bounding box.center)}]
                \rectangle{0}{3}{0};
                \rectangle{-1}{2}{1};
                \rectangle{0}{1}{1};
                \rectangle{-1}{0}{0};
                \rectangle{0}{-1}{0};
                \draw [rounded corners,very thick,red] (-2.5,0) -- (0,-2.5) -- (1.5,-1) 
                -- (0.5,0) -- (1.5,1) -- (0,2.5) -- cycle;
                \draw [thick,-latex] (1.25,0) -- (4.75,0) 
                node[pos=0.5,above] {\scalebox{0.7}{update}};
                \rectangle{7}{3}{0};
                \rectangle{6}{2}{1};
                \rectangle{7}{1}{1};
                \rectangle{6}{0}{0};
                \rectangle{7}{-1}{0};
                \rectangle{8}{0}{1};
                \draw [rounded corners,very thick,red] (4.5,2) -- (7,-0.5) -- (8.5,1) 
                -- (7.5,2) -- (8.5,3) -- (7,4.5) -- cycle;
                \draw [thick] (8.25,2) -- (11.25,2) 
                node[pos=0.575,above] {\scalebox{0.7}{update}};
                \draw [thick,-latex] (11.25,2) -- (12.25,3);
                \draw [thick,-latex] (11.25,2) -- (12.25,1);
                \rectangle{14.5}{7.25}{0};
                \rectangle{13.5}{6.25}{1};
                \rectangle{14.5}{5.25}{1};
                \rectangle{13.5}{4.25}{0};
                \rectangle{14.5}{3.25}{0};
                \rectangle{15.5}{4.25}{1};
                \rectangle{15.5}{6.25}{0};

                \rectangle{14.5}{0.75}{0};
                \rectangle{13.5}{-0.25}{1};
                \rectangle{14.5}{-1.25}{1};
                \rectangle{13.5}{-2.25}{0};
                \rectangle{14.5}{-3.25}{0};
                \rectangle{15.5}{-2.25}{1};
                \rectangle{15.5}{-0.25}{1};

                \draw [very thick,red] (16.5,-4.25) -- (12.5,1.75);
                \draw [very thick,red] (12.5,-4.25) -- (16.5,1.75);
            \end{tikzpicture},\qquad
            \begin{tikzpicture}
                [scale=0.3,baseline={([yshift=-0.6ex]current bounding box.center)}]
                \rectangle{0}{3}{0};
                \rectangle{-1}{2}{1};
                \rectangle{0}{1}{1};
                \rectangle{-1}{0}{0};
                \rectangle{0}{-1}{1};
                \draw [rounded corners,very thick,red] (-2.5,0) -- (0,-2.5) -- (1.5,-1) 
                -- (0.5,0) -- (1.5,1) -- (0,2.5) -- cycle;
                \draw [thick,-latex] (1.25,0) -- (4.75,0) 
                node[pos=0.5,above] {\scalebox{0.7}{update}};
                \rectangle{7}{3}{0};
                \rectangle{6}{2}{1};
                \rectangle{7}{1}{1};
                \rectangle{6}{0}{0};
                \rectangle{7}{-1}{1};
                \rectangle{8}{0}{0};
                \draw [rounded corners,very thick,red] (4.5,2) -- (7,-0.5) -- (8.5,1) 
                -- (7.5,2) -- (8.5,3) -- (7,4.5) -- cycle;
                \draw [thick] (8.25,2) -- (11.25,2) 
                node[pos=0.575,above] {\scalebox{0.7}{update}};
                \draw [thick,-latex] (11.25,2) -- (12.25,3);
                \draw [thick,-latex] (11.25,2) -- (12.25,1);
                \rectangle{14.5}{7.25}{0};
                \rectangle{13.5}{6.25}{1};
                \rectangle{14.5}{5.25}{1};
                \rectangle{13.5}{4.25}{0};
                \rectangle{14.5}{3.25}{1};
                \rectangle{15.5}{4.25}{0};
                \rectangle{15.5}{6.25}{0};

                \rectangle{14.5}{0.75}{0};
                \rectangle{13.5}{-0.25}{1};
                \rectangle{14.5}{-1.25}{1};
                \rectangle{13.5}{-2.25}{0};
                \rectangle{14.5}{-3.25}{1};
                \rectangle{15.5}{-2.25}{0};
                \rectangle{15.5}{-0.25}{1};

                \draw [very thick,red] (16.5,8.25) -- (12.5,2.25);
                \draw [very thick,red] (12.5,8.25) -- (16.5,2.25);
            \end{tikzpicture}.
        \end{eqnarray}
        Thus we see that to find the updated value of the middle bit in the
        subconfiguration $(0,1,1)$, it is sufficient to extend the
        subconfiguration two sites backwards in time.  Similarly, for $(1,1,0)$
        we have to check the bits to sites forward in time. The rules for all
        such subconfigurations can be summarized by the following diagrams,
        \begin{eqnarray}\label{eq:SEdia2}
            \begin{tikzpicture}
                [scale=0.3,baseline={([yshift=-0.6ex]current bounding box.center)}]
                \rectangle{0}{5}{2};
                \rectangle{-1}{4}{2};
                \rectangle{0}{3}{0};
                \rectangle{-1}{2}{1};
                \rectangle{0}{1}{1};
                \rectangle{-1}{0}{0};
                \rectangle{0}{-1}{0};
                \redrectangle{1}{2}{0};
            \end{tikzpicture},\qquad
            \begin{tikzpicture}
                [scale=0.3,baseline={([yshift=-0.6ex]current bounding box.center)}]
                \rectangle{0}{5}{2};
                \rectangle{-1}{4}{2};
                \rectangle{0}{3}{0};
                \rectangle{-1}{2}{1};
                \rectangle{0}{1}{1};
                \rectangle{-1}{0}{0};
                \rectangle{0}{-1}{1};
                \redrectangle{1}{2}{1};
            \end{tikzpicture},\qquad
            \begin{tikzpicture}
                [scale=0.3,baseline={([yshift=-0.6ex]current bounding box.center)}]
                \rectangle{0}{5}{0};
                \rectangle{-1}{4}{0};
                \rectangle{0}{3}{1};
                \rectangle{-1}{2}{1};
                \rectangle{0}{1}{0};
                \rectangle{-1}{0}{2};
                \rectangle{0}{-1}{2};
                \redrectangle{1}{2}{0};
            \end{tikzpicture},\qquad
            \begin{tikzpicture}
                [scale=0.3,baseline={([yshift=-0.6ex]current bounding box.center)}]
                \rectangle{0}{5}{1};
                \rectangle{-1}{4}{0};
                \rectangle{0}{3}{1};
                \rectangle{-1}{2}{1};
                \rectangle{0}{1}{0};
                \rectangle{-1}{0}{2};
                \rectangle{0}{-1}{2};
                \redrectangle{1}{2}{1};
            \end{tikzpicture},
        \end{eqnarray}
        where grey squares correspond to the sites that do not influence the updated
        value (denoted by red-bordered squares).
\end{enumerate}
Together with the restriction to allowed configurations, diagrams from
Equations~\eqref{eq:SEdia1} and~\eqref{eq:SEdia2} completely determine the map $\phi$,
which has in our case support $7$ (i.e.\ $r=3$) and can be summarized as follows,
\begin{eqnarray}
    \phi(s_{-3},s_{-2},s_{-1},s_0,s_1,s_2,s_3)=
    \begin{cases}
        s_1& s_0=s_{-1}=0,\\
        s_{-1}& s_0=s_1=0,\\
        0&s_{-1}=s_1=1,\\
        s_{-3}&s_{-1}=s_0=1,\ s_{1}=0,\\
        s_{3}&s_{1}=s_0=1,\ s_{-1}=1.
    \end{cases}
\end{eqnarray}
Note that $\phi(s_{-3},s_{-2},s_{-1},s_0,s_1,s_2,s_3)$ \emph{does not depend} explicitly 
on the values $s_{-2}$, $s_2$, which means that all the maps applied in the same step
(see Eq.~\eqref{eq:defSEmap} and Fig.~\ref{fig:TSfig3}) commute. The dynamics in space can be
therefore understood as a composition of \emph{strictly local} maps, that are in each
step applied to all $7$-site subconfigurations centred around the sites labeled by temporal indices of the same
parity, and the order in which they are applied is not important. From this point of view,
space-dynamics is, apart from the larger support of local maps, defined in perfect analogy
to the usual evolution in the time-direction.

\subsection{Circuit representation of dynamics}\label{subsec:circuits}
The staggered structure of time evolution, where each time-step is given in terms of
mutually commuting local updates, implies a natural representation in terms of a (classical)
tensor-network. The motivation for this comes from recent advances in constructing 
(typically nonintegrable) solvable models of many-body nonequilibrium dynamics in terms
of quantum circuits, with notable examples being random~\cite{nahum2017quantum,nahum2018operator,chan2018solution,vonKeyserlingk2018operator} and
dual-unitary circuits~\cite{bertini2019exact,gopalakrishnan2019unitary,piroli2020exact,bertini2020operatorI,bertini2020operatorII,bertini2020random}.
These models constitute a family of analytically tractable quantum many-body systems, for which
many previously not-accessible out-of-equilibrium quantities can be obtained exactly. As we
show below, RCA54 admits a convenient representation in terms of local gates, which allows us
to find an equivalent diagrammatic procedure to recover the results from the previous
subsections. 

We start by defining two tensors represented by big and small circles, where
each leg can be either in a state $s=0$  or $s=1$,
\begin{eqnarray}
    \begin{tikzpicture}[scale=0.45,baseline={([yshift=-0.6ex]current bounding box.center)}]
        \gridLine{-1}{0}{1}{0}
        \gridLine{0}{-1}{0}{1}
        \bCircle{0}{0}{colU}
        \node at (-1.5,0) {\scalebox{0.85}{$s_1$}};
        \node at (0,1.5) {\scalebox{0.85}{$s_2$}};
        \node at (1.5,0) {\scalebox{0.85}{$s_3$}};
        \node at (0,-1.5) {\scalebox{0.85}{$s_4$}};
    \end{tikzpicture}= \delta_{s_4,\chi(s_1,s_2,s_3)},\qquad
    \begin{tikzpicture}[scale=0.45,baseline={([yshift=-0.6ex]current bounding box.center)}]
        \def\sqrtThree{1.73205}
        \gridLine{-1}{0}{0}{0}
        \gridLine{-0.5}{(\sqrtThree/2.)}{0}{0}
        \gridLine{0.5}{(\sqrtThree/2.)}{0}{0}
        \gridLine{0}{-1}{0}{0}
        \sCircle{0}{0}{colU};
        \node at (-1.5,0) {\scalebox{0.85}{$s_1$}};
        \node at (-0.75,{0.75*\sqrtThree}) {\scalebox{0.85}{$s_2$}};
        \node at (0.75,{0.75*\sqrtThree}) {\scalebox{0.85}{$s_3$}};
        \node at (0,-1.5) {\scalebox{0.85}{$s_k$}};
        \node[label={[rotate=-105]{{$\cdots$}}},inner sep=0]
        at ({0.05},{-0.025*(1-\sqrtThree/2)}) {};
    \end{tikzpicture} = \prod_{j=1}^{k-1} \delta_{s_{j},s_{j+1}},\qquad k\ge 2.
\end{eqnarray}
The big circle encodes the deterministic update~\eqref{eq:TErules}, while
the small circle forces each one of the  legs to be in the same state.
Using this graphical convention, the local time-evolution operator~\eqref{eq:defU}
can be represented as
\begin{eqnarray}
    U%_{s_1^{\phantom{\prime}} s_2^{\phantom{\prime}} s_3^{\phantom{\prime}}}
    %^{s_1^{\prime} s_2^{\prime}s_3^{\prime}}
    =
    \begin{tikzpicture}[scale=0.45,baseline={([yshift=-0.6ex]current bounding box.center)}]
        \gridLine{-1}{-1}{-1}{1};
        \gridLine{0}{-1}{0}{1};
        \gridLine{1}{-1}{1}{1};
        \prop{-1}{1}{0}{colU};
        %\node at (-1,1.5)  {\scalebox{0.85}{$s_1^{\prime}$}};
        %\node at (0,1.5)   {\scalebox{0.85}{$s_2^{\prime}$}};
        %\node at (1,1.5)   {\scalebox{0.85}{$s_3^{\prime}$}};
        %\node at (-1,-1.5) {\scalebox{0.85}{$s_1$}};
        %\node at (0,-1.5)  {\scalebox{0.85}{$s_2$}};
        %\node at (1,-1.5)  {\scalebox{0.85}{$s_3$}};
    \end{tikzpicture}.%=
    %\delta_{s_1^{\prime},s_1}
    %\delta_{s_2^{\prime},\chi(s_1,s_2,s_3)}
    %\delta_{s_3^{\prime},s_3}
\end{eqnarray}
The definition of the smaller tensors immediately implies that two circles sharing
a leg can be combined into one with more legs. In particular, we note the following identity,
\begin{eqnarray}
    \begin{tikzpicture}[scale=0.45,baseline={([yshift=-0.6ex]current bounding box.center)}]
        \gridLine{0}{-1.5}{0}{1.5};
        \gridLine{-1}{-0.5}{0}{-0.5};
        \gridLine{1}{0.5}{0}{0.5};
        \sCircle{0}{0.5}{colU};
        \sCircle{0}{-0.5}{colU};
    \end{tikzpicture}=
    \begin{tikzpicture}[scale=0.45,baseline={([yshift=-0.6ex]current bounding box.center)}]
        \gridLine{0}{-1.5}{0}{1.5};
        \gridLine{1}{-0.5}{0}{-0.5};
        \gridLine{-1}{0.5}{0}{0.5};
        \sCircle{0}{0.5}{colU};
        \sCircle{0}{-0.5}{colU};
    \end{tikzpicture}=
    \begin{tikzpicture}[scale=0.45,baseline={([yshift=-0.6ex]current bounding box.center)}]
        \gridLine{0}{-1}{0}{1};
        \gridLine{1}{0}{-1}{0};
        \sCircle{0}{0}{colU};
    \end{tikzpicture}.
\end{eqnarray}
This allows us to combine all the local operators applied at the same time-step into 
one horizontal line of small and big balls positioned at sites with the opposite parity,
\begin{eqnarray}
    \Ue=\!
    \begin{tikzpicture}[scale=0.45,baseline={([yshift=-0.6ex]current bounding box.center)}]
        \foreach \x in {1,...,9}{
            \gridLine{\x}{-1}{\x}{1}
        }
        \gridLine{0.5}{-0.25}{1}{-0.25}
        \gridLine{9}{0.25}{9.5}{0.25}
        \node[inner sep=0] at (-0.25,0) {$\cdots$};
        \node[inner sep=0] at (10.25,0) {$\cdots$};
        \sCircle{1}{-0.25}{colU}
        \sCircle{9}{0.25}{colU}

        \foreach \x in {1,5}{
            \prop{\x}{\x+2}{0.25}{colU}
            \prop{\x+2}{\x+4}{-0.25}{colU}
        }
    \end{tikzpicture}\!=\!
    \begin{tikzpicture}[scale=0.45,baseline={([yshift=-0.6ex]current bounding box.center)}]
        \foreach \x in {1,...,9}{
            \gridLine{\x}{-1}{\x}{1}
        }
        \gridLine{0.5}{0}{9.5}{0}
        \node[inner sep=0] at (-0.25,0) {$\cdots$};
        \node[inner sep=0] at (10.25,0) {$\cdots$};
        \foreach \x in {1,3,...,9}{
            \sCircle{\x}{0}{colU}
        }
        \foreach \x in {2,4,...,8}{
            \bCircle{\x}{0}{colU}
        }
    \end{tikzpicture},
\end{eqnarray}
and $\Uo$ takes a similar form with exchanged roles of the two tensors. Interchangeably
applying $\Ue$ and $\Uo$ we obtain a network consisting of small and big tensors, positioned
at sites with the sum of the two coordinates $x+t\equiv 0\pmod{2}$  and $x+t\equiv 1 \pmod{2}$, 
respectively, as is diagrammatically shown in Fig.~\ref{fig:TSfig4}.

\begin{figure}
  \centering
  \includegraphics[width=0.4\textwidth]{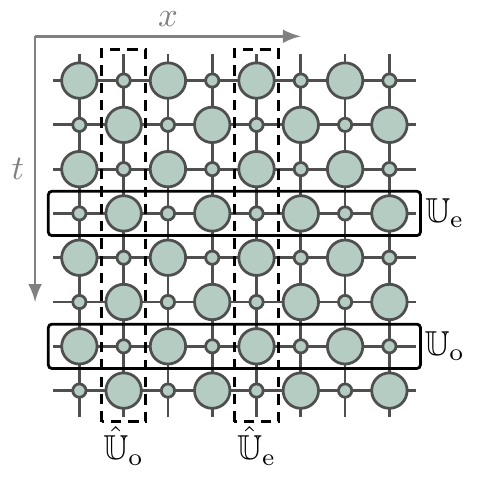}
    \caption{\label{fig:TSfig4} Circuit representation of RCA54 dynamics. Time-evolution can
    be expressed as the tensor-network of small and big tensors positioned at sites where
    the sum of space and time coordinate is even and odd respectively. The solid rectangles
    denote two instances of $\Ue$ and $\Uo$, while the dashed rectangles show analogously
    defined operators in the space direction, $\tUe$ amd $\tUo$.
    }
\end{figure}

\subsubsection{Space dynamics in terms of non-deterministic dual gates}
The symmetry of the tensor-network suggests a natural definition of space evolution
in terms of the \emph{dual} tensor $\hat{U}$, defined as,
\begin{eqnarray}
    \hat{U}%_{s_1^{\phantom{\prime}} s_2^{\phantom{\prime}} s_3^{\phantom{\prime}}}
    %^{s_1^{\prime} s_2^{\prime} s_3^{\prime}}
    =
    \begin{tikzpicture}[scale=0.45,baseline={([yshift=-0.6ex]current bounding box.center)}]
        \gridLine{-1}{1}{1}{1}
        \gridLine{-1}{0}{1}{0}
        \gridLine{-1}{-1}{1}{-1}
        \tprop{0}{1}{-1}{colU}
        %\node at (-1.5,1)  {\scalebox{0.8}{$s_3$}};
        %\node at (-1.5,0)  {\scalebox{0.8}{$s_2$}};
        %\node at (-1.5,-1) {\scalebox{0.8}{$s_1$}};
        %\node at (1.5,1)  {\scalebox{0.8}{$s_3^{\prime}$}};
        %\node at (1.5,0)  {\scalebox{0.8}{$s_2^{\prime}$}};
        %\node at (1.5,-1) {\scalebox{0.8}{$s_1^{\prime}$}};
    \end{tikzpicture}=
    \begin{bmatrix}
        1 & & & & & & & \\
        &0& &1& & & & \\
        & &0& & & & & \\
        &1& &1& & & & \\
        & & & &0& &1& \\
        & & & & &1& & \\
        & & & &1& &1& \\
        & & & & & & &0
    \end{bmatrix}.
\end{eqnarray}
Operators $\tUe$ and $\tUo$, describing the full one-step evolution in space,
are in analogy to $\Ue$ and $\Uo$ given in terms of products of $\hat{U}$
centred around the sites on the dual lattice with even and odd parity
respectively, as is diagrammatically shown in Fig.~\ref{fig:TSfig4}. One
quickly realizes that $\hat{U}$ is \emph{not} deterministic, which is expected,
since we know that the deterministic map has the support $7$. To understand how
this happens, we define the following $3$-site projector to the space of allowed
configurations,
\begin{eqnarray}
P=\begin{tikzpicture}[scale=0.45,baseline={([yshift=-0.6ex]current bounding box.center)}]
        \gridLine{-1}{1}{1}{1}
        \gridLine{-1}{0}{1}{0}
        \gridLine{-1}{-1}{1}{-1}
        \gridLine{0}{-1}{0}{1}
        \sCircle{0}{1}{colP}
        \sCircle{0}{0}{colP}
        \sCircle{0}{-1}{colP}
        %\node at (-1.5,1)  {\scalebox{0.8}{$s_3$}};
        %\node at (-1.5,0)  {\scalebox{0.8}{$s_2$}};
        %\node at (-1.5,-1) {\scalebox{0.8}{$s_1$}};
        %\node at (1.5,1)  {\scalebox{0.8}{$s_3^{\prime}$}};
        %\node at (1.5,0)  {\scalebox{0.8}{$s_2^{\prime}$}};
        %\node at (1.5,-1) {\scalebox{0.8}{$s_1^{\prime}$}};
    \end{tikzpicture},\qquad
    P_{s_1^{\phantom{\prime}} s_2^{\phantom{\prime}} s_3^{\phantom{\prime}}}
    ^{s_1^{\prime} s_2^{\prime} s_3^{\prime}}
    =
    \delta_{s_1^{\phantom{\prime}},s_1^{\prime}}
    \delta_{s_2^{\phantom{\prime}},s_2^{\prime}}
    \delta_{s_3^{\phantom{\prime}},s_3^{\prime}}
    \Big(1-\delta_{s_1,s_3}\Big).
\end{eqnarray}
We note that the local three-site dual operator projects to the same subspace as $P$, 
\begin{eqnarray}\label{eq:projRel}
    \hat{U}=P\hat{U}=\hat{U} P,\qquad
    \begin{tikzpicture}[scale=0.45,baseline={([yshift=-0.6ex]current bounding box.center)}]
        \foreach \t in {-1,...,1}{\gridLine{-1}{\t}{1}{\t}}
        \gridLine{0}{-1}{0}{1}
        \sCircle{0}{1}{colU}
        \bCircle{0}{0}{colU}
        \sCircle{0}{-1}{colU}
    \end{tikzpicture}=
    \begin{tikzpicture}[scale=0.45,baseline={([yshift=-0.6ex]current bounding box.center)}]
        \foreach \t in {-1,...,1}{\gridLine{-1}{\t}{1.5}{\t}}
        \gridLine{0}{-1}{0}{1}
        \gridLine{1}{-1}{1}{1}
        \sCircle{0}{1}{colU}
        \bCircle{0}{0}{colU}
        \sCircle{0}{-1}{colU}
        \sCircle{1}{1}{colP}
        \sCircle{1}{0}{colP}
        \sCircle{1}{-1}{colP}
    \end{tikzpicture}=
    \begin{tikzpicture}[scale=0.45,baseline={([yshift=-0.6ex]current bounding box.center)}]
        \foreach \t in {-1,...,1}{\gridLine{-1.5}{\t}{1}{\t}}
        \gridLine{0}{-1}{0}{1}
        \gridLine{-1}{-1}{-1}{1}
        \sCircle{0}{1}{colU}
        \bCircle{0}{0}{colU}
        \sCircle{0}{-1}{colU}
        \sCircle{-1}{1}{colP}
        \sCircle{-1}{0}{colP}
        \sCircle{-1}{-1}{colP}
    \end{tikzpicture},
\end{eqnarray}
and additionally, $P$ and $\hat{U}$ commute if they overlap for at most one site,
while they do not commute if the overlap is $2$,
\begin{eqnarray}\fl
    \begin{tikzpicture}[scale=0.45,baseline={([yshift=-0.6ex]current bounding box.center)}]
        \foreach \t in {-1,...,3}{\gridLine{-1}{\t}{1}{\t}}
        \gridLine{0}{-1}{0}{1}
        \gridLine{-0.5}{1}{-0.5}{3}
        \sCircle{0}{1}{colU}
        \bCircle{0}{0}{colU}
        \sCircle{0}{-1}{colU}
        \sCircle{-0.5}{1}{colP}
        \sCircle{-0.5}{2}{colP}
        \sCircle{-0.5}{3}{colP}
    \end{tikzpicture}=
    \begin{tikzpicture}[scale=0.45,baseline={([yshift=-0.6ex]current bounding box.center)}]
        \foreach \t in {-1,...,3}{\gridLine{-1}{\t}{1}{\t}}
        \gridLine{0}{-1}{0}{1}
        \gridLine{0.5}{1}{0.5}{3}
        \sCircle{0}{1}{colU}
        \bCircle{0}{0}{colU}
        \sCircle{0}{-1}{colU}
        \sCircle{0.5}{1}{colP}
        \sCircle{0.5}{2}{colP}
        \sCircle{0.5}{3}{colP}
    \end{tikzpicture},\qquad
    \begin{tikzpicture}[scale=0.45,baseline={([yshift=-0.6ex]current bounding box.center)}]
        \foreach \t in {-3,...,1}{\gridLine{-1}{\t}{1}{\t}}
        \gridLine{0}{-1}{0}{1}
        \gridLine{-0.5}{-1}{-0.5}{-3}
        \sCircle{0}{1}{colU}
        \bCircle{0}{0}{colU}
        \sCircle{0}{-1}{colU}
        \sCircle{-0.5}{-1}{colP}
        \sCircle{-0.5}{-2}{colP}
        \sCircle{-0.5}{-3}{colP}
    \end{tikzpicture}=
    \begin{tikzpicture}[scale=0.45,baseline={([yshift=-0.6ex]current bounding box.center)}]
        \foreach \t in {-3,...,1}{\gridLine{-1}{\t}{1}{\t}}
        \gridLine{0}{-1}{0}{1}
        \gridLine{0.5}{-1}{0.5}{-3}
        \sCircle{0}{1}{colU}
        \bCircle{0}{0}{colU}
        \sCircle{0}{-1}{colU}
        \sCircle{0.5}{-1}{colP}
        \sCircle{0.5}{-2}{colP}
        \sCircle{0.5}{-3}{colP}
    \end{tikzpicture},\qquad
    \begin{tikzpicture}[scale=0.45,baseline={([yshift=-0.6ex]current bounding box.center)}]
        \foreach \t in {-1,...,2}{\gridLine{-1.5}{\t}{1}{\t}}
        \gridLine{0}{-1}{0}{1}
        \gridLine{-1}{0}{-1}{2}
        \sCircle{0}{1}{colU}
        \bCircle{0}{0}{colU}
        \sCircle{0}{-1}{colU}
        \sCircle{-1}{0}{colP}
        \sCircle{-1}{1}{colP}
        \sCircle{-1}{2}{colP}
    \end{tikzpicture}\neq
    \begin{tikzpicture}[scale=0.45,baseline={([yshift=-0.6ex]current bounding box.center)}]
        \foreach \t in {-1,...,2}{\gridLine{-1}{\t}{1.5}{\t}}
        \gridLine{0}{-1}{0}{1}
        \gridLine{1}{0}{1}{2}
        \sCircle{0}{1}{colU}
        \bCircle{0}{0}{colU}
        \sCircle{0}{-1}{colU}
        \sCircle{1}{0}{colP}
        \sCircle{1}{1}{colP}
        \sCircle{1}{2}{colP}
    \end{tikzpicture},\qquad
    \begin{tikzpicture}[scale=0.45,baseline={([yshift=-0.6ex]current bounding box.center)}]
        \foreach \t in {-2,...,1}{\gridLine{-1.5}{\t}{1}{\t}}
        \gridLine{0}{-1}{0}{1}
        \gridLine{-1}{0}{-1}{-2}
        \sCircle{0}{1}{colU}
        \bCircle{0}{0}{colU}
        \sCircle{0}{-1}{colU}
        \sCircle{-1}{0}{colP}
        \sCircle{-1}{-1}{colP}
        \sCircle{-1}{-2}{colP}
    \end{tikzpicture}\neq
    \begin{tikzpicture}[scale=0.45,baseline={([yshift=-0.6ex]current bounding box.center)}]
        \foreach \t in {-2,...,1}{\gridLine{-1}{\t}{1.5}{\t}}
        \gridLine{0}{-1}{0}{1}
        \gridLine{1}{0}{1}{-2}
        \sCircle{0}{1}{colU}
        \bCircle{0}{0}{colU}
        \sCircle{0}{-1}{colU}
        \sCircle{1}{0}{colP}
        \sCircle{1}{-1}{colP}
        \sCircle{1}{-2}{colP}
    \end{tikzpicture}.
\end{eqnarray}
This allows us to show that the evolution in space can be equivalently given in terms
of the operators reduced to the space of allowed configurations,
\begin{eqnarray}\fl
    \begin{tikzpicture}[scale=0.45,baseline={([yshift=-0.6ex]current bounding box.center)}]
        \foreach \t in {1,...,4}{\gridLine{-0.6}{\t}{3.6}{\t}}
        \foreach \x in {0,...,3}{\gridLine{\x}{0.4}{\x}{4.6}}
        \foreach \t in {1,3}{
            \foreach \x in {0,2}{\bCircle{\x}{\t}{colU}}
            \foreach \x in {1,3}{\sCircle{\x}{\t}{colU}}
        }
        \foreach \t in {2,4}{
            \foreach \x in {0,2}{\sCircle{\x}{\t}{colU}}
            \foreach \x in {1,3}{\bCircle{\x}{\t}{colU}}
        }
        \foreach \t in {1,...,4}{\gridLine{6.4}{\t}{16.6}{\t}}
        \foreach \x in {7,10,13,16}{\gridLine{\x}{0.4}{\x}{4.6}}
        \foreach \t in {1,3}{
            \foreach \x in {7,13}{\bCircle{\x}{\t}{colU}}
            \foreach \x in {10,16}{\sCircle{\x}{\t}{colU}}
        }
        \foreach \t in {2,4}{
            \foreach \x in {7,13}{\sCircle{\x}{\t}{colU}}
            \foreach \x in {10,16}{\bCircle{\x}{\t}{colU}}
        }
        \foreach \x in {9.125,10.875,15.125}
        {
            \gridLine{\x-0.175}{0.4}{\x-0.175}{1}
            \gridLine{\x+0.175}{1}{\x+0.175}{3}
            \gridLine{\x-0.175}{3}{\x-0.175}{4.6}
            \sCircle{\x-0.175}{1}{colP}
            \sCircle{\x+0.175}{1}{colP}
            \sCircle{\x+0.175}{2}{colP}
            \sCircle{\x+0.175}{3}{colP}
            \sCircle{\x-0.175}{3}{colP}
            \sCircle{\x-0.175}{4}{colP}
        }
        \foreach \x in {7.875,12.125,13.875}
        {
            \gridLine{\x-0.175}{0.4}{\x-0.175}{2}
            \gridLine{\x+0.175}{2}{\x+0.175}{4}
            \gridLine{\x-0.175}{4}{\x-0.175}{4.6}
            \sCircle{\x-0.175}{1}{colP}
            \sCircle{\x-0.175}{2}{colP}
            \sCircle{\x+0.175}{2}{colP}
            \sCircle{\x+0.175}{3}{colP}
            \sCircle{\x+0.175}{4}{colP}
            \sCircle{\x-0.175}{4}{colP}
        }
        \foreach \t in {1,...,4}{\gridLine{19.4}{\t}{29.6}{\t}}
        \foreach \x in {20,23,26,29}{\gridLine{\x}{0.4}{\x}{4.6}}
        \foreach \t in {1,3}{
            \foreach \x in {20,26}{\bCircle{\x}{\t}{colU}}
            \foreach \x in {23,29}{\sCircle{\x}{\t}{colU}}
        }
        \foreach \t in {2,4}{
            \foreach \x in {20,26}{\sCircle{\x}{\t}{colU}}
            \foreach \x in {23,29}{\bCircle{\x}{\t}{colU}}
        }
        \foreach \x in {22.125,23.875,28.125}
        {
            \gridLine{\x-0.175}{0.4}{\x-0.175}{2}
            \gridLine{\x+0.175}{2}{\x+0.175}{4}
            \gridLine{\x-0.175}{4}{\x-0.175}{4.6}
            \sCircle{\x-0.175}{1}{colP}
            \sCircle{\x-0.175}{2}{colP}
            \sCircle{\x+0.175}{2}{colP}
            \sCircle{\x+0.175}{3}{colP}
            \sCircle{\x+0.175}{4}{colP}
            \sCircle{\x-0.175}{4}{colP}

        }
        \foreach \x in {20.875,25.125,26.875}
        {
            \gridLine{\x-0.175}{0.4}{\x-0.175}{1}
            \gridLine{\x+0.175}{1}{\x+0.175}{3}
            \gridLine{\x-0.175}{3}{\x-0.175}{4.6}
            \sCircle{\x-0.175}{1}{colP}
            \sCircle{\x+0.175}{1}{colP}
            \sCircle{\x+0.175}{2}{colP}
            \sCircle{\x+0.175}{3}{colP}
            \sCircle{\x-0.175}{3}{colP}
            \sCircle{\x-0.175}{4}{colP}
        }
        \node at (1,5) {\scalebox{0.8}{$\tUe$}};
        \node at (2,5) {\scalebox{0.8}{$\tUo$}};

        \node at (9.125,5) {\scalebox{0.8}{$P_{\mathrm{e}}$}};
        \node at (10,5) {\scalebox{0.8}{$\tUe$}};
        \node at (10.875,5) {\scalebox{0.8}{$P_{\mathrm{e}}$}};
        
        \node at (12.125,5) {\scalebox{0.8}{$P_{\mathrm{o}}$}};
        \node at (13,5) {\scalebox{0.8}{$\tUo$}};
        \node at (13.875,5) {\scalebox{0.8}{$P_{\mathrm{o}}$}};

        \node at (22.125,5) {\scalebox{0.8}{$P_{\mathrm{o}}$}};
        \node at (23,5) {\scalebox{0.8}{$\tUe$}};
        \node at (23.875,5) {\scalebox{0.8}{$P_{\mathrm{o}}$}};
        
        \node at (25.125,5) {\scalebox{0.8}{$P_{\mathrm{e}}$}};
        \node at (26,5) {\scalebox{0.8}{$\tUo$}};
        \node at (26.875,5) {\scalebox{0.8}{$P_{\mathrm{e}}$}};
        \draw [decorate,decoration={brace,amplitude=5pt},xshift=0pt,yshift=0pt]
        (11.125,0.375) -- (8.875,0.375) node [midway,yshift=-12]
        {\scalebox{0.8}{$\tUe$}};
        \draw [decorate,decoration={brace,amplitude=5pt},xshift=0pt,yshift=0pt]
        (14.125,0.375) -- (11.875,0.375) node [midway,yshift=-12]
        {\scalebox{0.8}{$\tUo$}};
        \draw [decorate,decoration={brace,amplitude=5pt},xshift=0pt,yshift=0pt]
        (24.125,0.375) -- (21.875,0.375) node [midway,yshift=-12]
        {\scalebox{0.8}{$\tdUe$}};
        \draw [decorate,decoration={brace,amplitude=5pt},xshift=0pt,yshift=0pt]
        (27.125,0.375) -- (24.875,0.375) node [midway,yshift=-12]
        {\scalebox{0.8}{$\tdUo$}};
        \draw[thick,-latex] (4.25,2.5) -- (5.75,2.5);
        \draw[thick,-latex] (17.25,2.5) -- (18.75,2.5);
    \end{tikzpicture},
\end{eqnarray}
where $P_{\mathrm{e/o}}$ are defined as products of $P$ on even/odd sites (in
analogy to $\tUeo$). To get from the first to the second diagram, we use the
relation~\eqref{eq:projRel}, and to obtain the third diagram we note that all
projectors $P$ commute amongst themselves. The operators $\tdUe$, $\tdUo$
are precisely the space-evolution maps on the reduced subspace, and from the
previous subsection we know that they can be expressed as a composition of
$7$-site deterministic maps. One can show this independently by expressing
the $7$-site deterministic gates in terms of the non-deterministic gate~$\hat{U}$
squeezed between the relevant projectors, and then prove that the two descriptions
are equivalent by utilizing appropriate local algebraic relations fulfilled
by the gates and projectors. The full construction with the proof is reported in
Ref.~\cite{klobas2020space}.

\subsubsection{Tensor-network representation of multi-time correlation functions}
To demonstrate the usefulness of the circuit formulation of RCA54, we return to
the multi-time correlation functions at the same point and show an independent
derivation of time-states. For simplicity we will only consider the maximum-entropy
state, but a similar analysis can be done for a stationary state of the form introduced
in Subsection~\ref{sec:PBmps}.

Before starting, we recall that by definition the expectation value of an observable $a$ 
in a probability distribution on $n$ sites $\vec{p}$ is a sum over all the configurations
of the products of the probability and the value of the observable in the configurations,
\begin{eqnarray}\label{eq:defMultiCorrsOmega}
    \expval{a}_{\vec{p}}=\sum_{\ul{s}} a(\ul{s}) p_{\ul{s}}
    = \left(\vec{\omega}^T\right)^{\otimes n}\cdot A\cdot \vec{p},\qquad 
    \vec{\omega}=\begin{bmatrix}1 \\1\end{bmatrix},
\end{eqnarray}
where we introduced the unnormalized one-site maximum-entropy state
$\vec{\omega}$ and $A$ is a $2^{n}\times 2^{n}$ diagonal matrix
with the values $a(\ul{s})$. Note that the maximum entropy state on $n$
sites corresponds to $\vec{p}=2^{-n} \vec{\omega}^{\otimes n}$.  This allows
us to introduce the multi-time correlation function of one-site observables
$a_1,a_2,\ldots,a_{k}$ at times $t=0,1,\ldots k-1$ as,
\begin{eqnarray}\label{eq:multiCorrDiagram1}
    C_{a_1,\ldots,a_{k}}=
    \frac{1}{2^{n}}\ 
    \begin{tikzpicture}[scale=0.55,baseline={([yshift=0.4ex]current bounding box.center)}]
        \foreach \x in {1,...,10}{
            \gridLine{\x}{0}{\x}{9}
            \ME{\x}{0}
            \ME{\x}{9}
        }
        \foreach \t in {1,...,8}{
            \gridLine{0.5}{\t}{10.5}{\t}
            \leftHook{0.5}{\t}
            \rightHook{10.5}{\t}
        }
        \foreach \t in {1,3,...,8}
        {
            \foreach \x in {1,3,...,10}{\sCircle{\x}{\t}{colU}}
            \foreach \x in {2,4,...,10}{\bCircle{\x}{\t}{colU}}
        }
        \foreach \t in {2,4,...,8}
        {
            \foreach \x in {1,3,...,10}{\bCircle{\x}{\t}{colU}}
            \foreach \x in {2,4,...,10}{\sCircle{\x}{\t}{colU}}
        }
        \foreach \t in {0,2,...,7}{\obs{6}{(\t+0.375)}}
        \obs{6}{(8+0.5)}
        \foreach \t in {1,3,...,7}{\obs{5}{(\t+0.375)}}
        \node[inner sep=0] at (6.6,8.6) {\scalebox{0.7}{$a_1$}};
        \node[inner sep=0] at (5.5,7.5) {\scalebox{0.7}{$a_2$}};
        \node[inner sep=0] at (6.5,6.5) {\scalebox{0.7}{$a_3$}};
        \node[inner sep=0] at (6.5,0.4) {\scalebox{0.7}{$a_k$}};
        \draw [decorate,decoration={brace,amplitude=5pt},xshift=0pt,yshift=0pt]
        (10.75,8.125) -- (10.75,0.625) node [midway,xshift=20,inner sep=0]{\scalebox{1}{$k-1$}};
        \draw [decorate,decoration={brace,amplitude=5pt},xshift=0pt,yshift=0pt]
        (10.25,-0.125) -- (0.75,-0.125) node [midway,yshift=-10,inner sep=0]{\scalebox{1}{$n$}};
    \end{tikzpicture}\,,
\end{eqnarray}
where $\begin{tikzpicture}[scale=0.45,baseline={([yshift=-0.6ex]current bounding box.center)}]
\gridLine{0}{0.5}{0}{-0.5} \obs{0}{0} \end{tikzpicture}$ represents a one-site
observable and
$\begin{tikzpicture} [scale=0.45,baseline={([yshift=-0.6ex]current bounding box.center)}]
\gridLine{0}{0.75}{0}{0} \ME{0}{0} \end{tikzpicture}=\vec{\omega}$. For simplicity we will
assume that $k\equiv 1\pmod{2}$, but the same can be repeated for even $k$.
Note that these correlation functions can be directly related to time-states $\vec{q}$ as
\begin{eqnarray}\label{eq:defRelTSs}
    C_{a_1,a_2,\ldots,a_{2m}} = \smashoperator{\sum_{s_1,s_2,\ldots,s_{2m}}}
    q_{s_1,\ldots s_{2m}} \prod_{j=1}^{2m} a_j(s_j).
\end{eqnarray}

To reduce the diagrammatic representation of $C_{a_1,\ldots,a_k}$, we first note that
the local time-evolution operator $U$ is deterministic and as a consequence
the maximum-entropy state is invariant under it (as well as under its transpose),
\begin{eqnarray}\label{eq:detUrels}
    \begin{tikzpicture}[scale=0.45,baseline={([yshift=-0.6ex]current bounding box.center)}]
        \gridLine{-1}{-1}{-1}{1}
        \gridLine{0}{-1}{0}{1}
        \gridLine{1}{-1}{1}{1}
        \gridLine{-1}{0}{1}{0}
        \ME{-1}{1}
        \ME{0}{1}
        \ME{1}{1}
        \sCircle{-1}{0}{colU}
        \bCircle{0}{0}{colU}
        \sCircle{1}{0}{colU}
    \end{tikzpicture}=
    \begin{tikzpicture}[scale=0.45,baseline={([yshift=-0.6ex]current bounding box.center)}]
        \gridLine{-1}{1}{-1}{0}
        \gridLine{0}{1}{0}{0}
        \gridLine{1}{1}{1}{0}
        \ME{-1}{1}
        \ME{0}{1}
        \ME{1}{1}
    \end{tikzpicture},\qquad
    \begin{tikzpicture}[scale=0.45,baseline={([yshift=-0.6ex]current bounding box.center)}]
        \gridLine{-1}{-1}{-1}{1}
        \gridLine{0}{-1}{0}{1}
        \gridLine{1}{-1}{1}{1}
        \gridLine{-1}{0}{1}{0}
        \ME{-1}{-1}
        \ME{0}{-1}
        \ME{1}{-1}
        \sCircle{-1}{0}{colU}
        \bCircle{0}{0}{colU}
        \sCircle{1}{0}{colU}
    \end{tikzpicture}=
    \begin{tikzpicture}[scale=0.45,baseline={([yshift=-0.6ex]current bounding box.center)}]
        \gridLine{-1}{-1}{-1}{0}
        \gridLine{0}{-1}{0}{0}
        \gridLine{1}{-1}{1}{0}
        \ME{-1}{-1}
        \ME{0}{-1}
        \ME{1}{-1}
    \end{tikzpicture}.
\end{eqnarray}
Therefore the circuit~\eqref{eq:multiCorrDiagram1} can be simplified by removing gates
connected to $\vec{\omega}$, which reduces it to a tilted square shape. Additionally,
since the observables are diagonal, they commute with the small tensor,
\begin{eqnarray}
    \begin{tikzpicture}[scale=0.45,baseline={([yshift=-0.6ex]current bounding box.center)}]
        \gridLine{0}{1}{0}{-1}
        \gridLine{1}{0}{-1}{0}
        \sCircle{0}{0}{colU}
        \obs{0}{0.5}
    \end{tikzpicture}=
    \begin{tikzpicture}[scale=0.45,baseline={([yshift=-0.6ex]current bounding box.center)}]
        \gridLine{0}{1}{0}{-1}
        \gridLine{1}{0}{-1}{0}
        \sCircle{0}{0}{colU}
        \obs{0.5}{0}
    \end{tikzpicture}=
    \begin{tikzpicture}[scale=0.45,baseline={([yshift=-0.6ex]current bounding box.center)}]
        \gridLine{0}{1}{0}{-1}
        \gridLine{1}{0}{-1}{0}
        \sCircle{0}{0}{colU}
        \obs{0}{-0.5}
    \end{tikzpicture}=
    \begin{tikzpicture}[scale=0.45,baseline={([yshift=-0.6ex]current bounding box.center)}]
        \gridLine{0}{1}{0}{-1}
        \gridLine{1}{0}{-1}{0}
        \sCircle{0}{0}{colU}
        \obs{-0.5}{0}
    \end{tikzpicture},
\end{eqnarray}
and we can recast the circuit in terms of dual gates~$\hat{U}$ as
\begin{eqnarray}\label{eq:multiCorrDiagram2}
    C_{a_1,\ldots,a_{k}}=
    \frac{1}{2^{k}}\ 
    \begin{tikzpicture}[scale=0.55,baseline={([yshift=-0.6ex]current bounding box.center)}]
        \gridLine{3}{3}{3}{5}
        \gridLine{4}{2}{4}{6}
        \gridLine{5}{1}{5}{7}
        \gridLine{6}{0}{6}{8}
        \gridLine{7}{1}{7}{7}
        \gridLine{8}{2}{8}{6}
        \gridLine{9}{3}{9}{5}
        \ME{2}{3}
        \ME{2}{4}
        \ME{2}{5}
        \ME{3}{2}
        \ME{3}{6}
        \ME{4}{1}
        \ME{4}{7}
        \ME{5}{0}
        \ME{5}{8}
        \ME{7}{0}
        \ME{7}{8}
        \ME{8}{1}
        \ME{8}{7}
        \ME{9}{2}
        \ME{9}{6}
        \ME{10}{3}
        \ME{10}{4}
        \ME{10}{5}

        \gridLine{5}{8}{7}{8}
        \gridLine{4}{7}{8}{7}
        \gridLine{3}{6}{9}{6}
        \gridLine{2}{5}{10}{5}
        \gridLine{2}{4}{10}{4}
        \gridLine{2}{3}{10}{3}
        \gridLine{3}{2}{9}{2}
        \gridLine{4}{1}{8}{1}
        \gridLine{5}{0}{7}{0}
        \sCircle{6}{8}{colU}
        \sCircle{5}{7}{colU}
        \bCircle{6}{7}{colU}
        \sCircle{7}{7}{colU}
        \sCircle{4}{6}{colU}
        \bCircle{5}{6}{colU}
        \sCircle{6}{6}{colU}
        \bCircle{7}{6}{colU}
        \sCircle{8}{6}{colU}

        \sCircle{3}{5}{colU}
        \bCircle{4}{5}{colU}
        \sCircle{5}{5}{colU}
        \bCircle{6}{5}{colU}
        \sCircle{7}{5}{colU}
        \bCircle{8}{5}{colU}
        \sCircle{9}{5}{colU}

        \bCircle{3}{4}{colU}
        \sCircle{4}{4}{colU}
        \bCircle{5}{4}{colU}
        \sCircle{6}{4}{colU}
        \bCircle{7}{4}{colU}
        \sCircle{8}{4}{colU}
        \bCircle{9}{4}{colU}

        \sCircle{6}{0}{colU}
        \sCircle{5}{1}{colU}
        \bCircle{6}{1}{colU}
        \sCircle{7}{1}{colU}
        \sCircle{4}{2}{colU}
        \bCircle{5}{2}{colU}
        \sCircle{6}{2}{colU}
        \bCircle{7}{2}{colU}
        \sCircle{8}{2}{colU}

        \sCircle{3}{3}{colU}
        \bCircle{4}{3}{colU}
        \sCircle{5}{3}{colU}
        \bCircle{6}{3}{colU}
        \sCircle{7}{3}{colU}
        \bCircle{8}{3}{colU}
        \sCircle{9}{3}{colU}
        \foreach \t in {1,3,...,7}{\obs{5.375}{\t}}
        \foreach \t in {2,4,...,7}{\obs{5.625}{\t}}
        \obs{5.5}{0}
        \obs{5.5}{8}
    \end{tikzpicture}\,.
\end{eqnarray}
So far we have not taken into account any special property of RCA54 and the
computational complexity of such an object would generally still grow
exponentially. However, we know that in our case the complexity is constant,
therefore there must be some simplification occurring.

If the gates $\hat{U}$ were deterministic, a relation analogous to
\eqref{eq:detUrels} would be satisfied in the space direction and the
correlation functions would trivialize. This does not hold precisely in our
case, but the fact that the dual evolution is deterministic on the reduced
subspace suggests that the gates $\hat{U}$ have some additional structure.
In particular, we find that the following two few-site algebraic relations are
fulfilled,
\begin{eqnarray}\label{eq:dualFewSiteRelations}
    \begin{tikzpicture}[scale=0.45,baseline={([yshift=-0.6ex]current bounding box.center)}]
        \foreach \t in {-2,...,2}{
            \gridLine{0}{\t}{3.75}{\t}
            \ME{0}{\t}
        }
        \tprop{1}{1}{-1}{colU}
        \tprop{2}{2}{0}{colU}
        \tprop{2}{-2}{0}{colU}
        \tproj{3}{-1}{1}{colP}
    \end{tikzpicture}=
    \begin{tikzpicture}[scale=0.45,baseline={([yshift=-0.6ex]current bounding box.center)}]
        \foreach \t in {-2,...,2}{
            \gridLine{1}{\t}{3.75}{\t}
            \ME{1}{\t}
        }
        \tprop{2}{2}{0}{colU}
        \tprop{2}{-2}{0}{colU}
        \tproj{3}{-1}{1}{colP}
    \end{tikzpicture},\qquad
    \begin{tikzpicture}[scale=0.45,baseline={([yshift=-0.6ex]current bounding box.center)}]
        \foreach \t in {-3,...,2}{
            \gridLine{0}{\t}{3.75}{\t}
            \ME{0}{\t}
        }
        \tprop{1}{2}{0}{colU}
        \tprop{1}{0}{-2}{colU}
        \tprop{2}{-3}{-1}{colU}
        \tprop{2}{-1}{1}{colU}
        \tproj{3}{-2}{0}{colP}
    \end{tikzpicture}=
    \begin{tikzpicture}[scale=0.45,baseline={([yshift=-0.6ex]current bounding box.center)}]
        \foreach \t in {-3,...,2}{
            \gridLine{0}{\t}{3.75}{\t}
            \ME{0}{\t}
        }
        \tprop{1}{2}{0}{colU}
        \tprop{2}{-3}{-1}{colU}
        \tprop{2}{-1}{1}{colU}
        \tproj{3}{-2}{0}{colP}
    \end{tikzpicture}, 
\end{eqnarray}
and since $\hat{U}^T=\hat{U}$, an analogous set of left-right flipped identities holds. These relations
are used, together with the projector identity~\eqref{eq:projRel}, to remove dual gates from
the diagram~\eqref{eq:multiCorrDiagram2} layer by layer, starting at the edges and working our
way inwards. We are left with only two layers of gates, one on each side of the column of
observables,
\begin{eqnarray}\label{eq:multiCorrDiagram3}
    C_{a_1,a_2,\ldots,a_k}=
    \frac{1}{2^{k}}\ 
    \begin{tikzpicture}[scale=0.45,baseline={([yshift=-0.6ex]current bounding box.center)}]
        \foreach \t in {1,...,4,7,8,9}{
            \gridLine{0}{\t}{3.5}{\t}
            \ME{0}{\t}
            \ME{3.5}{\t}
            \obs{1.75}{\t}
        }
        \gridLine{1}{2}{1}{4.6}
        \gridLine{2.5}{1}{2.5}{4.6}
        \gridLine{1}{6.4}{1}{8}
        \gridLine{2.5}{6.4}{2.5}{9}
        \foreach \t in {1,3,7,9}{\sCircle{2.5}{\t}{colU}}
        \foreach \t in {2,4,8}{\bCircle{2.5}{\t}{colU}}
        \foreach \t in {2,4,8}{\sCircle{1}{\t}{colU}}
        \foreach \t in {3,7}{\bCircle{1}{\t}{colU}}
        \node[inner sep=0] at (1.75,5.75) {\scalebox{1.25}{$\mathbf{\vdots}$}};
        \draw[decorate,decoration={brace,amplitude=5pt},xshift=0pt,yshift=0pt]
        (3.75,9.125) -- (3.75,0.875) node [midway,xshift=10] {\scalebox{0.8}{$k$}};
    \end{tikzpicture},
\end{eqnarray}
at which point the algebraic relations~\eqref{eq:dualFewSiteRelations} can no
longer be applied.  Note that here we implicitly used the fact that observables
are diagonal and therefore commute with projectors $P$.

By definition (cf.\ \eqref{eq:defRelTSs}) the fully simplified tensor
network~\eqref{eq:multiCorrDiagram3} gives the following diagrammatic representation
of components of the time-state~$\vec{q}$ corresponding to the maximum entropy state,
\begin{eqnarray}\label{eq:defTSs2}
    q_{s_1,s_2,s_3,\ldots,s_{k}}=
    \frac{1}{2^{k}}\ 
    \begin{tikzpicture}[scale=0.45,baseline={([yshift=-0.6ex]current bounding box.center)}]
        \foreach \t in {1,2,5,6,7,8,9}{
            \gridLine{0}{\t}{2}{\t}
            \gridLine{5}{\t}{7}{\t}
            \ME{0}{\t}
            \ME{7}{\t}
        }
        \node[inner sep=0] at (2.5,9) {\scalebox{0.85}{$s_1$}};
        \node[inner sep=0] at (4.5,9) {\scalebox{0.85}{$s_1$}};
        \node[inner sep=0] at (2.5,8) {\scalebox{0.85}{$s_2$}};
        \node[inner sep=0] at (4.5,8) {\scalebox{0.85}{$s_2$}};
        \node[inner sep=0] at (2.5,7) {\scalebox{0.85}{$s_3$}};
        \node[inner sep=0] at (4.5,7) {\scalebox{0.85}{$s_3$}};
        \node[inner sep=0] at (2.5,6) {\scalebox{0.85}{$s_4$}};
        \node[inner sep=0] at (4.5,6) {\scalebox{0.85}{$s_4$}};
        \node[inner sep=0] at (2.5,5) {\scalebox{0.85}{$s_5$}};
        \node[inner sep=0] at (4.5,5) {\scalebox{0.85}{$s_5$}};
        \node[inner sep=0] at (2.75,2) {\scalebox{0.8}{$s_{k-1}$}};
        \node[inner sep=0] at (4.25,2) {\scalebox{0.8}{$s_{k-1}$}};
        \node[inner sep=0] at (2.5,1) {\scalebox{0.85}{$s_{k}$}};
        \node[inner sep=0] at (4.5,1) {\scalebox{0.85}{$s_{k}$}};
        \gridLine{1}{2}{1}{2.6}
        \gridLine{6}{1}{6}{2.6}
        \gridLine{1}{4.4}{1}{8}
        \gridLine{6}{4.4}{6}{9}
        \foreach \t in {1,5,7,9}{\sCircle{6}{\t}{colU}}
        \foreach \t in {2,6,8}{\bCircle{6}{\t}{colU}}
        \foreach \t in {2,6,8}{\sCircle{1}{\t}{colU}}
        \foreach \t in {5,7}{\bCircle{1}{\t}{colU}}
        \node[inner sep=0] at (3.5,3.75) {\scalebox{1.25}{$\mathbf{\vdots}$}};
    \end{tikzpicture}.
\end{eqnarray} 
The diagram immediately implies that the Schmidt rank does not grow with time
and can be bounded from above by $4$. The equivalence between this object and
the corresponding time-state~\eqref{eq:timeStateFinal} can be straightforwardly
shown by explicitly extracting $4\times 4$ matrices and boundary vectors
from~\eqref{eq:defTSs2}. The details of the proof are reported in
Ref.~\cite{klobas2020space}, together with the generalisation of the circuit
representation of the multi-time correlations to the full class of Gibbs-like
stationary states $\vec{p}$.

Furthermore, analogous treatment can be straightforwardly recast for the
quantum version of the
model~\cite{klobas2021exact,klobas2021exactII,klobas2021entanglement}. The main
conceptual difference between the quantum and classical case is the size of the
vector space, which is in the quantum case doubled (squared in dimensionality),
as we are considering density matrices rather than probability distributions.
Despite that, a similar set of few-site algebraic relations can be formulated,
which gives access to objects analogous to~\eqref{eq:defTSs2}.

\subsection{Correlation functions at one site}
The MPA representation of~$\vec{q}$ gives us information about
all multi-point correlation functions. In particular, one can
obtain the full probability distribution of gaps between the consecutive
observed particles, which in related stochastic models shows nontrivial
phenomenology~\cite{garrahan2018aspects,jung2005dynamical,hedges2007decoupling}.
In our case, however, this behaviour is not reproduced, since the solitons passing
through the origin are uncorrelated (see~\cite{klobas2020matrix} for details).

Another quantity of interest is the one-site dynamic density-density correlation
function
\begin{eqnarray}
    C(0,t)=\expval{\rho(t)\rho}_{\vec{p}}-\expval{\rho(t)}_{\vec{p}} \expval{\rho}_{\vec{p}},
\end{eqnarray}
where $\rho$ denotes the local density of full sites at position $x=0$. The correlation
function can be easily expressed in terms of the MPS~\eqref{eq:timeStateFinal} as,
\begin{eqnarray}\fl
    C(0,t)=\mel{L}{\A_1\Big((\B_0+\B_1)(\A_0+\A_1)\Big)^{t/2-1} (\B_0+\B_1)\A_1}{R}
    -\mel{L}{A_1}{R}^2,
\end{eqnarray}
where we for simplicity assume even $t$. Since the matrices $\AB_s$ are finite-dimensional,
the above matrix product can be easily evaluated and shown to be equal to
\begin{eqnarray}\fl
    C(0,t)=\frac{(\Lambda_1^{\frac{t}{2}}-\Lambda_2^{\frac{t}{2}})
    ((p_l+p_r)^2-p_l p_r(1+p_l+p_r))+
    (\Lambda_1^{\frac{t-2}{2}}-\Lambda_2^{\frac{t-2}{2}})p_l p_r(p_l+p_r)}
    {(1+p_l+p_r)^2(\Lambda_1-\Lambda_2)},
\end{eqnarray}
where $\Lambda_{1,2}$ are subleading eigenvalues of the transfer matrix $(\B_0+\B_1)(\A_0+\A_1)$,
\begin{eqnarray}
    \Lambda_{1,2}=-\frac{1}{2}\left(
    p_l+p_r-p_l p_r\pm\sqrt{(p_l+p_r-p_l p_r)^2-4 p_l p_r}
    \right).
\end{eqnarray}
The correlation function decays exponentially for all values of parameters, and
furthermore, one can show~\cite{klobas2021exact,klobas2021exactII} that the
correlations exponentially decay also when we move away from $x=0$ as long as
$|x/t|<\frac{1}{3}$. This is compatible with the hydrodynamic picture,
according to which the correlations move with velocity larger than
$\frac{1}{3}$ (cf.~\eqref{eq:admVels}), and therefore we cannot observe power-law
decay on rays with $|x/t|<\frac{1}{3}$.

\section{Time-dependent matrix product ansatz}
\label{sect:timeMPA}
In the previous section we were discussing evolution of an observable
restricted to one spatial coordinate. Here we go one step further and discuss
the \emph{full} time-evolution of local one-site observables, which is the
classical analogue of the Heisenberg picture time-evolution. In generic
systems, this is a hard problem and the computational complexity
typically increases exponentially with time. However, in the case of Rule 54
the problem reduces significantly and the complexity of time-evolution grows
quadratically with time. This is the consequence of the fact that the
deterministic time-evolution reduces to the problem of back-tracking positions
of solitons, which can be due to the simplicity of the interactions done efficiently.

Most of the section is devoted to the construction of the time-dependent
matrix-product state that encodes this procedure, which allows us to find exact
density profile after an inhomogeneous quench and the density-density
spatio-temporal correlation function in the maximum entropy state. We conclude
by discussing the quantum generalisation to non-diagonal
observables, which provides a strict bound on operator-spreading in the quantum
version of the model.

\subsection{Dynamics of classical observables}
We assume the finite chain with periodic boundaries of \emph{even}
length $n$, and the sites are labelled by integers in the interval $[-n/2+1,n/2]$,
so that the site $0$ is at the middle of the chain. The length of half of the
chain $n/2$ is assumed to be strictly larger than any finite time $t$ we are
considering.
Here we use $\U(t)$ to denote time-evolution for the first $t$ time-steps, i.e.\
\begin{eqnarray}\label{eq:tEvoltMPA}
    \U(t)=\begin{cases}
        (\Uo\Ue)^{\frac{t}{2}},& t\equiv 0\pmod{2},\\
        \Ue(\Uo\Ue)^{\frac{t-1}{2}},& t\equiv 1\pmod{2},
    \end{cases}
\end{eqnarray}
and in the first time-step, \emph{even} sites get updated.

As we already noted in the previous section, classical observables can
be understood as diagonal operators acting on the $2^{n}$-dimensional
space of (macroscopic) states, so that the classical analogue of the
Heisenberg picture time-evolution is given by
%(see \ref{app:determEvol} for the details)
\begin{eqnarray}
    a(t) = \U(t)^{-1} a \U(t),
\end{eqnarray}
which in our case reduces to
\begin{eqnarray}\label{eq:obstEvoltMPA}
    a(t+1) =
    \begin{cases}
        \Ue a(t) \Ue,& t\equiv 0 \pmod{2},\\
        \Uo a(t) \Uo,& t\equiv 1 \pmod{2}.
    \end{cases}
\end{eqnarray}
In the following we will consider the simplest local observable: local density
in the centre of the chain $\rho=\rho_0$, where $\rho_x$ denotes the
density at the position $x$,
\begin{eqnarray}
    \rho_x=\one^{\otimes n/2+x-1} \otimes
    \begin{bmatrix}
        0 & 0 \\ 0 & 1 
    \end{bmatrix} \otimes \one^{n/2-x} \equiv \ketbra{1}_x.
\end{eqnarray}
Here we introduced the shorthand symbol $\ketbra{s}_x$, $s=0,1$, for the projector
to one of the local physical states at site $x$. Additionally, we assume the following
notation for the projector to the configuration $(s_{-m},s_{-m+1},\ldots, s_{m})$
on the section $[-m,m]$ of the lattice,
\begin{eqnarray}\fl
    \ketbra{s_{-m}\ldots s_m} = 
    \ketbra{s_{-m}}_{-m} \cdot
    \ketbra{s_{-m+1}}_{-m+1}\cdots
    \ketbra{s_{m}}_{m}.
\end{eqnarray}
Note that $\rho_x(t)$ can be obtained from $\rho(t)$ by a lattice spatial shift operator $\eta$,
\begin{eqnarray}\label{eq:tMPAShiftStaggering1}\fl
    \rho_{x}(t) =
    \begin{cases}
        \eta^{x} \rho(t) \eta^{-x}, & x\equiv 0\pmod{2},\\
        \eta^{x} \rho(t-1) \eta^{-x}, & x\equiv 1\pmod{2},
    \end{cases}\qquad
    \eta^{x} \ketbra{s}_y \eta^{-x} = \ketbra{s}_{y+x}.
\end{eqnarray}
Furthermore, $\rho_x(t)$ also gives immediate access to the time-evolved local
density of empty sites $\bar{\rho}_x=\ketbra{0}_x$, through the relation
$\bar{\rho}_x(t)=\one-\rho_x(t)$. Therefore, assuming that we know
$\rho(t)$, we are in principle able to express time evolution of any local diagonal
operator, since time evolution of observables is a homomorphism, $(A
B)(t)=A(t)B(t)$.

Time evolution is deterministic and local (see e.g.\ \eqref{eq:detUrels}), which
implies that at time $t$ the time-evolved observable acts nontrivially on the section of
the chain between $-t$ and $t$,
\begin{eqnarray}
    \rho(t) = \sum_{\ul{s}\in\mathbb{Z}_2^{2t+1}}
    c_{\ul{s}}(t)
    \ketbra{\ul{s}}
\end{eqnarray}
where $c_{\ul{s}}(t)$ are the appropriate coefficients in the basis expansion.
By the definition of time evolution~\eqref{eq:defU} each one of these
observables will be mapped into observables with a (two sites) larger support,
\begin{eqnarray}\fl
    \eqalign{
        &\Ueo  \ketbra{s_{-t} s_{-t+1}\ldots s_{t}} \Ueo\\
    &=
    \quad \smashoperator{\sum_{s_{-t-1},s_{t+1}\in\{0,1\}}}
    \ketbra{s_{-t-1} s_{-t}^{\prime} s_{-t+1} \ldots s_{t}^{\prime}s_{t+1}},
\quad
    s_j^{\prime}=\chi(s_{j-1},s_j,s_{j+1}),}
\end{eqnarray}
where the choice between $\Ue$ and $\Uo$ depends on the parity of the time step
(cf.\ \eqref{eq:obstEvoltMPA}). This immediately implies that there are exactly
$4^t$ different configurations $\ul{s}$, referred to as \emph{accessible} configurations,
of length $2t+1$, for which $c_{\ul{s}}(t)=1$, while the other $4^t$ coefficients are
zero. Furthermore, accessible configurations are precisely the ones for which the site
at the origin was full at time $t=0$, which means that the configuration $\ul{s}$
contains a soliton that was at time $t=0$ at the site $x=0$. The problem of time-evolution
of local density can be therefore mapped onto the equivalent problem of identifying
solitons in the configuration and backtracking their position in time.

\begin{figure}
  \centering
  \includegraphics[width=0.7\textwidth]{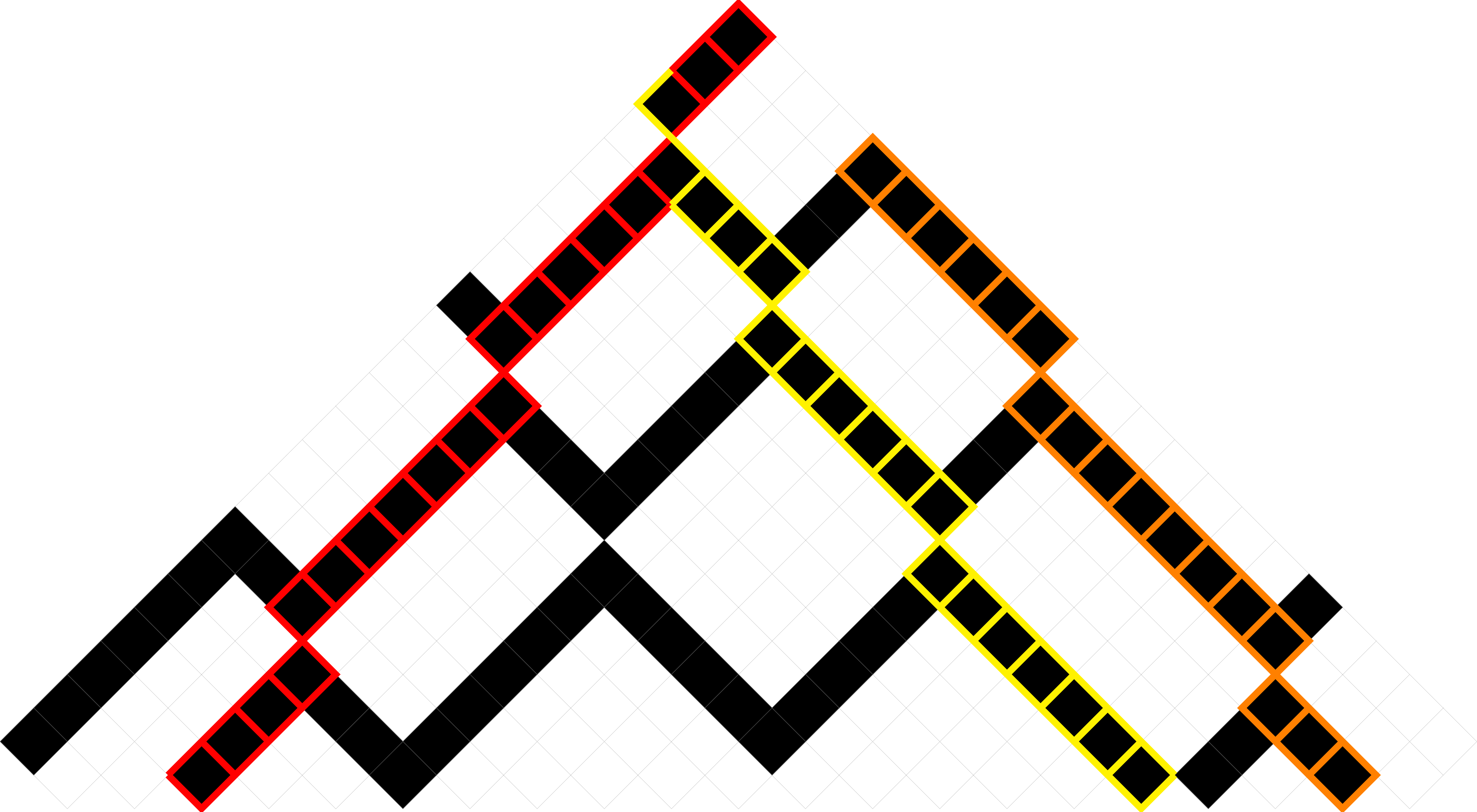}
    \caption{\label{fig:tMPAfig2} An example of an \emph{accessible} configuration
    (bottom-most saw), with the full history. The red-bordered sites denote the left-moving
    soliton that was at $t=0$ in the origin. Its position in the configuration is determined
    by the number of right-movers it encountered: it scattered $3$ times and each time
    it got displaced for one \emph{ray} (i.e.\ two sites) towards the middle of the light-cone.
    Solitons denoted by yellow and orange edges are examples of right-movers. The yellow
    soliton scattered with the red soliton, while the orange soliton is too far
    from the red one to have scattered with it.
    }
\end{figure}

By itself, this realization is not very deep and in principle the complexity of
time-evolution could still grow exponentially with time $t$.  However, in our
case, the dynamics of solitons is very simple, which significantly reduces the
complexity of the problem. In particular, for any soliton found in the
configuration $\ul{s}=(s_{-t},s_{-t+1},\ldots,s_{t})$ it is possible to
determine whether or not at time $t=0$ it was at the origin by counting the
number of scatterings it was involved in. An example of an accessible
configuration is shown in Fig.~\ref{fig:tMPAfig2}. It contains several
solitons, amongst which there is a left-mover (denoted by the red colour), that
was at time $t=0$ at the position $x=0$. To show that the red soliton passed
through the origin, one has to identify three right-moving solitons it
scattered with, which are all contained in the part of the configuration to the
right of it. However, one should be careful not to count too many right movers,
as e.g.\ the orange soliton never scattered with the red one, which is achieved
by keeping track of the \emph{left}-movers positioned to the \emph{right} of
the red soliton.  

To make this simple idea more precise, we start by splitting the coefficient $c_{\ul{s}}(t)$
into the sum of two terms $c_{\ul{s}}^{\mathrm{L}}(t)$ and
$c_{\ul{s}}^{\mathrm{R}}(t)$, which give $1$ when the soliton from the origin is moving in
the left and the right direction respectively. As we will argue later, both contributions can be
efficiently expressed in terms of products of matrices
\begin{eqnarray}\label{eq:tMPAdef}
    \fl
    c_{\ul{s}}(t) =
    \underbrace{\!\!
    \mel{\lL(t)}{\lV_{s_{-t}}\lW_{s_{-t+1}}\lV_{s_{-t+2}}\cdots \lV_{s_{t}}}{\lR}\!}
    _{\begin{matrix}c_{\ul{s}}^{\mathrm{L}}(t)\end{matrix}}
    +
    \underbrace{\!\!
    \mel{\rL}{\rV_{s_{-t}}\rW_{s_{-t+1}}\rV_{s_{-t+2}}\cdots \rV_{s_{t}}}{\rR(t)}\!}
    _{\begin{matrix} c_{\ul{s}}^{\mathrm{R}}(t)\end{matrix}}\, ,
\end{eqnarray}
where $\lrV_s,\lrW_s\in\End(\mathcal{V})$, $s\in\{0,1\}$ are linear operators over
the infinite dimensional auxiliary Hilbert space $\mathcal{V}$,
\begin{eqnarray}
    \mathcal{V}=\lspan\big\{
        \ket{c,w,n,a} \mid c,w\in\mathbb{N}_0,\ n\in\{0,1,2\},\ a\in\{0,1\}
    \big\},
\end{eqnarray}
and $\ket{\lL(t)}$, $\ket{\lR}$, $\ket{\rR(t)}$, $\ket{\rR}$ boundary vectors 
from the auxiliary space. Since the boundary vectors change with time $t$, we refer
to the expression~\eqref{eq:tMPAdef} as the \emph{time-dependent matrix-product ansatz}
(tMPA). To compactly express the matrices, we define the ladder operators $\vec{c}^{+/-}$,
$\vec{w}^{+/-}$ that change the value of $c$ and $w$ by one,
\begin{eqnarray}\label{eq:defLadder}
    \eqalign{
        \vec{c}^{+}=\sum_{c,w,n,a}\ketbra{c+1,w,n,a}{c,w,n,a},\qquad
        &\vec{c}^{-}=\vec{c}^{+\, T},\\
        \vec{w}^{+}=\sum_{c,w,n,a}\ketbra{c,w+1,n,a}{c,w,n,a},\qquad
        &\vec{w}^{-}=\vec{w}^{+\, T},
    }
\end{eqnarray}
and the projectors
\begin{eqnarray}\label{eq:defProj}\fl
    \eqalign{
        \vec{c}_{c_1 c_2}=\sum_{w,n,a}\ketbra{c_1,w,n,a}{c_2,w,n,a},\qquad
        &\vec{w}_{w_1 w_2}=\sum_{c,n,a}\ketbra{c,w_1,n,a}{c,w_2,n,a},\\
        \vec{n}_{n_1 n_2}=\sum_{c,w,a}\ketbra{c,w,n_1,a}{c,w,n_2,a},\qquad
        &\vec{a}_{a_1 a_2}=\sum_{c,w,n}\ketbra{c,w,n,a_1}{c,w,n,a_2}.
    }
\end{eqnarray}
We will elaborate on the physical interpretation of the auxiliary space in the
next subsection. The operators $\lV_s$, $\lW_s$ can be expressed as $3\times
3$ matrices acting on the space spanned by $\{\ket{n}\}_{n=0}^2$ as
\begin{eqnarray}\label{eq:tMPAleftmats}
    \eqalign{
    \lV_0\!=\!\!
    \scalemath{0.9}{\begin{bmatrix}
        1 & 0 & 0 \\ 
        \vec{a}_{00}+\vec{a}_{01} \vec{c}^+ 
        + \vec{a}_{11}\vec{c}^{+}\vec{w}^{+}
        & 0 & 0 \\ 1 & 0 & 0
    \end{bmatrix}
    +\vec{a}_{11}\vec{w}_{00}
    \begin{bmatrix}
        0&0&0\\
        1&0&0\\
        0&0&0
    \end{bmatrix}}
    ,\qquad
    \\
    \lV_1\!=\!\!
    \scalemath{0.9}{\begin{bmatrix}
        0 & 1 & 0 \\
        0 & 0 & \vec{a}_{00} + \vec{a}_{01} + \vec{a}_{11} \vec{w}^{+}\\
        0 & 0 & \vec{a}_{00} + \vec{a}_{01} + \vec{a}_{11} \vec{w}^{+}
    \end{bmatrix}+
    \vec{a}_{11}\vec{w}_{00}
    \begin{bmatrix}
        0 & 0 & 0 \\
        0 & 0 & 1 \\
        0 & 0 & 1
    \end{bmatrix}},\\
    \lW_0\!=\!\!
    \scalemath{0.9}{
    \begin{bmatrix}
        \vec{a}_{00} \vec{c}^{-}\vec{w}^{+} + \vec{a}_{11}\vec{w}^{+} & 0 & 0\\
        \vec{a}_{00} \vec{c}^{-}\vec{w}^{+} + \vec{a}_{01} \vec{w}^{+}
        +\vec{a}_{11} \vec{c}^{+}\left.\vec{w}^{+}\right.^2 & 0 & 0\\
        \vec{a}_{00} \vec{c}^{-}\vec{w}^{+} + \vec{a}_{11}\vec{w}^{+} & 0 & 0
    \end{bmatrix}+
    \vec{a}_{11}
    \begin{bmatrix}
        \vec{w}_{00} & 0 & 0\\
        \vec{c}^{+}\vec{w}_{10}+\vec{w}_{00}& 0 & 0\\
        \vec{w}_{00} & 0 & 0
    \end{bmatrix}},\\
    \lW_1\!=\!\!
    \scalemath{0.9}{
        \begin{bmatrix}
            0&\vec{a}_{00}\vec{c}^{-}\vec{w}^{+} + \vec{a}_{11}\vec{w}^{+} & 0 \\
            0&0&\vec{a}_{00}\vec{c}^{-}\vec{w}^{+}+\vec{a}_{11}\vec{c}^{+}\vec{w}^{+}\\
            0&0&\vec{a}_{00}\vec{c}^{-}\vec{w}^{+}+\vec{a}_{11}\vec{c}^{+}\vec{w}^{+}
        \end{bmatrix}+
        \vec{a}_{11}\vec{w}_{00}
        \begin{bmatrix}
            0&1 & 0 \\
            0&0&1\\
            0&0&1
        \end{bmatrix}},
    }
\end{eqnarray}
while the time-dependent boundary vectors are given by
\begin{eqnarray}\label{eq:tMPAleftvecs}
    \bra{\lL(t)}=\bra{0,t,0,0},\qquad \ket{\lR}=\sum_{n=0}^2\ket{0,0,n,1}.
\end{eqnarray}
The contribution $c_{\ul{s}}^{\mathrm{R}}(t)$ can be obtained from $c_{\ul{s}}^{\mathrm{L}}(t)$
by taking a transpose, and adding additional boundary terms (see the discussion at the end
of this section for details),
\begin{eqnarray}\label{eq:tMPArightmats}\fl
    \eqalign{
        \left.\rV_0\right.^T=
            \lV_0-\vec{c}_{10}
            \big(\vec{a}_{01}\vec{w}_{11}
            +\vec{a}_{11}(\vec{w}_{10}+\vec{w}_{21}) \big)
        \scalemath{0.9}{
            \begin{bmatrix}
                0 & 0 & 0 \\
                1 & 0 & 0 \\
                0 & 0 & 0
            \end{bmatrix}
        },\quad
        \left.\rV_1\right.^T=\lV_1,\\
        \left.\rW_0\right.^T=\lW_0-\vec{a}_{11}\vec{c}_{10}(\vec{w}_{10}+\vec{w}_{20})
        \scalemath{0.9}{
            \begin{bmatrix}
                0&0&0\\
                1&0&0\\
                0&0&0
            \end{bmatrix}
        },\quad
        \bra{\rL}=\bra{\lR}\!+\!\bra{0,1,0,1}\!+\!\bra{0,1,2,1},\\
        \left.\rW_1\right.^T=\lW_1-\vec{a}_{11}\vec{c}_{10}(\vec{w}_{10}+\vec{w}_{21})
        \scalemath{0.9}{
            \begin{bmatrix}
                0&0&0\\
                0&0&1\\
                0&0&1
            \end{bmatrix}
        },\quad
        \ket{\rR(t)}=\ket{\lL(t+1)}.
    }
\end{eqnarray}
The structure of these $3\times 3$ matrices is very similar to the
representation of the cubic algebra introduced in Section~\ref{sec:NESS_mpa}.
Indeed, operators $\lV_s$, $\lW_s$, $\left.\rV_s\right.^T$,
$\left.\rW_s\right.^T$ can be understood as generalisation of $W^{(\prime)}_s$
(cf.~\eqref{eq:matWWp}), where the entries are no longer scalars, but rather
more general operators acting on an infinite-dimensional vector space.

Before discussing the construction of tMPAs for $c_{\ul{s}}^{\mathrm{L}}(t)$
and $c_{\ul{s}}^{\mathrm{R}}(t)$, let us remark that even though the operators
$\lrV_s$, $\lrW_s$ are acting on an infinite-dimensional Hilbert space, we can
for every finite $t$ replace them by finite matrices, with the dimension that
scales as $\mathcal{O}(t^2)$. In particular, the possible values of $c$ and $w$
in definitions of ladder operators~\eqref{eq:defLadder} and
projectors~\eqref{eq:defProj} can be restricted to $0\le c,w\le t+1$. Note
that this is an exceptional property, since typically one expects the
complexity to grow exponentially with time.

\subsubsection{The contribution of left-movers}
We start the construction of the matrix-product
representation~\eqref{eq:tMPAdef} by discussing the first term
$c_{\ul{s}}^{\mathrm{L}}(t)$. Afterwards we will show how to adapt the result
to obtain also the contribution from the right-movers.

Let us consider a configuration
$(s_{-t},s_{-t+1},\ldots,s_{t})=\ul{s}\in\mathbb{Z}_2^{2t+1}$.  To figure out
whether or not $c_{\ul{s}}^{\mathrm{L}}(t)=1$, we imagine starting at the
left-most site of the configuration and moving towards the right, reading the
configuration $\ul{s}$ site by site. Whenever we encounter a left-mover, we
designate it the \emph{test} soliton and, aiming to determine whether or not it
originated from the centre, we start counting the solitons on its right. To
encode the soliton counting procedure, we introduce four auxiliary degrees of
freedom, $\ket{c,w,n,a}$.
\begin{enumerate}[label=(\roman*)]
    \item The \emph{activation bit}, $a\in\{0,1\}$, tells us whether we are on the
        left ($a=0$) or the right ($a=1$) side of the test soliton. If the activation
        bit is turned off, the state splits into two parts whenever we encounter a left mover.
        The first part corresponds to $a=0$, describing the situation in which the left mover
        is not the test soliton, while the second part represents the opposite case, with
        $a=1$. In any other situation the activation bit remains unchanged.
    \item The \emph{collision counter}, $c\ge0$, represents the number of scatterings
        the test soliton had to undergo if it passed through the origin. As long as
        $a=0$, the collision counter increases by $1$ every two sites. If $a=1$, the collision
        counter decreases by $1$ whenever a right-mover that scattered with the test soliton
        is encountered. When we reach the right edge of the configuration, the value $c=0$
        tells us that the test soliton passed through the origin, and $c\neq 0$
        implies the opposite.
    \item The \emph{scattering width}, $w\ge 0$, keeps track of the number of scatterings
        of the right-movers encountered after the test soliton. At the left edge the width
        $w$ is equal to the time-step $t$, and every two sites it decreases by $1$. Additionally,
        the width changes as $w\to w-1$ whenever a \emph{left} mover on the right side of
        the test soliton is encountered. All the right-moving solitons that we meet after
        $w$ drops to $0$ could not have scattered with the test soliton.
    \item The \emph{occupation counter}, $n\in\{0,1,2\}$, provides additional information
        about the particle content needed to appropriately change $w$ and $c$. Explicitly,
        $n=0$ if the current site is empty, $n=1$ if the site is full and the left neighbour
        is empty, and $n=2$ if the site and the left neighbour are both occupied.
\end{enumerate}
At the left edge, the collision counter should be $0$, and the scattering width is set
to be equal to the time-step $t$. Furthermore, we start with the activation bit equal to $0$
(since we have not yet met any test soliton), and the particle content outside of the light-cone
does not influence the soliton counting, therefore we can without loss of generality assume 
$n=0$. This implies the following form of the left boundary vector, 
\begin{eqnarray}
    \bra{\lL(t)}=\bra{0,t,0,0}.
\end{eqnarray}
The right boundary vector should have nonzero overlap with vectors that correspond to a test
soliton that passed through the origin, which is given by $c=w=0$ and $a=1$, while $n$ can
be arbitrary,
\begin{eqnarray}
    \ket{\lR}=\sum_{n=0}^2\ket{0,0,n,1}.
\end{eqnarray}
To construct the matrices that correspond to the procedure summarized above, we consider
$3$ regimes.

\paragraph{Left side of the test soliton.}
Before the test soliton is encountered, the scattering width decreases by $1$ every
two sites and at the same time the collision counter should increase, while $n$ should
be appropriately adjusted to keep track of the consecutive full sites. The left action
of the matrices on the sector $a=0$ is therefore,
\begin{eqnarray}\label{eq:tMPAconst1}
    \eqalign{
    \bra{c,w,n,0}\lV_s \vec{a}_{00}=\bra{c,w, s\cdot \min\{n+1,2\},0},\\
    \bra{c,w,n,0}\lW_s \vec{a}_{00}=\bra{c+1,w-1, s\cdot \min\{n+1,2\},0}.
}
\end{eqnarray}

\paragraph{Encountering the test soliton.} Whenever a left-mover is encountered while
$a=0$, an additional vector with $a=1$ is created. There are $4$ configurations corresponding
to this situation. The first two are simpler and represent a situation where a left-mover
is detected while moving uninterrupted,
\begin{eqnarray}
    \label{eq:simpleLeftConf}
    \begin{tikzpicture}
        [scale=0.3,baseline={([yshift=-0.6ex]current bounding box.center)}]
        \rectangle{-1}{1}{0};
        \rectangle{0}{0}{1};
        \rectangle{1}{1}{1};
        \begin{scope}[shift={(6,0)}]
            \rectangle{-1}{1}{1};
            \rectangle{0}{0}{1};
            \rectangle{1}{1}{1};
        \end{scope}
    \end{tikzpicture},
\end{eqnarray}
which can be summarized with the following matrix elements
\begin{eqnarray}\label{eq:tMPAconst2}
    \bra{c,w,1,0}\lV_1 \vec{a}_{11} =
    \bra{c,w,2,0}\lV_1 \vec{a}_{11} = \bra{c,w,2,1}.
\end{eqnarray}
The other two configurations describe the soliton encountered while it is scattering,
\begin{eqnarray}
    \label{eq:scatteringSolsConf}
    \begin{tikzpicture}
        [scale=0.3,baseline={([yshift=-0.6ex]current bounding box.center)}]
        \halfdashedrectangle{0}{2};
        \halfdashedrectangle{1}{3};
        \rectangle{-1}{1}{0};
        \rectangle{0}{0}{1};
        \rectangle{1}{1}{0};
        \begin{scope}[shift={(6,0)}]
            \halfdashedrectangle{1}{2};
            \rectangle{-1}{0}{0};
            \rectangle{0}{1}{1};
            \rectangle{1}{0}{0};
        \end{scope}
    \end{tikzpicture},
\end{eqnarray}
where the grey squares denote the path of the soliton in the previous time-steps.
The corresponding matrix elements are
\begin{eqnarray}\label{eq:tMPAconst3}
\fl 
    \bra{c,w,1,0}\lV_0 \vec{a}_{11} = \bra{c-1,w,0,1},\qquad
    \bra{c,w,1,0}\lW_0 \vec{a}_{11} = \bra{c,w-1,0,1},
\end{eqnarray}
where in the first case we note that one scattering already occurred (hence $c-1$),
while in the second case the decrease of $w$ is just due to moving to the right.
This exhausts all the subconfigurations in which the bit $a$ can flip (from the left)
from $0$ to $1$.

\paragraph{Right side of the test soliton.} Let us first assume $w>0$.
After the test soliton has been chosen, the changes of $c$ and $w$ 
are summarized as:
\begin{enumerate}[label=(\roman*)]
    \item $c\to c-1$ if there is a right-mover,
    \item $w\to w-1$ if there is a left-mover
    \item $w\to w-1$ every two sites.
\end{enumerate}
Explicitly, there are two possible configurations of an uninterrupted
right-mover appearing,
\begin{eqnarray}
    \label{eq:simpleRightConf}
    \begin{tikzpicture}
        [scale=0.3,baseline={([yshift=-0.6ex]current bounding box.center)}]
        \rectangle{-1}{0}{0};
        \rectangle{0}{1}{1};
        \rectangle{1}{0}{1};
        \begin{scope}[shift={(6,0)}]
            \rectangle{-1}{0}{1};
            \rectangle{0}{1}{1};
            \rectangle{1}{0}{1};
        \end{scope}
    \end{tikzpicture},
\end{eqnarray}
which are described by,
\begin{eqnarray}\label{eq:tMPAconst4}
    \bra{c,w,1,1}\lW_1=\bra{c,w,2,1}\lW_1=\bra{c-1,w-1,2,1},
\end{eqnarray}
while the configurations in~\eqref{eq:simpleLeftConf} imply the following,
\begin{eqnarray}\label{eq:tMPAconst5}
    \bra{c,w,1,1}\lV_{1}=\bra{c,w,2,1}\lV_{1}=\bra{c,w-1,2,1}.
\end{eqnarray}
The configurations in~\eqref{eq:scatteringSolsConf} describe the two scattering
solitons, therefore the corresponding matrix elements contain the decrease of 
both $c$ and $w$,
\begin{eqnarray}\label{eq:tMPAconst6}
    \eqalign{
    \bra{c,w,1,1}\lV_0=\bra{c-1,w-1,0,1},\\
    \bra{c,w,1,1}\lW_0=
    (1-\delta_{w,0})\bra{c-1,w-2,0,1}+
    \delta_{w,0}\bra{c-1,0,0,1}.}
\end{eqnarray}
All the remaining configurations do not contain any newly detected solitons,
\begin{eqnarray}
    \begin{tikzpicture}
        [scale=0.3,baseline={([yshift=-0.6ex]current bounding box.center)}]
        \begin{scope}[shift={(16,0)}]
            \rectangle{0}{1}{0};
            \rectangle{1}{0}{0};
        \end{scope}
        \begin{scope}[shift={(21,0)}]
            \rectangle{0}{1}{0};
            \rectangle{1}{0}{1};
        \end{scope}
        \begin{scope}[shift={(27,0)}]
            \rectangle{-1}{0}{1};
            \rectangle{0}{1}{1};
            \rectangle{1}{0}{0};
        \end{scope}
        \begin{scope}[shift={(0,0)}]
            \rectangle{0}{0}{0};
            \rectangle{1}{1}{0};
        \end{scope}
        \begin{scope}[shift={(5,0)}]
            \rectangle{0}{0}{0};
            \rectangle{1}{1}{1};
        \end{scope}
        \begin{scope}[shift={(11,0)}]
            \rectangle{-1}{1}{1};
            \rectangle{0}{0}{1};
            \rectangle{1}{1}{0};
        \end{scope}
    \end{tikzpicture},
\end{eqnarray}
therefore $c$ remains constant and $w$ decreases only due to moving two sites to
the right,
\begin{eqnarray}\label{eq:tMPAconst7}
    \eqalign{
    \bra{c,w,0,1}\lV_s = \bra{c,w,s,1},\qquad
    & \bra{c,w,2,1}\lV_0 = \bra{c,w,0,1},\\
    \bra{c,w,0,1}\lW_s = \bra{c,w-1,s,1},\qquad
    & \bra{c,w,2,1}\lW_0 = \bra{c,w-1,0,1}.
}
\end{eqnarray}
The right-moving solitons encountered after the width drops to $0$ could not scatter
with the test soliton, because they were too far. Therefore $c$ should stop decreasing,
\begin{eqnarray}\label{eq:tMPAconst8}
    \bra{c,0,n,1}\lV_s=
    \bra{c,0,n,1}\lW_s=\bra{c,0,s\cdot\max\{2,n+1\},1}.
\end{eqnarray}

Combining the matrix elements from
Equations~(\ref{eq:tMPAconst1}--\ref{eq:tMPAconst8}), we finally recover the
matrices given in~\eqref{eq:tMPAleftmats}. This completes the construction of
the tMPA corresponding to $c_{\ul{s}}^{\mathrm{L}}(t)$.

\subsubsection{The contribution of right-movers}
The tMPA for $c_{\ul{s}}^{\mathrm{R}}(t)$ can be derived in a similar fashion,
by simply reversing the directions of all solitons, which is achieved by
exchanging the roles of the boundary vectors and by transposing the tMPA
matrices. Additionally, we have to exclude all the configurations with a
central right-mover that were already captured by $c_{\ul{s}}^{\mathrm{L}}(t)$.
These are exactly the configurations  where at time $t=0$ the soliton in the
origin is in the process of scattering, i.e.\ their trajectory starts with one
of the following two diagrams,
\begin{eqnarray}
    \begin{tikzpicture}[scale=0.3,baseline={([yshift=-0.6ex]current bounding box.center)}]
        \begin{scope}[shift={(0,0)}]
            \rectangle{-1}{0}{0};
            \rectangle{0}{1}{1};
            \rectangle{1}{0}{0};
            \rectangle{0}{-1}{1};
        \end{scope}
        \begin{scope}[shift={(6,0)}]
            \rectangle{-1}{0}{1};
            \rectangle{0}{1}{1};
            \rectangle{1}{0}{1};
            \rectangle{0}{-1}{0};
        \end{scope}
    \end{tikzpicture}.
\end{eqnarray}
In the previous case, these would correspond to the configurations with another
scattering occurring exactly when $w=0$. Therefore to exclude them, we have to
increase $w$ in the right boundary vector by $1$, while at the same time modify the
left boundary vector to allow for $w$ not dropping to $0$,
\begin{eqnarray}
    \ket{\rR(t)}=\ket{\lL(t+1)},\qquad
    \bra{\rL}=\bra{\lR}+\bra{0,1,0,1}+\bra{0,1,2,1}.
\end{eqnarray}
Additionally, it does not suffice to map $\left.\lV_s\right.^T\to\rV_s$,
$\left.\lW_s\right.^T\to\rW_s$, but we have to further modify them so that the following
matrix elements are $0$,
\begin{eqnarray}
    \eqalign{
    \vec{a}_{11}\rV_0\ket{1,1,1,0}=\rV_0\ket{1,1,1,1}=\rV_0\ket{1,2,1,1}=0,\\
    \rW_0\ket{1,1,1,1}=\rW_0\ket{1,2,1,1}=0,\\
    \rW_1\ket{1,1,1,1}=\rW_1\ket{1,2,1,1}=\rW_1\ket{1,1,2,1}=\rW_1\ket{1,2,2,1}=0,
}
\end{eqnarray}
which gives exactly the form summarized in~\eqref{eq:tMPArightmats}.

\subsection{Physical applications}
The tMPA provides means to study transport in the model. In this subsection, we provide 
two examples of physically relevant quantities that can be exactly obtained with our solution:
density profile after a bipartite quench and a dynamical density-density correlation
function. We finish the subsection by suggesting a formal extension of the tMPA to richer
stationary states.

\subsubsection{Inhomogeneous quench}
Let us consider a particular bipartitioning protocol, where at time $t=0$, a half-infinite
chain prepared in the maximum-entropy state is joined together with a half-infinite empty
chain, so that the initial state is given by,
\begin{eqnarray}
    \vec{p}=\frac{1}{2^{n/2}}
    \begin{bmatrix}
        1 \\ 1 
    \end{bmatrix}^{\otimes n/2}\!
    \otimes
    \begin{bmatrix}
        1 \\ 0
    \end{bmatrix}^{\otimes n/2}.
\end{eqnarray}
We wish to obtain the time-dependent profile of the particle density
$\tilde{\rho}(x,t)$, which is given by the expectation value of $\rho_x$ in
the time-evolved inhomogeneous state,
\begin{eqnarray}\label{eq:tMPAShiftStaggering2}
\tilde{\rho}(x,t)=\ev{\rho_x}_{\vec{p}(t)}=\ev{\rho_x(-t)}_{\vec{p}}
    =\begin{cases}
        \ev{\rho(t)}_{\eta^{-x}\vec{p}},& x+t\equiv 1,\\
        \ev{\rho(t-1)}_{\eta^{-x}\vec{p}},& x+t\equiv 0.
    \end{cases}
\end{eqnarray}
Note that the negative sign in $\rho_x(-t)$ together with the staggering forces
a different convention with respect to~\eqref{eq:tMPAShiftStaggering1}.

The locality of time-evolution immediately implies that outside of the light-cone the
density profile matches the appropriate boundary values,
i.e.\ $\tilde{\rho}(x<-t,t)=\frac{1}{2}$ and $\tilde{\rho}(x>t,t)=0$.
For intermediate values of $x$, the expectation value is by definition
(see Eq.~\eqref{eq:defMultiCorrsOmega} and the discussion around it)
\begin{eqnarray}\fl
    \eqalign{
        r(\tfrac{t-x+1}{2},t)\!= \!\ev{\rho(t)}_{\eta^{-x}\vec{p}} \!\!&=  \!
        \sum_{\ul{s}} c_{\ul{s}}(t)\,
        2^{-(t-x+1)}
        \Bigg[
            \vec{\omega}^{T\,\otimes 2t+1} \cdot \ketbra{\ul{s}} \cdot
            \begin{bmatrix}
                1 \\ 1
            \end{bmatrix}^{\otimes t-x+1}
            \mkern-18mu
            \otimes
            \begin{bmatrix}
                1 \\ 0
            \end{bmatrix}^{\otimes t+x}\Bigg]\\
        &=\smashoperator{\sum_{s_{-t}s_{-t+1}\ldots s_{-x}}}
        2^{-(t-x+1)} c_{s_{-t}s_{-t+1}\ldots s_{-x} 0 \ldots 0}(t),
    }
\end{eqnarray}
where $\vec{\omega}=[1 \  1]^T$
denotes a (unnormalised) one-site maximum-entropy state.
The expectation value is (assuming odd $x+t$, cf.~\eqref{eq:tMPAShiftStaggering2})
conveniently expressed in terms of the tMPA as
\begin{eqnarray}\fl
    \eqalign{
        r(m,t)=2^{-2m} &\Big(
        \mel{\lL(t)}{\big((\lV_0+\lV_1)(\lW_0+\lW_1)\big)^m\lV_0\big(\lW_0\lV_0\big)^{t-m}}{\lR}\\
        &+
        \mel{\rL}{\big((\rV_0+\rV_1)(\rW_0+\rW_1)\big)^m \rV_0 \big(\rW_0\rV_0\big)^{t-m}}{\rR(t)}
        \Big).
    }
\end{eqnarray}

The matrices $\lrV_s$, $\lrW_s$ are infinitely dimensional, therefore it is not
immediately clear whether it is possible to evaluate this expression. However, their
sparse structure enables us to find exact form of certain simple matrix
elements (see~\cite{klobas2019timedependent} for the details).  In particular,
one can prove that $r(m,t)$ takes the following form,
\begin{eqnarray}\fl\label{eq:inhQprofile}
    \eqalign{
        r(m,t) &= 
        \frac{3}{8}\delta_{m,t} + \frac{1}{8}(\delta_{m,1}\delta_{t,1}+\delta_{m,2}\delta_{t,2})
        +\frac{1}{4}\delta_{m,3}\delta_{t,4}
        +\theta_{2m-t-3}\,2^{t-2m-1}\binom{2m-t-3}{t-m-1}\\
        &+\frac{1}{8}\sum_{y=t-m-2}^{2m-t-3}
        2^{-(2m-t-3)}\binom{2m-t-3}{y}
        +\frac{3}{16}\sum_{y=0}^{2m-t-3}2^{-y}\binom{y}{t-m-1}\\
        &+\frac{1}{2}\sum_{y=0}^{t-m-1}2^{-(m-1-y)}\binom{m-1-y}{y},
        }
\end{eqnarray}
where $\theta_x$ denotes the discrete Heaviside function (i.e.\ $\theta_{x\ge0}=1$ and
$\theta_{x<0}=0$).

\begin{figure}
    \centering
    \includegraphics[width=\textwidth]{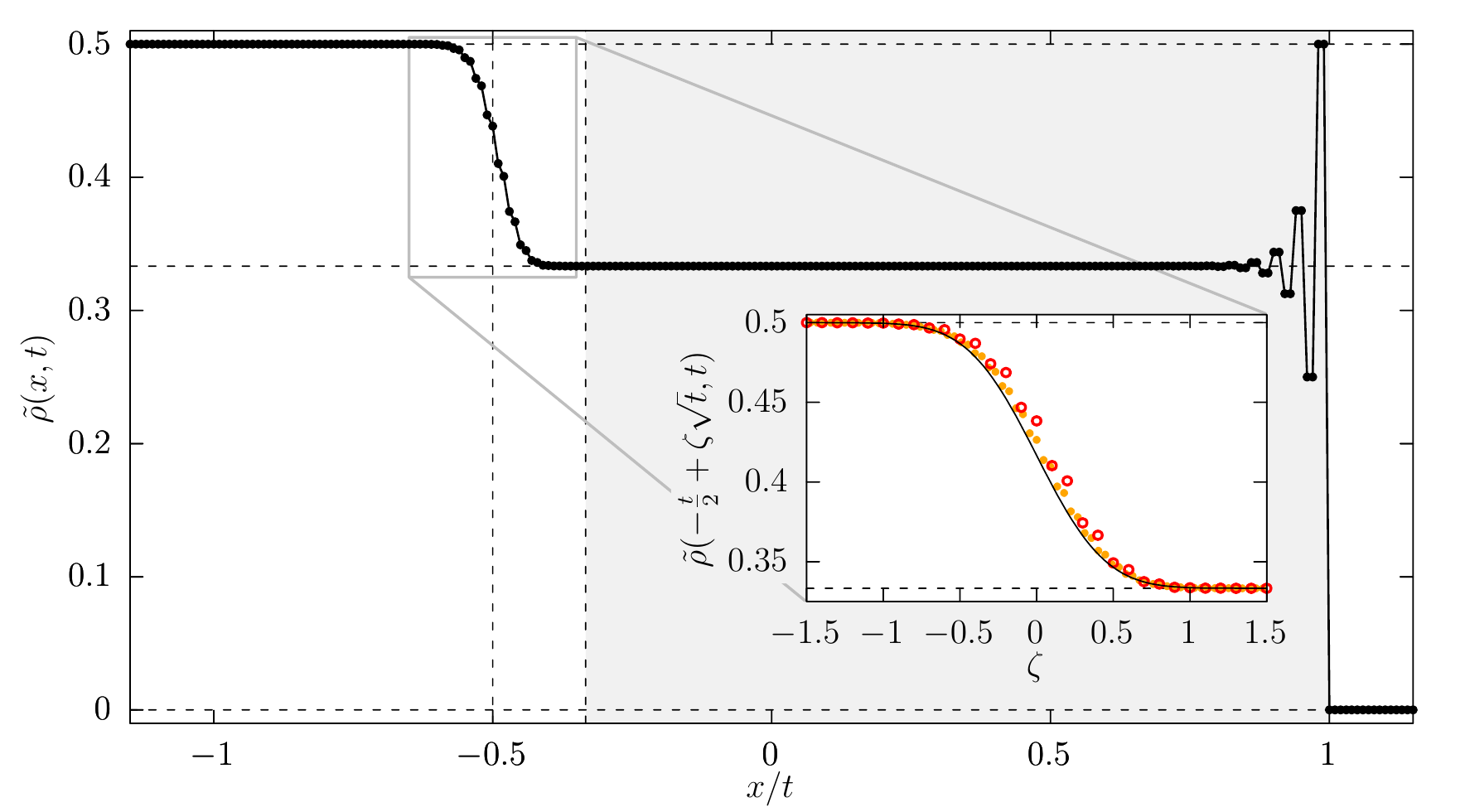}
    \caption{\label{fig:tMPAfig4} Profile after the inhomogeneous quench for $t=100$. 
    The grey background highlights the section of the lightcone $x/t\ge -1/3$,
    in which the profile consists only of exponentially suppressed oscillations around
    $\frac{1}{3}$ (cf.~\eqref{eq:tMPAinhQdensity1}). The inset shows the diffusive scaling of
    the profile centred around $x/t=-1/2$, with red circles and orange dots
    corresponding to $t=100$ and $t=500$ respectively. The solid line represents the asymptotic
    profile~\eqref{eq:tMPAinhQdensity2}.
    }
\end{figure}

If $m<(2t+1)/3$, only the last term remains and $r(m,t)$ greatly simplifies
to give the following density profile for $x>-(t+1)/3$,
\begin{eqnarray}\label{eq:tMPAinhQdensity1}
    \tilde{\rho}(x>-\tfrac{t+2}{3},t)=
    \frac{1}{3}
    \left(1-\left(-\frac{1}{2}\right)^{\lfloor\tfrac{t+x+1}{2}\rfloor}\right).
\end{eqnarray}
This profile consists of a front moving ballistically to the right with
velocity $1$ and exponentially suppressed oscillations behind it, as can also
be seen in Fig.~\ref{fig:tMPAfig4}. Additionally, it exhibits an intuitive
physical picture: in this section of the light-cone, all the solitons are
moving to the \emph{right} with the dressed velocity $1$, since the right half
of the lattice was empty at the beginning and there are no left-movers to slow
them down. To determine the left-edge of this section, we note that the
\emph{right-most} left-moving solitons are moving with the velocity $-1/3$
(cf.\ Fig.~\ref{fig:hydrofig1} and the discussion in Section~\ref{sect:hydro}),
therefore the full section of the lattice between $x=-t/3$ and $x=t$ is
populated only by right-movers~\cite{klobas2019timedependent}.  

As we increase $m$ (or equivalently, as we move towards the left edge of the
lightcone), we are no longer able to make further simplifications to the
expression~\eqref{eq:inhQprofile}. However, one can easily find an asymptotic
expression that describes $\tilde{\rho}(x,t)$ on large scale.  Using the
Stirling formula it is possible to show that the density profile is well
approximated by a shifted error function centred around $x/t=-1/2$, with the
width that scales as $\sqrt{t}$,
\begin{eqnarray}\label{eq:tMPAinhQdensity2}
    \lim_{t\to\infty}\tilde{\rho}\left(-\tfrac{t}{2}+\zeta\sqrt{t},t\right)=
    \frac{1}{12}(5-\erf(2\zeta)).
\end{eqnarray}

On the rays with the constant $\zeta=x/t$, the density profile matches the Euler scale
prediction~\eqref{eq:inhQpred}: the profile consists of two steps, with their
velocities equal to $-\tfrac{1}{2}$ and $1$. In the central part, the hydrodynamic
prediction is given by the expectation value of density in the state determined by
$(\vartheta_{+},\vartheta_{-})=(\frac{1}{2},0)$, and can be shown to be equal to
$\frac{1}{3}$ using the MPA representation of GGEs (see~\ref{app:MPSGibbs}). Furthermore,
the diffusive broadening around the left step is compatible with $D_{--}=\frac{1}{16}$,
which is consistent with the expression~\eqref{eq:defD} in the case $\vartheta_{+}=\frac{1}{2}$.
The right-step does not exhibit any diffusive scaling, since the term $D_{++}$ vanishes
for $\vartheta_{-}=0$. The oscillations around $\frac{1}{3}$ decay exponentially with
the distance from the step and are therefore not detectable on the hydrodynamic scale.

\subsubsection{Dynamical structure factor}
We proceed to the second example: the dynamic density-density correlation function,
$C(x,t)$, evaluated in the maximum-entropy state,
\begin{eqnarray}
    C(x,t)=
    \underbrace{\expval{\rho_0(t)\rho_x}_{\vec{p}_{\infty}}}_{
        \displaystyle 2^{2t+1} \tilde{C}(x,t)}
    -\underbrace{\expval{\rho}^2_{\vec{p}_{\infty}}}_{\displaystyle 2^{-2}},\qquad
    \vec{p}_{\infty}=2^{-n}
    \begin{bmatrix}
        1 \\ 1
    \end{bmatrix}^{\otimes n},
\end{eqnarray}
where we introduced $\tilde{C}(x,t)$ to denote the rescaled expectation value of 
the product of densities. Analogously to $r(m,t)$ defined before, $\tilde{C}(x,t)$
can be expressed as a simple sum of matrix products corresponding to the left and
right tMPAs,
\begin{eqnarray}\fl
    \eqalign{
        \tilde{C}(x,t)&=
        \mel{\lL(t)}{\left((\lV_0+\lV_1)(\lW_0+\lW_1)\right)^{\frac{t+x}{2}}
        \lV_1 \left((\lW_0+\lW_1)(\lV_0+\lV_1)\right)^{\frac{t-x}{2}}}{\lR}\\
        &+
\mel{\rL}{\left((\rV_0+\rV_1)(\rW_0+\rW_1)\right)^{\frac{t+x}{2}}
        \rV_1 \left((\rW_0+\rW_1)(\rV_0+\rV_1)\right)^{\frac{t-x}{2}}}{\rR(t)},
}
\end{eqnarray}
where we assumed $x+t\equiv 0\pmod{2}$. In the opposite case, we take the
advantage of the identity $\left.C(x,t)\right|_{x+t\equiv 1 \pmod{2}}=C(x,t-1)$,
and thus reduce it to the previous case.  It is again possible to evaluate
these sums, and we obtain the following expression~\cite{klobas2019timedependent},
\begin{eqnarray}
    C(x,t)=2^{-t-1}\sum_{m=0}^{\frac{t-\abs{x}-2}{2}}
    4^m \left(
    2\binom{t-2m-3}{m}-\binom{t-2m-2}{m}
    \right),
\end{eqnarray}
where for simplicity of notation we assume $\binom{n<0}{k}=0$ for any $k,n\in\mathbb{Z}$.

\begin{figure}
    \centering
    \includegraphics[width=\textwidth]{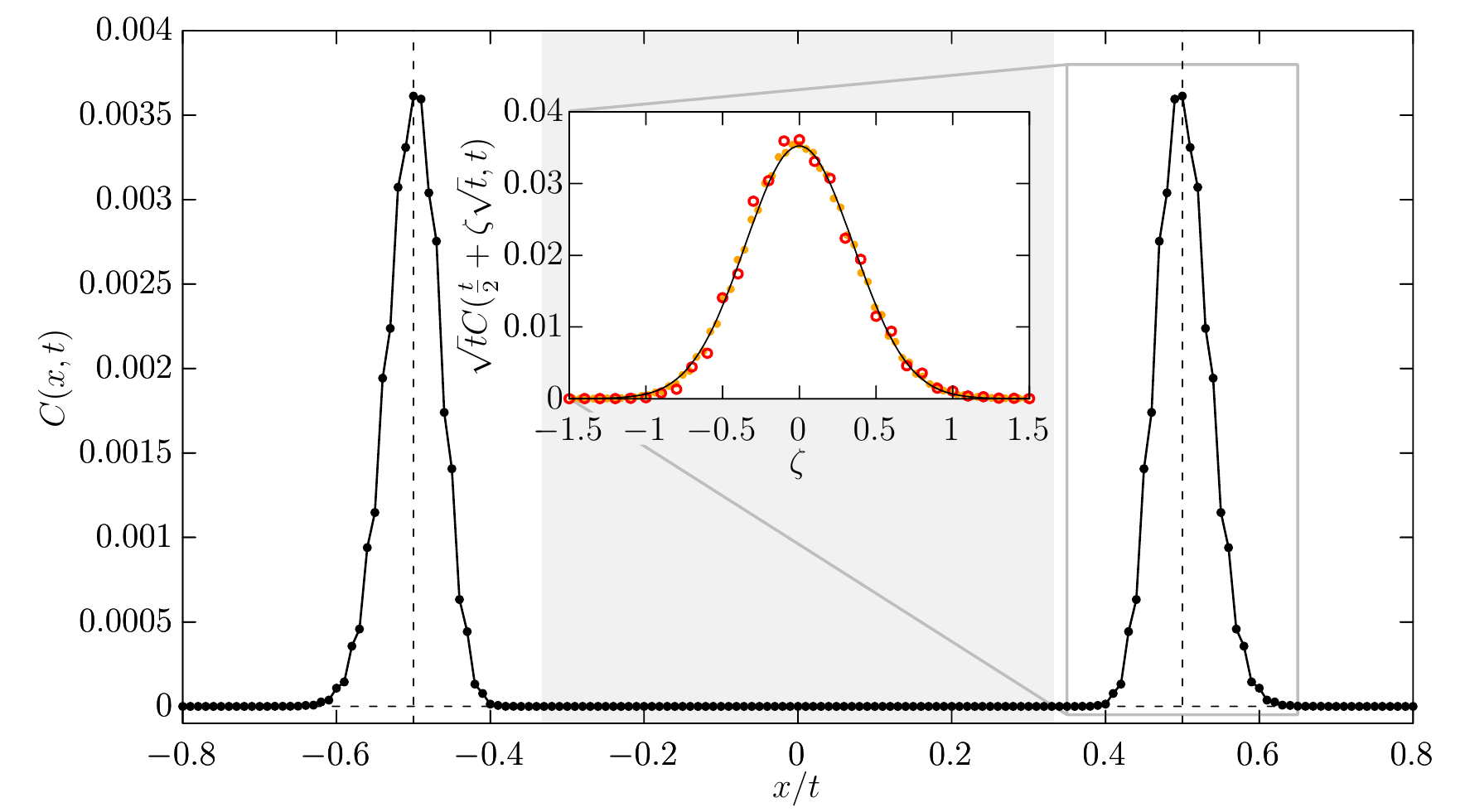}
    \caption{\label{fig:tMPAfig5} Dynamical correlation function $C(x,t)$ for $t=100$.
    The grey background highlights the section of the lightcone $\abs{x}/t \le 1/3$,
    in which the correlation function decays exponentially and loses the spatial dependence.
    The inset shows the diffusive scaling of the right peak, with red circles and orange dots
    corresponding to $t=100$ and $t=500$ respectively. The solid line represents the asymptotic
    profile~\eqref{eq:CxtAsymptoticProfile}.
    }
\end{figure}

The correlation profile can be again split into two qualitatively different parts.
If $\abs{x}<(t+1)/3$, the spatial dependence disappears (apart from the staggering) and
the correlations reduce to
\begin{eqnarray}\fl
    \left.C(x,t)\right|_{\abs{x}<\frac{t+1}{3}}=
    2^{-2t-3}\left(
    \left(1+\frac{\mathrm{i}}{\sqrt{7}}\right)\left(-1-\mathrm{i}\sqrt{7}\right)^t
    +
    \left(1-\frac{\mathrm{i}}{\sqrt{7}}\right)\left(-1+\mathrm{i}\sqrt{7}\right)^t
    \right).
\end{eqnarray}
In this regime, correlations decay exponentially with time, $\abs{C(0,t)}\sim
2^{-t/2}$, and the above expression exactly matches the infinite temperature
result obtained from multi-time correlation functions in Section~\ref{sec:TS}.
Furthermore, in this regime $C(x,t)$ is homogeneous in space and exponentially
suppressed also when one considers the underlying stationary states to belong
to the richer class of two-species GGEs considered in the previous sections,
as a consequence of the \emph{minimal} velocity of excitations being equal to $1/3$
(cf.\ \eqref{eq:admVels}). This can be proven by generalising the circuit
approach outlined in Section~\ref{sec:TS}~\cite{klobas2021exact,klobas2021exactII}.

If $\abs{x}/t>1/3$, we are left with a binomial sum that cannot be further simplified,
therefore it again makes sense to find the asymptotic shape of the correlation profile.
As is shown in Fig.~\ref{fig:tMPAfig5}, the profile consists of two ballistically
moving peaks  with velocities $\pm\frac{1}{2}$ that spread diffusively. Indeed, one can
straightforwardly show
\begin{eqnarray}\label{eq:CxtAsymptoticProfile}
    \lim_{t\to\infty}
    C\left( \pm\left(\tfrac{t}{2}+\zeta\sqrt{t}\right),t \right)
    =\frac{1}{16\sqrt{t \pi}}\mathrm{e}^{-4\zeta^2}.
\end{eqnarray}
The result nicely matches the hydrodynamics: the peaks move with dressed
velocities $v_{\pm}=\pm\frac{1}{2}$ and their diffusive broadening is governed
by $D_{\nu\nu}=\frac{1}{16}$, $\nu\in\{+,-\}$. Both quantities are compatible
with their predictions given by~\eqref{eq:hydroDressedVelocities}
and~\eqref{eq:defD} in the maximum-entropy state.

\subsubsection{Generalisation to generalised Gibbs states}
So far we used the tMPA to express objects either evaluated in the vacuum state with with only empty configuration
or in the simple maximum entropy state, where all configurations are equally likely. However,
the fact that the matrices $\lrV_s$, $\lrW_s$ have in the space spanned by
$\{\ket{n}\}_{n=0,1,2}$ similar block structure to $\WV_s$, allows us to simply expand the
tMPA to include the right probabilities by a simple element-wise multiplication of both
sets of matrices, and thus generalise the approach to the class of GGEs given by~\eqref{GGE}.
To make this point more precise, let us consider the dynamical correlation
function in a state $\vec{p}$ given by the pair of parameters $(\xi,\omega)$
(corresponding to $\mu_{+}=-\log\xi$, $\mu_{-}=-\log\omega$ in~\eqref{GGE}),
which is by definition expressed as
\begin{eqnarray}%\fl
    \expval{\rho_0(t)\rho_x}_{\vec{p}}=
    \smashoperator{\sum_{s_{-t},\ldots,s_t}}
    \delta_{s_x,1}
    \underbrace{c_{\ul{s}}(t)
    p\Big(\!\begin{tikzpicture}
        [scale=0.3,baseline={([yshift=-0.6ex]current bounding box.center)}]
        \textrectangle{0}{0.5}{$s_{-t}$};
        \textrectangle{1}{-0.5}{\scalebox{0.85}{$s_{\!-\!t\!+\!1}$}};
        \textrectangle{2}{0.5}{$\cdots$};
        \textrectangle{3}{-0.5}{};
        \textrectangle{4}{0.5}{};
        \textrectangle{5}{-0.5}{$\cdots$};
        \textrectangle{6}{0.5}{$s_{t}$};
    \end{tikzpicture}\!\Big)}_{c^{\ast}_{\ul{s}}(t)}
    =
    \smashoperator{\sum_{s_{-t},\ldots,s_t}}
    \delta_{s_x,1}
    c^{\ast}_{\ul{s}}(t),
\end{eqnarray}
where the state-dependent coefficient $c^{\ast}_{\ul{s}}(t)$ is defined as the product
of the tMPA coefficient $c_{\ul{s}}(t)$ and the probability of finite
configuration~\eqref{eq:asymptPropOdd}. Using the MPA formulation of these probabilities
outlined in~\ref{appsubsec:thermodynamicEVs}, the state-dependent coefficient can
be rewritten as
\begin{eqnarray}
    c^{\ast}_{s_{-t}s_{-t+1}\ldots s_t}(t)=
    c_{s_{-t}s_{-t+1}\ldots s_t}(t)
        \frac{\mel{l^{\prime}}{\V_{s_{-t}}\W_{s_{-t+1}}\cdots \V_{s_{t}}}{r}}
        {\lambda^{t}\mel{l^{\prime}}{\V_0+\V_1}{r}}.
\end{eqnarray}
Since both the matrices describing stationary probabilities and the matrices 
constituting the tMPA have very similar $3\times 3$ block structure
(cf.~\eqref{eq:tMPAleftmats}, \eqref{eq:tMPArightmats},
and~\eqref{eq:matWWp}), we can in analogy to~\eqref{eq:tMPAdef}
express $c^{\ast}_{\ul{s}}(t)$ as
a sum of the left and right time-dependent matrix products, by defining the new
operators $\left.\lrV_s\right.^{\ast}$, $\left.\lrW_s\right.^{\ast}$ as
\begin{eqnarray}\label{eq:extMPA1}
    \eqalign{
    \left.\lV_s\right.^{\ast}=\lV_s \odot \V_s,\qquad
    &\left.\lW_s\right.^{\ast}=\frac{1}{\lambda}\lW_s \odot \W_s,\\
    \left.\rV_s\right.^{\ast}= \rV_s \odot \left.\W_s\right.^{\!\!T}\!\!,
    &\left.\rW_s\right.^{\ast}=\frac{1}{\lambda}\rW_s \odot \left.\V_s\right.^T\!\!,},
\end{eqnarray}
and similarly,
\begin{eqnarray}\label{eq:extMPA2}
    \eqalign{
    \bra{\left.\lL\right.^{\ast}(t)}=\bra{\lL(t)} \odot \bra{l^{\prime}},\qquad
    &\bra{\left.\rL\right.^{\ast}}=\frac{\bra{\rL} \odot \bra{r}}
    {\mel{r}{\left.\V_0\right.^T+\left.\V_1\right.^T}{l^{\prime}}},\\
    \ket{\left.\rR\right.^{\ast}(t)}= 
    \ket{\rR(t)} \odot \ket{l^{\prime}},\qquad
    &\ket{\left.\lR\right.^{\ast}}=
    \frac{\ket{\lR}\odot \ket{r}}{\mel{l^{\prime}}{\V_0+\V_1}{r}}.
}
\end{eqnarray}
Here $\odot$ denotes the element-wise multiplication, i.e.\ $\big(
a_1 \ a_2 \ a_3\big)^T \odot \big( b_1 \ b_2 \ b_3\big)^T=
\big(a_1 b_1 \ a_2 b_2 \ a_3 b_3\big)^T$ (defined analogously
for $3\times 3$ matrices).
The generalised expressions~(\ref{eq:extMPA1},\ref{eq:extMPA2}) can be
understood simply as attaching a weight $\xi$ or $\omega$ to
each left or right-mover that we observe in the soliton-counting procedure.
When $\xi=\omega=1$ we recover the original tMPA matrices and
vectors~(\ref{eq:tMPAleftmats},\ref{eq:tMPAleftvecs},\ref{eq:tMPArightmats}).
Note that in the case of right-movers we have to use the transpose of the GGE
matrices, and additionally the role of parameters is swapped. The explicit form
of the new set of operators can be found in~\cite{klobas2020exactPhD}.

This construction in principle allows us to obtain $C(x,t)$ for the full class
of two-parameter GGEs (space and time translation invariant states) and, by appropriately choosing parameters, one can also
obtain the tMPA formulation of the time-dependent density profile after a bipartite quench from an arbitrary pair of
these GGE states. However, so far no attempt has been done to evaluate these expressions
and obtain explicit results. Whether or not this is feasible is still an open question.

\subsection{Quantum interpretations and operator spreading}
Even though we have been discussing RCA54 in the context of classical dynamics,
there is nothing preventing us from treating it in the quantum setting.
Indeed, since the one-time-step evolution operators $\Ue$ and $\Uo$ describe
deterministic dynamics, they are also a special case of unitaries,
$\left(\Uo\Ue\right)\left(\Uo\Ue\right)^{\dagger}=\one$. In this respect the
model represents one of the simplest interacting models, in which different
measures of operator spreading can be easily
studied~\cite{gopalakrishnan2018operator,gopalakrishnan2018hydrodynamics,
alba2019operator,alba2020diffusion}.  In particular, one such quantity is
\emph{operator space entanglement entropy} (OSEE), which measures the
complexity of simulability of quantum dynamics in the Heisenberg
picture~\cite{prosen2007operator,hartmann2009density,pizorn2009operator,muth2011dynamical}.
In the case of RCA54 the bound of the rate of
OSEE growth with time can be obtained exactly, by generalising the tMPA
introduced above to the time-evolution of \emph{quantum} (rather than
classical) observables~\cite{alba2019operator}.

Time-evolution of a quantum one-site observable $a$ can be split into $4$ contributions, each
one corresponding to a basis operator $\ketbra{s_1}{s_2}$, 
\begin{eqnarray}\fl
    a=\begin{bmatrix}
        a_{00} & a_{01}\\
        a_{10} & a_{11}
    \end{bmatrix},\quad
    \eqalign{
        a(t)=\U(t)^{\dagger} a \U(t) &= 
        a_{00} 
        \U(t)^{\dagger}\ketbra{0}{0}\U(t)+
        a_{01}\U(t)^{\dagger}\ketbra{0}{1}\U(t)\\
        &+
        a_{10}\U(t)^{\dagger}\ketbra{1}{0}\U(t)+
        a_{11}\U(t)^{\dagger}\ketbra{1}{1}\U(t).
    }
\end{eqnarray}
We immediately note that the diagonal elements,
\begin{eqnarray}
    \U(t)^{\dagger}\ketbra{1}{1}\U(t)=\rho(t),\qquad
    \U(t)^{\dagger}\ketbra{0}{0}\U(t)=\bar{\rho}(t)=\one-\rho(t),
\end{eqnarray}
are given by the tMPA introduced in the previous subsections, while one of the
off-diagonal terms can be expressed as the transpose of the other,
\begin{eqnarray}
    \bar{\sigma}(t)=\U(t)^{\dagger}\ketbra{1}{0}\U(t)=
    \left(\U(t)^{\dagger}\ketbra{0}{1}\U(t)\right)^{\dagger}
    =\sigma(t)^{\dagger}.
\end{eqnarray}
Therefore, the only remaining thing to do is to find an efficient matrix-product description
of $\sigma(t)=\U(t)^{\dagger}\ketbra{0}{1}\U(t)$.

\begin{figure}
  \centering
  \includegraphics[width=0.6\textwidth]{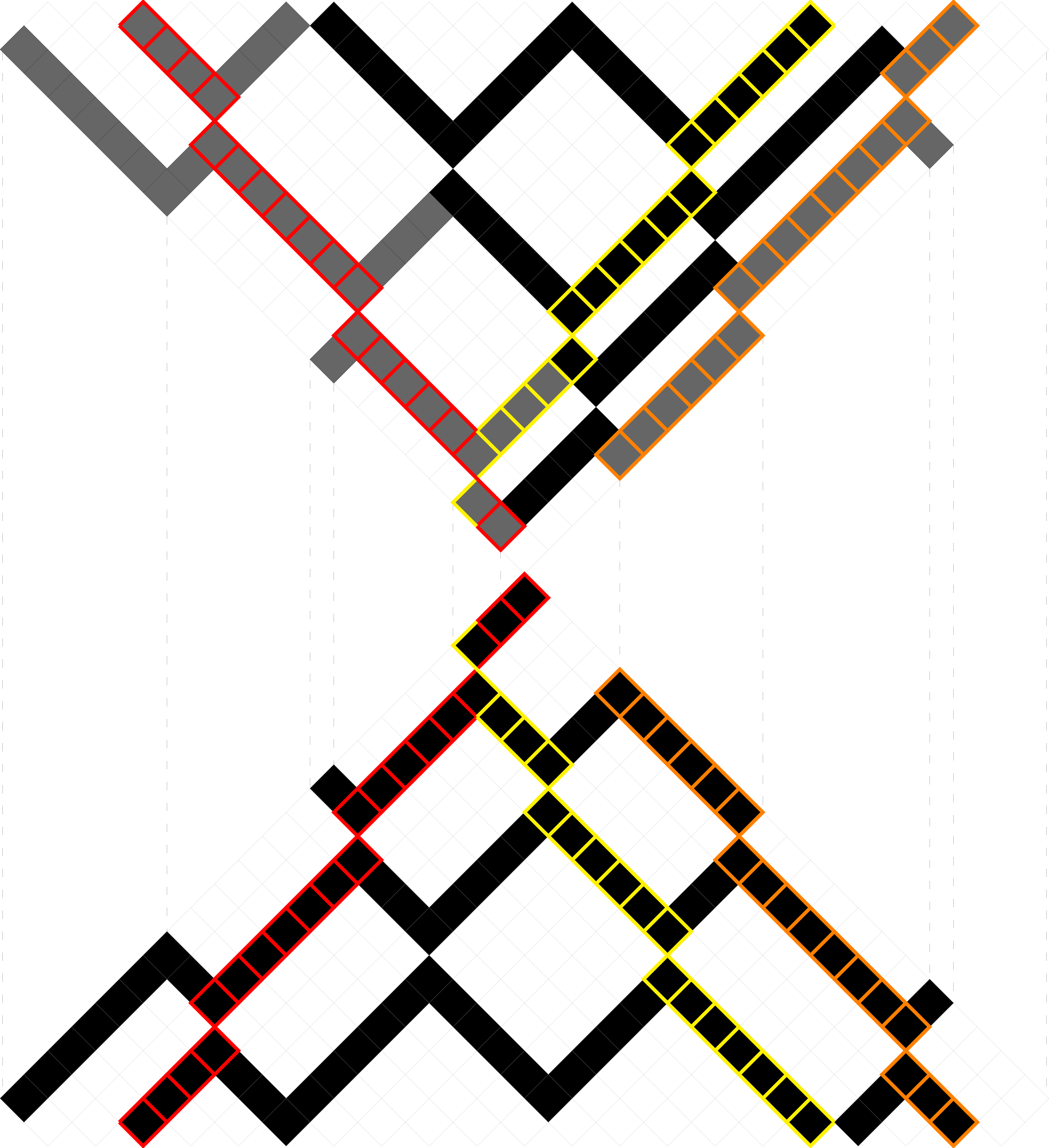}
    \caption{\label{fig:tMPAfig3}
    An example of two trajectories, starting with $1$ and $0$ in the centre and
    resulting in an \emph{accessible} pair of configurations $\ul{s}$
    (bottom-most zig-zag line), $\ul{b}$ (top-most configuration), for which
    $d_{\ul{s},\ul{b}}(t)=1$. For convenience we represent one trajectory on
    top of the other with time visually pointing in opposite directions.  In
    both cases, the configurations of bits at the very edges of the light-cone
    are the same (suggested by dashed vertical lines), except for the initial
    site, which is $1$ for $\ul{s}$ and $0$ for $\ul{b}$. The sites denoted by
    the red, yellow and orange colour in both cases denote the same three
    solitons, while in the top part the grey sites denote solitons, whose
    positions are the same as in the bottom part. To understand how $\ul{s}$
    maps into $\ul{b}$, we note that in this example, changing the initial
    state from $1$ to $0$ causes an additional right-mover to appear (between
    the yellow and orange soliton), which causes the displacement of all the
    solitons that scatter with it.
    }
\end{figure}

Analogously to the diagonal case, at time $t$ the time-evolved operator $\sigma(t)$ nontrivially
acts on the sublattice of length $2t+1$ and we can introduce coefficients $d_{\ul{s},\ul{b}}(t)$,
such that
\begin{eqnarray}
    \sigma(t)=\smashoperator{\sum_{\ul{s},\ul{b}\in\mathbb{Z}_2^{2t+1}}}
    d_{\ul{s},\ul{b}}(t) \ketbra{\ul{s}}{\ul{b}}.
\end{eqnarray}
By definition,  each one of these terms will be at time $t+1$ mapped
into $4$ different operators with support $2(t+1)+1$,
\begin{eqnarray}
    \eqalign{
        &\Ueo \ketbra{s_{-t}s_{-t+1}\ldots s_{t}}{b_{-t}b_{-t+1}\ldots b_{t}}\Ueo\\
    &=
    \smashoperator{\sum_{\substack{s_{-t-1},s_{t+1}\\b_{-t-1},b_{t+1}}}}
    \ketbra{
        s^{\phantom{\prime}}_{-t-1} 
        s^{\prime}_{-t} \cdots
        s^{\prime}_{t}
        s^{\phantom{\prime}}_{t+1}}{
        b^{\phantom{\prime}}_{-t-1} 
        b^{\prime}_{-t} \cdots
        b^{\prime}_{t}
        b^{\phantom{\prime}}_{t+1}}
    \delta_{b_{-(t+1)},s_{-(t+1)}}\delta_{b_{t+1},s_{t+1}},
}
\end{eqnarray}
where we use the shorthand notation
$s_{x}^{\prime}=\chi(s_{x-1},s_{x},s_{x+1})$ and
$b_{x}^{\prime}=\chi(b_{x-1},b_x,b_{x+1})$.
The deterministic nature of time evolution has direct implications that
can be summarised in the next three points.
\begin{enumerate}[label=(\roman*)]
    \item The coefficient $d_{\ul{s},\ul{b}}(t)$ can be either $0$ or $1$.
    \item If $c_{\ul{s}}(t)=0$, then also $d_{\ul{s},\ul{b}}(t)=0$ for all $\ul{b}$.
    \item \label{num:lightcone}
        If $c_{\ul{s}}(t)=1$, there exists a \emph{unique} configuration 
        $\ul{b}$, for which $d_{\ul{s},\ul{b}}(t)=1$. In this case, if the configurations
        $\ul{s}$ and $\ul{b}$ are evolved backwards in time, we will end up in
        a central $1$ and $0$ respectively, while in \emph{all} intermediate time-steps,
        the states at the edges of the two light-cones coincide.
\end{enumerate}
%We immediately notice that the
%coefficient $d_{\ul{s},\ul{b}}(t)$ can be either $0$ or $1$, and
%$c_{\ul{s}}(t)=0$ implies $d_{\ul{s},\ul{b}}(t)=0$ for all $\ul{b}$. If
%$c_{\ul{s}}(t)=1$, there exists a \emph{unique} configuration $\ul{b}$ such
%that $d_{\ul{s},\ul{b}}(t)=1$.
%In particular, if $d_{\ul{s},\ul{b}}(t)=1$,
%evolving $\ul{s}$ and $\ul{b}$ backwards in time for $t$ time-steps
%and keeping only the information inside the light-cone, will result in $1$ and $0$ 
%respectively, and additionally, the values of bits on the very edges of the light-cone
%are the same in both instances.
An example illustrating the point~\ref{num:lightcone} is shown in
Fig.~\ref{fig:tMPAfig3}. In this case, changing the initial bit from $1$ to $0$
produces an additional right-mover, therefore $\ul{s}$ can be mapped into
$\ul{b}$ by adding the additional soliton and displacing all the
quasi-particles that were affected by it.

More generally, it can be shown that changing $1$ to $0$ in the initial
configuration can be always understood as adding or removing $1$ or $2$
solitons, and in all these cases, there exists a simple deterministic map $M_t$
from $\ul{s}$ to $\ul{b}$ that accounts for the relevant displacements of
quasi-particles. Furthermore, as was demonstrated in Ref.~\cite{alba2019operator}, the
map $M_t$ can be expressed in terms of products of \emph{finite-dimensional}
matrices, with the size that does not change with time $t$.  This immediately implies
that the growth of the complexity of time-evolution is given by the size of
operators constituting $c_{\ul{s}}(t)$, which is $\mathcal{O}(t^2)$.  The
growth of OSEE is therefore bounded from above by $2\log(t)$ (up to a constant
factor). Note that this bound need not be saturated, and indeed, there are
recent indications of OSEE increasing as
$\frac{1}{2}\log{t}$~\cite{alba2020diffusion}.

\section{Related exactly solvable models}\label{sec:relModels}
While, as argued in this review, the RCA54 model allows for a remarkable set of
explicit results on nonequilibrium dynamics and statistical mechanics, an
important question is if similar can be achieved for other, perhaps closely
related models.

Similar results have been so far demonstrated in two other (RCA type) deterministic lattice models. Specifically, in Refs.~\cite{wilkinson2020exact,iadecola2020nonergodic} the rule 201 RCA has been studied, which can be understood simply as a {\em negation} of rule 54, namely with the local rule modified as
\be
\chi_{201}(s,s',s'') = 1-\chi_{54}(s,s',s'').
\ee
Both the RCA201 and RCA54 can be interpreted as deterministic variants of \emph{kinetically constrained} (or \emph{facilitated}) models, typically studied in the context of glassy dynamics~\cite{fredrickson1984kinetic,palmer1984models,jackle1991hierarchically,ritort2003glassy,garrahan2011dynamical}: the middle bit (cell $s'$) flips whenever the neighbouring bits (cells $s,s^{\prime\prime}$) satisfy a particular condition. In the case of RCA54 the flip occurs when at least one of $s$ and $s^{\prime\prime}$ is in the state $1$, while in RCA201 the flip happens when both $s$ and $s^{\prime\prime}$ are in the state $0$, and as such it provides a deterministic classical analogy of the PXP model~\cite{PXP,lesanovsky2011manybody,turner2018weak,fendley2004competing}.
It has been shown \cite{wilkinson2020exact} that the invariant (equilibrum) states and NESS of the boundary driven setup can again be written in terms of a patch state and matrix product ansatz with very similar constituent matrix algebra as for the rule 54. However, at the moment it is still unclear precisely which other features of RCA54 dynamics remain exactly solvable for the RCA201, which is left as an open question. Similar explicit results are expected to be achievable for the \emph{free} RCA rule 150, $\chi_{150}(s,s',s'') = s + s' + s'' \pmod{2}$ (or equivalently, the flip occurs when exactly one of $s$, $s^{\prime\prime}$ is $0$ and one $1$), which represents a caricature of non-interacting dynamics in $1+1$ dimensions.

Another class of related deterministic discrete nonequilibrium dynamics that has been studied in some detail is the two-species cellular automation describing a synchronous lattice gas of charged point particles with hard-core interaction~\cite{medenjak2017diffusion,klobas2018exactly,medenjak2019two}. This model fundamentally differs from the RCAs of Ref.~\cite{bobenko1993two} (such as rules 54, 150, 201) in the way the global time-evolution is defined. In the first time-step, the \emph{even} pairs of neighbouring cells are updated using a two-site update rule,
\begin{eqnarray}
    (s^{t+1}_x,s^{t+1}_{x+1}) = \phi(s^t_x,s^t_{x+1}),\qquad s^{t}_x\in\{+,-,0\},
\end{eqnarray}
while in the second time-step, the update rule is applied to \emph{odd} pairs of sites, and the two steps are then periodically repeated. This is analogous to trotterisations of nearest-neighbour interacting models, or to brickwork circuits. The two-species hardcore interacting model specified by a self-invertible bijective mapping,
\begin{eqnarray}
    \phi(\pm,\mp) = (\pm,\mp), \text{ and } \phi(s,s^{\prime}) = (s^{\prime},s)
    \text{ otherwise,}
\end{eqnarray}
has been shown to display a coexistence of ballistic and diffusive transport (similar to Rule 54) with explicitly calculable diffusion constant and Drude weight~\cite{medenjak2017diffusion,klobas2018exactly}, 
explicitly solvable boundary driven NESS exhibiting a nonequilibrium phase transition~\cite{medenjak2019two},
but with a slightly lower complexity of time evolution compared to RCA54: the Schmidt rank of time-dependent MPA grows as $\propto t$~\cite{medenjak2019two} rather than $t^2$.

\section{Conclusions and perspectives}

This review paper provides a coherent overview of a series of explicit results, appearing over the last five years -- most of which were co-contributed by the authors, on dynamics and nonequilibrium statistical mechanics of a particular interacting reversible cellular automaton (RCA54). We believe that establishing such a set of explicit results on the connection between microscopic reversible laws of motion and effectively irreversible dynamics of macroscopic states is indispensable for the field whose holy grail is to establish such a connection on a general level \cite{lebowitz1999statistical}.

The phenomenology that the model helps to rigorously understand is the emergence of Fick's law of diffusive transport, as a leading-order correction to the ballistic motion.
%An important physics phenomenology that this model helps to understand is a rigorous emergence of the Fick's law of diffusive matter transport.
One can thus consider the RCA54 as an exactly solvable example of nonequilibrium universality class of diffusive transport. However, it has been recently discovered, that even without any intrinsic sources of noise (or dissipation), one-dimensional systems can display other types of super-diffusive (but sub-ballistic) transport, for example in the presence of integrability and non-abelian global symmetries, spin (charge) transport seem to follow the scaling of KPZ universality class with dynamical exponent $3/2$ (rather than $2$ for the diffusive class) \cite{ljubotina2019,spohn2019,krajnik2020a,krajnik2020b,bulchandani2021superdiffusion}.
It would be an interesting challenge to look for minimal solvable models of such anomalous transport within a class of reversible cellular automata. 
For instance, one can generalise RCA to multi-colour internal states of non-empty cells such that it reduces exactly to RCA54 if the colours are disregarded. The conserved charges of such automata (some of which appear to be integrable) seem to display superdiffusive transport \cite{klobasproseninprep}.

A different potentially interesting direction is to study quantum and/or stochastic deformations of Rule 54. Specifically, one can allow for coherent (quantum superposition) or stochastic mixtures of several local processes which may preserve exact solvability. For instance, replacing the
local (3-site) propagator (\ref{eq:defU}) by a non-permutation matrix which maps $(000) \to u_{00} (000) + u_{01} (010)$,
$(010) \to u_{10} (000) + u_{11} (010)$ (while all other processes remain as defined by (\ref{eq:defU})), where the $2\times 2$
mixing matrix $u_{ij}$ is either unitary or stochastic, one finds strong circumstantial evidence of integrability in all such cases \cite{proseninprep}. Another kind of quantum deformation of RCA54, introducing a non-trivial quasiparticle dispersion relation, was shown to allow for standard coordinate Bethe ansatz treatment \cite{friedman2019integrable}.

Perhaps the most urgent question regarding RCA54 is how to precisely fit it into the family of Yang-Baxter solvable models. At the moment it seems that the model can be interpreted as a (in some sense) degenerate limit of a more complicated integrable system. One example of the mapping in this direction is the quantum deformation~\cite{friedman2019integrable} mentioned above, in the context of which RCA54 is a zero-dispersion limit of a more general Bethe-ansatz solvable model. Another possibility~\cite{vernier2021yang} is to identify it as a singular limit of a vertex model~\cite{schultz1981solvable}, by first mapping it to a random-tiling model~\cite{widom1993bethe,degier1997integrability}. However, at the moment it is still not clear what is the best way to fit everything together. Nonetheless, interpreting RCA54 as a singular limit of a more generic model, puts it into the same family as other special (but not free) points of interacting models, such as the $q\to\infty$ limit of the $q$-boson model studied in~\cite{pozsgay2014quantum,pozsgay2016realtime}, or the model considered in~\cite{zadnik2021foldedI,zadnik2021foldedII,pozsgay2021integrable}, with which it shares some common features. It would be interesting to understand this connection in more detail.

Another interesting perspective for future research is studying non-stationary dynamics in the RCA54. More specifically, this would involve the existence of extensive quantities $A$ (sums of local terms) satisfying a Floquet \emph{dynamical symmetry} \cite{Buca2019nonstationary,Marko1,Marko2,Chinzei} condition $\U A=e^{\ii \omega} A$, where $\omega\neq 0$ is real, which would imply persistent oscillations and absence of relaxation.

\ack
We thank C.~Mej\'{i}a-Monasterio, M.~Medenjak, M.~ Vanicat, J.~P.~Garrahan, J.~W.~P.~Wilkinson, B.~Bertini and L.~Piroli for collaboration on parts of the topics discussed in this review. We thank B.~Bertini, J.~P.~Garrahan, E.~Vernier, and J.~W.~P.~Wilkinson for insightful comments and suggestions.
This work has been supported by the European Research Council under the
Advanced Grant No.\ 694544 -- OMNES, by the Slovenian Research Agency (ARRS)
under the Programme P1-0402, under the European Union's Seventh Framework Programme (FP7/2007-2013)/ERC Grant Agreement no. 319286, Q-MAC, by EPSRC under grant EP/S020527/1, programme grant EP/P009565/1, and by EPSRC National Quantum Technology Hub in Networked Quantum Information Technology (EP/M013243/1).

\section*{References}
\bibliographystyle{iopart-num}
\bibliography{bibliography}

\appendix
\addtocontents{toc}{\protect\setcounter{tocdepth}{1}}
\section{Boundary vectors for the NESS orbital}
\label{app:boundaryvectorsMPS}
Here we give the boundary vectors solving the boundary equations for the NESS orbital for the case of \emph{conditional driving} \eqref{eq:conditional}. The physical-space-components of the right boundary vectors that satisfy (\ref{bound3a},\ref{bound4a}) are given as,
\begin{eqnarray}\label{eq:rightBoundVecs}\fl
    \eqalign{
        \ket{r'_{0}} =
        \begin{bmatrix}
            \gamma  (\delta -1) \Lambda _{\text{R}} \\
            (\delta -1) \delta   \\
            (\delta -1) \delta  \Lambda _{\text{R}}^2
        \end{bmatrix},
        &\ket{r'_{1}} = 
        \begin{bmatrix}
            (\gamma-\gamma \delta +\delta -1) \Lambda _{\text{R}} \\
            -(\delta -1)^2  \\
            -(\delta -1)^2 \Lambda _{\text{R}}^2 
        \end{bmatrix},\\
        \ket{r_{00}} =
        \begin{bmatrix}
            \gamma  (\delta -1) \Lambda _{\text{R}} \\
            \delta  \left(\gamma +\delta -\delta  \Lambda _{\text{R}}-1\right) \\
            (\delta -1) \delta  \Lambda _{\text{R}}^2
        \end{bmatrix},\qquad
        &\ket{r_{01}} =
        \begin{bmatrix}
            -(\delta -1) \Lambda _{\text{R}} \left(\gamma -\Lambda _{\text{R}}\right) \\
            -(\delta -1) \left(\gamma +\delta -\delta  \Lambda _{\text{R}}-1\right) \\
            -(\delta -1)^2 \Lambda _{\text{R}}^2
        \end{bmatrix},\\
        \ket{r_{10}} = 
        \begin{bmatrix}
            \delta  \Lambda _{\text{R}} \left(\gamma -\Lambda _{\text{R}}\right) \\
            \gamma  (\delta -1) \Lambda _{\text{R}} \\
            \gamma  \Lambda _{\text{R}}^2 \left(\gamma -\Lambda _{\text{R}}\right)
        \end{bmatrix},
        &\ket{r_{11}} = 
        \begin{bmatrix}
            -(\gamma -1) (\delta -1) \Lambda _{\text{R}} \\
            -(\gamma -1) (\delta -1) \Lambda _{\text{R}} \\
            -(\gamma -1) \Lambda _{\text{R}}^2 \left(\gamma -\Lambda _{\text{R}}\right)
        \end{bmatrix}.
    }
\end{eqnarray}
Similarly, the left boundary vectors solving (\ref{bound1a},\ref{bound2a}) are,
\begin{eqnarray}\label{eq:leftBoundVecs}\fl
    \eqalign{
\bra*{l_0}\!=\!\!
    \begin{bmatrix}
        -\beta  \beta _1 \Lambda_{\mathrm{L}} (\alpha -\beta  \Lambda_{\mathrm{L}}) \left(\alpha +\beta _1-\Lambda_{\mathrm{L}}\right) \\
        \beta _1 \Lambda_{\mathrm{L}}^2 (\Lambda_{\mathrm{L}}-\alpha ) \left(\alpha  (\alpha +\beta -\Lambda_{\mathrm{L}}-1)-\beta  \beta
        _1 (\Lambda_{\mathrm{L}}+1)\right) \\
        %-\beta _1 \left(\alpha ^2 \left(\beta -\beta _1 \Lambda_{\mathrm{L}}\right)+\beta _1 \left(\Lambda_{\mathrm{L}}^2 \left(\alpha +\beta ^2\right)+\alpha  \beta \right)+\alpha  \left(\beta _1-2 \beta ^2\right) \Lambda_{\mathrm{L}}+\beta  \Lambda_{\mathrm{L}} \left(\Lambda_{\mathrm{L}}-\beta _1\right)\right)
        -\!\beta _1 \left(
        \alpha ^2 \left(\beta -\beta _1 \Lambda_{\mathrm{L}}\right)\!+\!\beta _1 \left(\Lambda_{\mathrm{L}}^2 \left(\alpha +\beta ^2\right)+\alpha  \beta \right)\!+\!\alpha  \left(\beta _1-2 \beta ^2\right) \Lambda_{\mathrm{L}}
        \!+\!\beta  \Lambda_{\mathrm{L}} \left(\Lambda_{\mathrm{L}}-\beta _1\right)
        \right)
    \end{bmatrix}^T\mkern-14mu,
            \mkern-48mu \\
\bra*{l_1}\!=\!\! 
    \begin{bmatrix}
 -\beta _1^2 \Lambda _{\text{L}} \left(\alpha +\beta _1-\Lambda _{\text{L}}\right) \left(\alpha -\beta  \Lambda
   _{\text{L}}\right) \\
 \beta _1 \Lambda _{\text{L}}^2 \left(\alpha -\Lambda _{\text{L}}\right) \left(\alpha ^2+\alpha  (\beta -2)-\beta ^2+\beta
   -(\alpha +(\beta -2) \beta ) \Lambda _{\text{L}}\right) \\
 \beta _1^2 \left[\beta _1 \left(\alpha +\beta  \Lambda _{\text{L}}^2\right)+\alpha  \left(\alpha +\Lambda _{\text{L}}
   \left[-\alpha -2 \beta +\Lambda _{\text{L}}+1\right]\right)\right]
    \end{bmatrix}^T\mkern-14mu,\\
\bra*{l'_{00}}\!=\!\! 
    \begin{bmatrix}
 \beta  \Lambda _{\text{L}} \left(\alpha +\beta _1-\Lambda _{\text{L}}\right){}^2 \left(\beta  \Lambda _{\text{L}}-\alpha
   \right) \\
 -\beta _1 \Lambda _{\text{L}}^2 \left(\alpha -\Lambda _{\text{L}}\right) \left(\alpha +\beta _1-\Lambda _{\text{L}}\right)
   \left(\alpha -\beta  \Lambda _{\text{L}}\right) \\
 (\beta -1) \beta _1 \Lambda _{\text{L}} \left(\alpha +\beta _1-\Lambda _{\text{L}}\right) \left(\alpha -\beta  \Lambda
   _{\text{L}}\right) \\
    \end{bmatrix}^T\mkern-14mu,\\
\bra*{l'_{01}}\!=\!\! 
    \begin{bmatrix}
 0 \\
 \Lambda _{\text{L}}^2 \left(-\left(\alpha -\Lambda _{\text{L}}\right)\right) \left(\alpha +\beta _1-\Lambda
   _{\text{L}}\right) \left(\alpha -\beta  \Lambda _{\text{L}}\right) \\
 -\left(\alpha +\beta _1\right) \left(\alpha -\Lambda _{\text{L}}\right) \left(\alpha +\beta _1-\Lambda _{\text{L}}\right)
   \left(\alpha -\beta  \Lambda _{\text{L}}\right) \\
    \end{bmatrix}^T\mkern-14mu,\\
\bra*{l'_{10}}\!=\!\! 
    \begin{bmatrix}
 \beta _1 \Lambda _{\text{L}} \left(\alpha +\beta _1-\Lambda _{\text{L}}\right){}^2 \left(\alpha -\beta  \Lambda
   _{\text{L}}\right) \\
 \beta _1 \Lambda _{\text{L}}^2 \left(\alpha -\Lambda _{\text{L}}\right) \left(\alpha +\beta _1-\Lambda _{\text{L}}\right)
   \left(\alpha -\beta  \Lambda _{\text{L}}\right) \\
 -\beta _1^2 \Lambda _{\text{L}} \left(\alpha +\beta _1-\Lambda _{\text{L}}\right) \left(\alpha -\beta  \Lambda
   _{\text{L}}\right)
    \end{bmatrix}^T\mkern-14mu,\\
\bra*{l'_{11}}\!=\!\! 
    \begin{bmatrix}
 0 \\
 0 \\
 (\alpha -1) \left(\alpha +\beta _1-\Lambda _{\text{L}}\right){}^2 \left(\alpha -\beta  \Lambda _{\text{L}}\right)
    \end{bmatrix}^T\mkern-14mu,
    }
\end{eqnarray}
where $\beta_1=\beta-1$.

\section{Explicit form of the representation for the large deviation MPA}
\label{app:Reps}
The coefficients in \eqref{DevMat} are given by the following explicit two-parameter $(\rho,\kappa)$ form ($\rho$ and $\kappa$ play the role of spectral parameters, see Ref.~\cite{buca2019exact} for details),
\begin{eqnarray}\fl
    \begin{aligned}
        &w_1^{(2k)} = \kappa b^{(k)} \frac{t_{10}^{(2k)}}{u_{10}^{(2k)}},&
        &w_2^{(2k)} = \kappa b^{(k)}, &
        &w_3^{(2k)} = \rho c^{(k)},\\
        &w_1^{(2k+1)} = \rho c^{(k)} \frac{1}{u_{00}^{(2k)}}, &
        &w_2^{(2k+1)} = \rho c^{(k+1)}, &
        &w_3^{(2k+1)} = \kappa b^{(k)}t_{00}^{(2k)},\\
        &v_1^{(2k)} = \rho c^{(k)}t_{10}^{(2k)}, &
        &v_2^{(2k)} = \frac{\rho c^{(k)}}{\varpi^{(2k)}}, &
        &v_3^{(2k)} = \frac{\kappa b^{(k)}}{u_{10}^{(2k)}},\\
        &v_1^{(2k+1)} = \kappa b^{(k)}t_{00}^{(2k)}u_{10}^{(2k+1)},\quad &
        &v_2^{(2k+1)} = \kappa b^{(k)}t_{00}^{(2k)}t_{10}^{(2k+1)},\quad &
        &v_3^{(2k+1)} = \frac{\rho c^{(k)}}{u_{00}^{(2k)}u_{10}^{(2k+1)}},
    \end{aligned}
\end{eqnarray}
and
\begin{eqnarray}\fl
  \begin{aligned}
  &z_1^{(2k)} = \prod_{l=1}^{2k-1} g_{00}^{(l)} , &
  &z_2^{(2k)} = z_1^{(2k)}w_1^{(2k)}, &
  &z_3^{(2k)} = z_1^{(2k)},\\ 
  &z_1^{(2k+1)} = \prod_{l=1}^{2k} \frac{1}{f_{00}^{(l)}},&
  &z_2^{(2k+1)} = z_1^{(2k+1)}w_1^{(2k+1)},\qquad&
  &z_3^{(2k+1)} = z_1^{(2k+1)},\\
  &\frac{z_4^{(2k)}}{z_1^{(2k)}} =\kappa b^{(k)}\frac{g_{00}^{(2k)}}{g_{10}^{(2k)}}, &
  &\frac{z_5^{(2k)}}{z_1^{(2k)}} =\frac{1}{\varpi^{(2k)}}\frac{g_{01}^{(2k-1)}}{g_{00}^{(2k-1)}}, &
  &\frac{z_6^{(2k)}}{z_1^{(2k)}}=\frac{\rho c^{(k)}}{\varpi^{(2k)}}\frac{g_{01}^{(2k-1)}}{g_{00}^{(2k-1)}},\\
  &\frac{z_4^{(2k+1)}}{z_1^{(2k+1)}}=\frac{\rho c^{(k)}}{\varpi^{(2k)}}\frac{f_{00}^{(2k)}}{f_{01}^{(2k)}},\qquad&
  &\frac{z_5^{(2k+1)}}{z_1^{(2k+1)}}=\frac{f_{10}^{(2k+1)}}{f_{00}^{(2k+1)}},\qquad&
  &\frac{z_6^{(2k+1)}}{z_1^{(2k+1)}}=\kappa b^{(k)}y^{(2k)}\frac{f_{10}^{(2k)}}{f_{11}^{(2k)}},
  \end{aligned}
\end{eqnarray}
where we introduced
\begin{eqnarray}
    \begin{aligned}
        &y^{(i)} = \frac{f_{01}^{(i-1)}f_{11}^{(i)}f_{10}^{(i+1)}}{f_{00}^{(i-1)}f_{00}^{(i)}f_{00}^{(i+1)}},\qquad&
        &\varpi^{(i)} = \frac{g_{01}^{(i-1)}g_{11}^{(i)}g_{10}^{(i+1)}}{g_{00}^{(i-1)}g_{00}^{(i)}g_{00}^{(i+1)}} \\
        &t_{nn'}^{(i)} = \frac{f_{n1}^{(i-1)}f_{1n'}^{(i)}}{f_{n0}^{(i-1)}f_{0n'}^{(i)}}, \qquad &
        &u_{nn'}^{(i)} = \frac{g_{n1}^{(i-1)}g_{1n'}^{(i)}}{g_{n0}^{(i-1)}g_{0n'}^{(i)}}, \\
        &b^{(i)} = \prod_{k=0}^{i-1} y^{(2k)}\varpi^{(2k+1)}, \qquad & 
        &c^{(i)} = \prod_{k=0}^{i-1} \frac{1}{y^{(2k+1)}\varpi^{(2k)}},
    \end{aligned}
\end{eqnarray}
and we implicitly assumed the following initial values for
$f_{n n^{\prime}}^{(l)}$, and $g_{n n^{\prime}}^{(l)}$,
\begin{eqnarray}
    f_{nn'}^{(-1)}=f_{nn'}^{(0)}=1,\qquad
    g_{nn'}^{(-1)}=g_{nn'}^{(0)}=1.
\end{eqnarray}

\section{MPA for generalized Gibbs states}\label{app:MPSGibbs}
\subsection{Asymptotic MPA for finite configurations}\label{appsubsec:thermodynamicEVs}
We first note that
the transfer matrices $T$ and $T^{\prime}$,
\begin{eqnarray}
    \eqalign{
        T=(\W_0+\W_1)(\V_0+\V_1)=
        \begin{bmatrix}
            1+\xi\omega&\omega&\xi\\
            1+\xi&\xi\omega&\xi\\
            1+\omega&\omega&\xi\omega
        \end{bmatrix},\\
        T^{\prime}=(\V_0+\V_1)(\W_0+\W_1)=
        \left.T^{\phantom{\prime}}\right|_{\xi\leftrightarrow\omega},}
\end{eqnarray}
have an isolated leading eigenvalue $\lambda$, which corresponds to the \emph{largest}
solution of the following cubic equation
\begin{eqnarray}
    \lambda^3 - (1+3\xi\omega) \lambda^2- (\xi+\omega+\xi\omega(1-3\xi\omega))\lambda
    -\xi\omega(1-\xi\omega)^2=0.
\end{eqnarray}
The corresponding left and right leading eigenvectors can be expressed as,
\begin{eqnarray}\fl
    \ket{l} =
    \begin{bmatrix}
        \omega\Big( (\lambda-\xi\omega)^2 - \xi\omega\Big)\\
        \omega^2\Big(\lambda-\xi\omega+\xi\Big)\\
        \xi\omega\Big(\lambda-\xi\omega+\omega\Big)
    \end{bmatrix},\quad
    \ket{r} =
    \begin{bmatrix}
        \omega\Big( \lambda- \xi\omega+\xi\Big)\\
        (\lambda-\xi\omega)^2-\lambda-\xi\\
        \omega\Big(\lambda-\xi\omega+\omega\Big)
    \end{bmatrix},\qquad
    \eqalign{
        \ket{l^{\prime}}=\left.\ket{l}\right|_{\xi\leftrightarrow \omega},\\
        \ket{r^{\prime}}=\left.\ket{r}\right|_{\xi\leftrightarrow \omega},}
\end{eqnarray}
where 
$\bra{l^{(\prime)}} T^{(\prime)} = \lambda\bra{l^{(\prime)}}$,
$T^{(\prime)}\ket{r^{(\prime)}} = \lambda\ket{r^{(\prime)}}$.
Since the leading eigenvalue $\lambda$ is isolated, the thermodynamic
expectation values~\eqref{eq:asymptPropEven} can be conveniently expressed
as finite products with the leading eigenvectors,
\begin{eqnarray}
    \eqalign{
        p\Big(\!\begin{tikzpicture}
            [scale=0.3,baseline={([yshift=-0.6ex]current bounding box.center)}]
            \textrectangle{0}{-0.5}{$s_1$};
            \textrectangle{1}{0.5}{$s_2$};
            \textrectangle{2}{-0.5}{$\cdots$};
            \textrectangle{3}{0.5}{};
            \textrectangle{4}{-0.5}{};
            \textrectangle{5}{0.5}{};
            \textrectangle{6}{-0.5}{$\cdots$};
            \textrectangle{7}{0.5}{$s_{2k}$};
        \end{tikzpicture}\!\Big)=
        \frac{\mel{l}{\W_{s_1}\V_{s_2}\cdots \V_{s_{2k}}}{r}}{\lambda^k \braket{l}{r}},\\
        p\Big(\!\begin{tikzpicture}
            [scale=0.3,baseline={([yshift=-0.6ex]current bounding box.center)}]
            \textrectangle{0}{0.5}{$s_1$};
            \textrectangle{1}{-0.5}{$s_2$};
            \textrectangle{2}{0.5}{$\cdots$};
            \textrectangle{3}{-0.5}{};
            \textrectangle{4}{0.5}{};
            \textrectangle{5}{-0.5}{};
            \textrectangle{6}{0.5}{$\cdots$};
            \textrectangle{7}{-0.5}{$s_{2k}$};
        \end{tikzpicture}\!\Big)=
        \frac{\mel{l^{\prime}}{\V_{s_1}\W_{s_2}\cdots \V_{s_{2k}}}{r^{\prime}}}
        {\lambda^k \braket{l^{\prime}}{r^{\prime}}}.}
\end{eqnarray}
The probabilities of odd configurations~\eqref{eq:asymptPropOdd} can be expressed analogously,
\begin{eqnarray}
    \eqalign{
        p\Big(\!\begin{tikzpicture}
            [scale=0.3,baseline={([yshift=-0.6ex]current bounding box.center)}]
            \textrectangle{0}{-0.5}{$s_1$};
            \textrectangle{1}{0.5}{$s_2$};
            \textrectangle{2}{-0.5}{$\cdots$};
            \textrectangle{3}{0.5}{};
            \textrectangle{4}{-0.5}{};
            \textrectangle{5}{0.5}{$\cdots$};
            \textrectangle{6}{-0.5}{\scalebox{0.8}{$s_{2k-1}$}};
        \end{tikzpicture}\!\Big)=
        \frac{\mel{l}{\W_{s_1}\V_{s_2}\cdots \V_{s_{2k-1}}}{r^{\prime}}}
        {\lambda^{k-1}\mel{l}{\W_0+\W_1}{r^{\prime}}}
        ,\\
        p\Big(\!\begin{tikzpicture}
            [scale=0.3,baseline={([yshift=-0.6ex]current bounding box.center)}]
            \textrectangle{0}{0.5}{$s_1$};
            \textrectangle{1}{-0.5}{$s_2$};
            \textrectangle{2}{0.5}{$\cdots$};
            \textrectangle{3}{-0.5}{};
            \textrectangle{4}{0.5}{};
            \textrectangle{5}{-0.5}{$\cdots$};
            \textrectangle{6}{0.5}{\scalebox{0.8}{$s_{2k-1}$}};
        \end{tikzpicture}\!\Big)=
        \frac{\mel{l^{\prime}}{\V_{s_1}\W_{s_2}\cdots \V_{s_{2k-1}}}{r}}
        {\lambda^{k-1}\mel{l^{\prime}}{\V_0+\V_1}{r}}
        ,
        }
\end{eqnarray}
where we use the fact that auxiliary space vectors $(W_0+W_1)\ket{r}$ and $\ket{r^{\prime}}$
are linearly dependent, i.e.\ there exists an $\alpha\in\mathbb{R}$ so that 
$(W_0+W_1)\ket{r}=\alpha\ket{r^{\prime}}$ (and a similar relation holds for swapped parameters).

\subsection{Partition sum}\label{appsubsec:partitionSum}
The MPA can be used to obtain the configurational entropy of the state with fixed numbers of 
left and right-movers. We start by noting that the partition sum $Z_{n}$ is expressed in terms
of powers of the transfer matrix $T=(\W_0+\W_1)(\V_0+\V_1)$, $Z_{n}=\tr T^{n/2}$, which can be
obtained recursively by introducing the following parametrisation of $T^k$,
\begin{eqnarray}
    T^{k} =
    \begin{bmatrix}
        a_k & \omega b_k & \xi c_k \\
        b_k+\xi c_k & a_k-b_k & \xi b_k \\
        \omega b_k+c_k & \omega c_k & a_k-c_k
    \end{bmatrix}.
\end{eqnarray}
By requiring $T^k\cdot T = T^{k+1}$
we obtain a set of recursive relations that $a_k$, $b_k$ and $c_k$ have to satisfy,
\begin{eqnarray}\label{eq:recurrencePartitionSum}
    \eqalign{
        a_{k+1}=(1+\xi\omega)a_k+(1+\xi)\omega b_k+(1+\omega)\xi c_k,\\
        b_{k+1}=a_k +\xi\omega b_k + \xi c_k,\\
        c_{k+1}=a_k +\omega b_k + \xi\omega c_k.
    }
\end{eqnarray}
From here we straightforwardly realize that $Z_{n}$ can be expressed only in terms 
of $a_{n/2}$,
\begin{eqnarray}
    Z_{n}=\tr T^{n/2}=3 a_{n/2}-(b_{n/2}+c_{n/2})=2 a_{n/2}+(\xi\omega-1) a_{n/2-1},
\end{eqnarray}
and therefore $Z_{2k}$ has to satisfy the same recurrence relation in $k$ as
$a_k$, which is obtained by reducing the set of identities of order
one into a single higher-order relation,
\begin{eqnarray}\label{eq:relationZ}
    \fl
    Z_{2k+2} = 
    (1+3\xi\omega)Z_{2k}
    + \big(\xi+\omega+\xi\omega(1-3\xi\omega)\big) Z_{2(k-1)}
    + \xi\omega(\xi\omega-1)^2 Z_{2(k-2)}.
\end{eqnarray}
Now we express $Z_{n}$ in terms of powers of $\xi$ and $\omega$,
\begin{eqnarray}
    Z_{n} = \sum_{x,y} \alpha_{n/2}(x,y) \xi^x \omega^y,
\end{eqnarray}
and recast the requirement~\eqref{eq:relationZ} in terms of identities that have
to be satisfied for coefficients $\alpha_{n/2}(x,y)$. It can be then verified that the
solution takes the form
\begin{eqnarray}
    \alpha_{k}(x,y) =
    \frac{k(k+x+y)}{(k+x-y)(k-x+y)}
    \binom{k+x-y}{y}
    \binom{k-x+y}{x},
\end{eqnarray}
where we implicitly assume $\alpha_k(x,y)=0$ for $k\le \pm(x-y)$. Asymptotically,
the combinatorial term reduces to the configurational entropy introduced
in~\eqref{eq:YYent},
\begin{eqnarray}
    \lim_{n\to\infty}
    \Big(\alpha_{n/2}\left( \tfrac{n}{2} n_{+} , \tfrac{n}{2} n_{-}\right)
    - \mathrm{e}^{\frac{n}{2}s[n_{+},n_{-}]}\Big)=0.
\end{eqnarray}
\end{document}